\documentclass{elsarticle}
\usepackage{latexsym}
\usepackage{float}
\usepackage[latin2]{inputenc}
\usepackage{graphicx}
\usepackage{ae}
\usepackage{aecompl}
\usepackage{setspace}
\usepackage{amsmath}
\usepackage{amssymb}
\usepackage{amsthm}
\usepackage{pdfpages}
\newtheorem{theorem}{Theorem}
\newtheorem{lemma}[theorem]{Lemma}
\newtheorem{definition}{Definition}
\newtheorem{corollary}[theorem]{Corollary}

\newtheorem{remark}[theorem]{Remark}

\usepackage[hidelinks]{hyperref}

\usepackage{verbatim}
\usepackage{fancyvrb}

\makeatletter
\def\ps@pprintTitle{%
	\let\@oddhead\@empty
	\let\@evenhead\@empty
	\def\@oddfoot{\centerline{\thepage}}%
	\let\@evenfoot\@oddfoot}
\makeatother

\begin{document}

\title{Robust Circuitry-Based Scores of Structural Importance of Human Brain Areas}
	
\author[p]{Dániel Hegedűs}
\ead{hegedus@pitgroup.org}
\author[p,u]{Vince Grolmusz\corref{cor1}}
\ead{grolmusz@pitgroup.org}
\cortext[cor1]{Corresponding author}
\address[p]{PIT Bioinformatics Group, Eötvös University, H-1117 Budapest, Hungary}
\address[u]{Uratim Ltd., H-1118 Budapest, Hungary}

\date{}

\begin{abstract}
	We consider the 1015-vertex human consensus connectome computed from the diffusion MRI data of 1064 subjects. We define seven different orders on these 1015 graph vertices, where the orders depend on parameters derived from the brain circuitry, that is, from the properties of the edges (or connections) incident to the vertices ordered. We order the vertices according to their degree, the sum, the maximum, and the average of the fiber counts on the incident edges, and the sum, the maximum and the average length of the fibers in the incident edges. We analyze the similarities of these seven orders by the Spearman correlation coefficient and by their inversion numbers and have found that all of these seven orders have great similarities. In other words, if we interpret the orders as scoring of the importance of the vertices in the consensus connectome, then the scores of the vertices will be similar in all seven orderings. That is, important vertices of the human connectome typically have many neighbors, connected with long and thick axonal fibers (where thickness is measured by fiber numbers), and their incident edges have high maximum and average values of length and fiber-number parameters, too. Therefore, these parameters may yield robust ways of deciding which vertices are more important in the anatomy of our brain circuitry than the others.

\end{abstract}

\maketitle

\bigskip
\noindent Running head: Circuitry-Based Scores of Human Brain Areas

\medskip
	
\section*{Introduction} 
Identifying the most important nodes in large networks solely from their graph-theoretical properties was an important problem in the late 1990s, applied in scoring the web search engine hits. The most well-known solutions, the PageRank of Google \cite{Brin98theanatomy} and the HITS algorithm of Kleinberg \cite{Kleinberg1998}, fundamentally influenced the related areas. 

Both the PageRank and the HITS algorithms score the nodes of a directed, unweighted graph, originally corresponded to the graph of the World Wide Web, but later, those algorithms were successfully applied for directed and undirected biological, social and chemical graphs, among other applications \cite{Ivan2011,Grolmusz2012,Banky2013, Grolmusz2015}.

Since all human activities are governed by the cooperation of the cells in our brain, the study of the connections of these cells has specific interest. Unfortunately, the connections of the 80 billion neurons of the human brain are not mapped yet, and will not be mapped in the foreseeable future: to date, the only adult organism with completely mapped connections between its neurons (also called the connectome or braingraph) is the nematode {\it C. elegans}, having only 302 neurons \cite{White1986}. Recently, after many years of concentrated efforts, the neuronal-level connectome of a part of the brain of the adult fruit fly {\it Drosophila melanogaster}, its central brain, is mapped and published \cite{Xu2020}. Out of the 100,000 neurons of the fruit fly, the central brain contains around 25,000 neurons. The whole {\it Drosophila melanogaster} connectome is not published yet. 

Instead of the neuronal-level connections, the imaging methods are capable today of mapping the human connectome on a much coarser scale than the level of the neurons. Due to the technical developments of magnetic resonance imaging (MRI) in the last fifteen years \cite{Hagmann2008,Hagmann2012}, today we can map the macroscopic connections between 1000 anatomically identified brain areas. These developments have opened up a new area of brain science called ``connectomics'', which examines the connections between the brain areas, and, instead of comparing the volumes of brain areas between healthy or diseased, old and young, male or female subjects, as in hundreds of previous cerebral volumetric studies, it concentrates to a more central question: the connections between those areas. 

Our research group has studied the mathematical properties of human connectomes by applying strict graph-theoretical methods, terms, and approaches. We have used the public release imaging data sets of the Human Connectome Project \cite{McNab2013}, and prepared publicly available braingraphs from the imaging data, downloadable at the address \url{https://braingraph.org} in five different resolutions \cite{Kerepesi2016b,Szalkai2016d,Varga2020,Keresztes2020}. The vertices of the braingraphs correspond to the anatomically identified areas of the cortical and sub-cortical gray matter, and two of the vertices are connected by an edge if the tractography phase \cite{Daducci2012, Tournier2012} of the processing identified axonal fibers between the areas, mapped to the vertices. 

Using the exact methods and deep algorithms and approaches of graph theory, we have discovered numerous connectomical properties, related to the human sex differences \cite{Szalkai2015b,Szalkai2015c,Szalkai2016a,Keresztes2019,Keresztes2022a}, early brain development \cite{Kerepesi2015b,Szalkai2016e,Kerepesi2016,Szalkai2016d}, different lobal structures and organizations \cite{Kerepesi2015a, Szalkai2017c,Szalkai2016c}, and frequent edge sets in the whole brain or only those which are adjacent to the hippocampus \cite{Fellner2017,Fellner2018,Fellner2019,Fellner2019a}.

In the present contribution, we consider an averaged consensus connectome, computed from the imaging data of 1064 subjects, and we order the vertices of the consensus braingraph, intended to catch their order of ``importance'' and compare the order of vertices in those lists. Our main result is that orders generated from the 
\begin{itemize}
\item degree of the nodes, 
\item the sum of the numbers, 
\item the maximum number,
\item and the average number of fibers in the incident edges,
\item the sum of the fiber lengths,
\item the maximum fiber length,
\item the average fiber length in the incident edges
\end{itemize}

are similar to one another, their Spearman's rank correlation is high, and their inversion numbers are low.

This result means that ordering by any of the seven parameters above produces similar orders of the nodes, where the ``similar'' word is explained in detail later in this work.

In other words, the result can be interpreted that the most important nodes in our braingraph statistically have numerous and long incident axonal fibers, with high maximum and averaged values either for the length or for the fiber numbers. That is, if a node is in front of others in one of the seven parameters above, then, typically, it will have high values in the remaining six parameters, too. Therefore, all of these seven orders are robust in comparison with the other six ones.

In what follows, we describe precisely our methods and results.

\section*{Methods}

\subsection*{Graph construction} The data source of the present work is the 1200-subject public release of the Human Connectome Project \cite{McNab2013}. The 3 Tesla diffusion magnetic resonance imaging data were processed with the help of the Connectome Mapper Tool Kit \cite{Daducci2012}. 

We have computed five different graphs for each subject with 83, 129, 234, 463, and 1015 nodes, where each node corresponded to an anatomic area of the cortical- and sub-cortical gray matter. The parcellation tool FreeSurfer was applied here \cite{Fischl2012,Desikan2006,Tournier2012}. 

The details of the workflow we followed are described in \cite{Varga2020}. Concisely, the axonal fibers were mapped by the MRtrix 0.2 tractography software, and repeated ten times for each subject. We connected two graph vertices, which corresponded to two gray matter areas, by an edge if, in all the 10 runs, axonal fibers were found running between the two areas. In this case, the maximum and the minimum number of fibers were deleted, and the remaining eight integer values were averaged and assigned to the edge as the fiber number weight. The length of the edge is also determined as the average length of the defining fibers. Consequently, all graph edges carry a positive weight (meaning the average of 8 fiber numbers) and a positive length (in millimeters). 

Next, we constructed one single consensus graph on 1015 vertices from the 1064 individual graphs as follows. We have averaged the weight and the length for each edge, but we have followed different strategies. For averaging the weight, we added up the edge weight in all subject's graph and divided the sum by 1064; if an edge was not present in a subject, then we counted it as an edge with (an artificial) weight of 0. In the case of computing the average length, we counted the existing edges and divided the length-sum by this integer (for vertex pairs, which do not appear as an edge, 0 lengths were assigned).

Consequently, if $\#\{i,j\}$ denotes the number of appearance of edge $\{i,j\}$, and  $s_{i,j,k}$ and $h_{i,j,k}$ denote in subject $k$ the weight and the length of edge $\{i,j\}$, respectively, then
	
	\begin{equation}
		s_{i,j}={1\over 1064}\sum_{k=1}^{1064} s_{i,j,k}
	\end{equation}
	
	\begin{equation}
		h_{i,j}={1\over{\#(i,j)}}\sum_{k=1}^{1064} h_{i,j,k}
	\end{equation}

We note that our earlier works \cite{Szalkai2015a,Szalkai2016} also describe parameterizable consensus graphs by user-selectable parameters at the website of the Budapest Reference Connectome \url{https://pitgroup.org/connectome}. In contrast, the dataset of the present contribution is a static graph.

\subsection*{Ordering the nodes}

Here we consider seven different orders on the set of the 1015 vertices of our consensus graph, with abbreviations:

\begin{itemize}
\item by the degree of the nodes (Degree); 
\item by the sum of the number of fibers in the incident edges (SUM-weight);
\item by the maximum number of fibers in the incident edges (MAX-weight);
\item by the average number of fibers in the incident edges (AVG-weight);
\item by the sum of the fiber lengths in the incident edges (SUM-length);
\item by the maximum fiber length in the incident edges (MAX-length);
\item by the average fiber length in the incident edges (AVG-length).
\end{itemize}

The orders are defined by the decreasing values of these parameters.

The Degree describes the number of the incident edges, i.e., more important nodes are believed to be connected to many other nodes, so, consequently, their degree should be larger than the degree of the less important nodes.

SUM-weight is the weighted version of the Degree parameter. Here the fiber numbers of the incident edges are added up. Edges with higher fiber number or weight may connect more small gray matter areas than those with less weight. Consequently, we think that the SUM-weight parameter is more relevant in deciding the importance of a node than just the Degree.

MAX-weight -- in a certain sense -- is a simplification of the SUM-weight parameter. It describes only the weight of the largest-weight incident edge instead of adding up all the weights of the incident edges. Theoretically, it may happen that the ordering according to the MAX-weight differs a lot from the order defined by SUM-weight if, in many vertices, the incident edges have a small number of large a large number of small weights. It turns out later that in the case of our graph, it is not true; the orders are similar.

AVG-weight: Clearly, for each vertex, the Degree times AVG-weight is the SUM-weight. Therefore, the AVG-weight-based vertex-order may differ strongly from both the Degree-order and the SUM-weight order. As we show, the AVG-weight based order is also similar to the Degree and to the SUM-weight order.

SUM-length is the length-weighted version of the Degree. The Degree and the SUM-length values may differ a lot if a node is adjacent to many other vertices by short edges or few other nodes but with very long edges. If an important node usually has numerous and long incident edges, then the orders by Degree and SUM-length would not differ a lot. We show that in the case of our consensus braingraph, this is the situation. 

MAX-length can be large, while SUM-length is small, so the order according to these parameters can be different in numerous positions. We show that this is not the case in our graph.

AVG-length times the Degree is the SUM-length for each vertex. Therefore, the AVG-length-based order can strongly differ from both the Degree and the SUM-length based order. We show the opposite for the case of the consensus braingraph.

The seven orders are explicitly given in the Appendix.

In the following subsections, we introduce two tools for the analysis of the similarity of these orders: the Spearman correlation and the inversion numbers.

\subsection*{Spearman's rank correlation coefficient}

The Spearman $\varrho$ coefficient \cite{Spearman} is an ideal tool for comparing different orders on the same base set. In our case, the base set is the set of vertices, and the seven different orders are defined by the seven parameters Degree, SUM-weight, MAX-weight, AVG-weight, SUM-length, MAX-length, and AVG-length. 

The $\varrho$ coefficient gives information about the correlation of two attributes, using the two orderings by the two attributes of the elements. Two indices are associated with every element, telling what its index is in the given ordering. There is a simple equation that calculates the coefficient if the indices are unique, meaning that no two attributes are the same. (Luckily, this is true for the consensus braingraph.) If we calculate two attributes of $n$ elements and $d_i$ is the difference of the $i$-th element's two indices, then:
	\begin{equation}
		\varrho=1-\frac{6\cdot\sum_{i=1}^{n} d_{i}^2}{n^3-n}
	\end{equation}

The value of Spearman's correlation coefficient $\varrho$ satisfies $-1\leq \varrho \leq 1$, where $\varrho=1$ means the perfect correlation and $\varrho=-1$ means the perfect opposition. 

\begin{remark} The bounds for $\varrho$ can be proven by using the cubic formula for the sum of the first $n$ square numbers: $\frac{n(n+1)(2n+1)}{6}$. In the case of perfect opposition, we should calculate the sum of the first $\frac{n}{2}$ odd square numbers. This can easily be acquired from the sum of (not counted) even square numbers, as this is exactly four times the sum of the first that many square numbers. 
\end{remark}

\begin{remark}
The closer the coefficient $\varrho$ is to $0$, the less we can say about the predicted correlation. As $n$ grows, the coefficients with smaller absolute values can also be significant. So the $p$-value is not only defined by $\varrho$, but it also depends on $n$. The acquired $p$-value is the probability of the correlation being that extreme under the assumption that the null hypothesis is true.
\end{remark}
	
\subsection*{Inversion numbers}
	
	\begin{definition}
		In two given permutations of $n$ elements, two different elements are \textit{in inversion with each other} if their order is opposite in the permutations.
	\end{definition}
	
	\begin{definition}
		Two permutations' \textit{inversion number} is the number of (not ordered) pairs of elements which are in inversion. An element's \textit{inversion} is the number of elements with which it is in inversion.
	\end{definition}
	
	\begin{lemma}
		The expected value of the inversion number of two permutations of length $n$ is $\frac{n(n-1)}{4}$.
	\end{lemma}
	
	\begin{proof}
		Look at an arbitrary unordered $\{i,j\}$ pair of elements. Because of symmetric reasoning, the expected value of the contribution of this pair to the inversion number is $\frac{1}{2}$. (As every permutation has a bijective pair which differs only in the $(i,j)$ transposition. Transposition is the function that only swaps two elements in a permutation.) Since $\binom{n}{2}=\frac{n(n-1)}{2}$ unordered pair of elements exist in the set of $n$ elements, by using the linearity of the expectation, we get that the expected value of the inversion number is the desired $\frac{1}{2}\cdot\frac{n(n-1)}{2}=\frac{n(n-1)}{4}$.
	\end{proof}
	
	\begin{corollary}
		The expected value of the inversion of any element is $\frac{n-1}{2}$.
	\end{corollary}
	
	\begin{proof}
		The linearity of the expected value can be used again, but now for the result of Lemma 1. As there are $n$ elements and each inversion is counted twice (for both elements of the pair), the expected value by elements is $\frac{n(n-1)}{4}\cdot\frac{2}{n}=\frac{n-1}{2}$.
	\end{proof}
	
	\noindent \textit{Alternative proof.} Make a bijection between the permutations: the pair of a permutation is the opposite permutation. In every pair of permutations, every unordered pair of elements $\{i,j\}$ has each of its two orderings in exactly one of the members of the permutation pair. So for every element $i$, there are $n-1$ different $j$ elements, and the expected value of their inversion one-by-one is $\frac{1}{2}$. So the expected value of the inversion of the arbitrary element $i$ is $\frac{n-1}{2}$. \hspace*{\fill} $\square$

\section*{Discussion and Results}

\subsection*{Spearman-correlations of different orderings}
	
	\begin{center}
		
		\textbf{Correlations with the degree}
		
		\small
		\begin{tabular}{ |c|c|c|c|c|c|c|c| } 
			\hline
			&  \multicolumn{6}{c|}{Degree} \\
			\hline
			& SUM-weight & MAX-weight & AVG-weight & SUM-length & MAX-length & AVG-length \\
			\hline
			$\varrho$ & $0.88$ & $0.84$ & $0.79$ & $0.98$ & $0.51$ & $0.76$ \\
			%\hline
			$p$ & $0.0$ & $10^{-278}$ & $10^{-220}$ & $0.0$ & $10^{-68}$ & $10^{-194}$ \\
			\hline
		\end{tabular}
		
		\smallskip
		
		\textit{Table 1} Spearman-correlations between the Degree-based and the six other orders. The first row contains the $\varrho$ coefficients, and the second the significance-characterizing p values. We note that the weakest correlation in the case of MAX-length is still very far from 0, and its p-value is very small. 
		\normalsize
		
		\bigskip
		
		\textbf{Correlations with SUM-weight and MAX-weight}
		
		\small
		\begin{tabular}{ |c|c|c|c|c|c|c|c| }
			\hline
			&  \multicolumn{3}{c|}{SUM-weight} & \multicolumn{3}{c|}{MAX-weight}\\
			\hline
			& SUM-length & MAX-length & AVG-length & SUM-length & MAX-length & AVG-length \\
			\hline
			$\varrho$ & $0.86$ & $0.42$ & $0.63$ & $0.83$ & $0.41$ & $0.62$ \\
			%\hline
			$p$ & $10^{-302}$ & $10^{-45}$ & $10^{-112}$ & $10^{-262}$ & $10^{-43}$ & $10^{-110}$ \\
			\hline
		\end{tabular}
		\smallskip
		
		\textit{Table 2} Spearman-correlations between the SUM-weight and MAX-weight orders and the three length-based orders. The first row contains the $\varrho$ coefficients, the second the significance-characterizing p-values. We note that the weakest correlations in the case of MAX-length are still very far from 0, and their p-values are very small. 
		\bigskip
		
		\normalsize
		\textbf{Correlations with AVG-weight}
		
		\small
		\begin{tabular}{ |c|c|c|c| }
			\hline
			&  \multicolumn{3}{c|}{AVG-weight} \\
			\hline
			& SUM-length & MAX-length & AVG-length \\
			\hline
			$\varrho$ & $0.77$ & $0.37$ & $0.54$ \\
			%\hline
			$p$ & $10^{-199}$ & $10^{-33}$ & $10^{-79}$ \\
			\hline
		\end{tabular}
		
		\smallskip
		
		\textit{Table 3} Spearman-correlations between the AVG-weight order and the three length-based orders. The first row contains the $\varrho$ coefficients, the second the significance-characterizing p values. We note that the weakest correlation in the case of MAX-length is still very far from 0, and its p-value is very small. 
		\normalsize
		
		\bigskip
		
		\textbf{Correlations between weight-weight and length-length orders}
		
		\small
		\begin{tabular}{ |c|c|c|c|c|c|c|c| } 
			\hline
			&  \multicolumn{3}{c|}{weight based orders} & \multicolumn{3}{c|}{length based orders}\\
			\hline
			& SUM-MAX& SUM-AVG & MAX-AVG & SUM-MAX & SUM-AVG& MAX-AVG \\
			\hline
			$\varrho$ & $0.97$ & $0.98$ & $0.96$ & $0.58$ & $0.86$ & $0.70$ \\
			%\hline
			$p$ & $0.0$ & $0.0$ & $0.0$ & $10^{-91}$ & $10^{-296}$ & $10^{-148}$ \\
			\hline
			
		\end{tabular}
		
		\smallskip
		
		\textit{Table 4} Spearman correlations and p-values between the weight-weight based and the length-length-based orders. All the correlations are high, and the lowest value belongs to the SUM-length vs. MAX-length correlations.
		\normalsize
		
	\end{center}
	
	\smallskip

\subsection*{A simple control}

For a simple control, we have computed the Spearman correlation between two obviously unrelated orders. Namely, we have taken the ordinal numbers of the vertices assigned by the parcellation software and the AVG-weight-defined order. The ordinal numbers are assigned in the way that around the first 500 numbered vertices are situated in the left and the second 500 vertices in the right hemisphere of the brain, in the same order.  For these two orders $\varrho=0.01$ and $p=0.65$, therefore, our results in Tables 1-4 present a biological rule.  

\subsection*{Analysis of the order-similarity by inversion numbers}

In the Methods section, we have defined the inversion numbers and listed some of their fundamental properties. Here we present a graphical evaluation of the inversion numbers between the pairs of the seven orders studied by the Spearman correlation in the previous section.

\begin{figure}[H]
	\centering
	\includegraphics[scale=0.8]{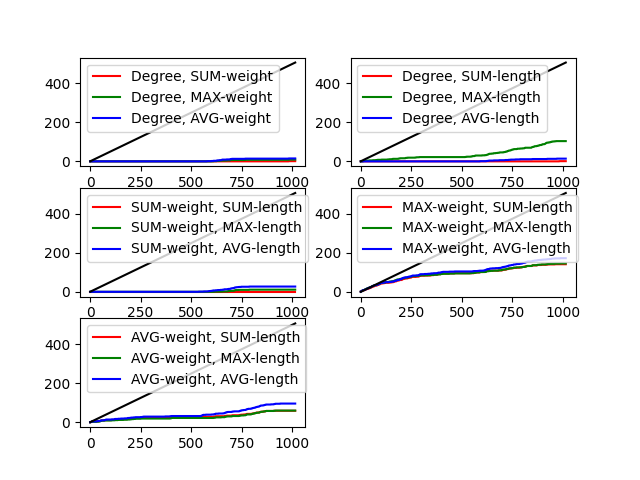}
	\caption{The number of vertices with high inversion-numbers according to distinct parameter-pairs. Point $n$ on axis $x$ correspond to the most important $n$ vertices in one of the pair-defined orders, while the height (i.e. the $y$ coordinate) of the point correspond to the number with higher-than-expectation inversion number between the most important $n$ pairs of orderings. On each panel the black line shows the $n/2$ expectation. It is easy to see on all panels that the higher-than-expectation inversion numbers appear in very few pairs, almost independently from the examined pairs of orderings.}

	\end{figure}

\section*{Conclusions}
We have analyzed the order of vertex-importance in the anatomically labeled consensus graph of the human brain, defined by circuitry-based parameters of the vertices: the degree (Degree), and the following parameters, computed on the incident edges for the vertex: Sum of fiber counts (SUM-weight), Max of fiber counts (MAX-weight), Average of fiber counts (AVG-weight), Sum of fiber lengths (SUM-length), Max of fiber lengths (MAX-length) and the Average of fiber lengths (AVG-length). For the analysis, we have used the Spearman correlation coefficient and the inversion numbers between the orders. We have found that the seven orders f vertex importance, defined by these seven circuitry-based parameters of the vertices, have a great similarity: i.e., the most important vertices - statistically - have many neighbors, connected with long and numerous fibers. We also have shown that orders defined by the maximum weight or length of the incident edges or the orders defined by the average weight or length of the incident edges do not differ very much from the orders defined by the sum of these parameters.

The results show the robustness of the orders by these seven parameters and also shows that vertex-importance in the human brain can be characterized by numerous parameters, but the list of the important vertices (or anatomical brain areas) will not be changed much.

%\section*{References}

%\bibliography{v:/vince/CIKKEK/medl}

\begin{thebibliography}{37}
\providecommand{\natexlab}[1]{#1}
\providecommand{\url}[1]{\texttt{#1}}
\expandafter\ifx\csname urlstyle\endcsname\relax
  \providecommand{\doi}[1]{doi: #1}\else
  \providecommand{\doi}{doi: \begingroup \urlstyle{rm}\Url}\fi

\bibitem[Brin and Page(1998)]{Brin98theanatomy}
Sergey Brin and Lawrence Page.
\newblock The anatomy of a large-scale hypertextual web search engine.
\newblock \emph{Computer Networks and ISDN Systems}, 30:\penalty0 107--117,
  1998.

\bibitem[Kleinberg(1998)]{Kleinberg1998}
Jon~M. Kleinberg.
\newblock Authoritative sources in a hyperlinked environment.
\newblock In \emph{Proceedings of the Ninth Annual {ACM}-{SIAM} Symposium on
  Discrete Algorithms}, pages 668--677, San Francisco, California, 25--27
  January 1998. ACM Press.

\bibitem[Iv{\'a}n and Grolmusz(2011)]{Ivan2011}
G{\'a}bor Iv{\'a}n and Vince Grolmusz.
\newblock When the web meets the cell: using personalized pagerank for
  analyzing protein interaction networks.
\newblock \emph{Bioinformatics}, 27\penalty0 (3):\penalty0 405--407, 2011.

\bibitem[Grolmusz et~al.(2012)Grolmusz, Ivan, Banky, and
  Szerencsi]{Grolmusz2012}
Vince Grolmusz, Gabor Ivan, Daniel Banky, and Balazs Szerencsi.
\newblock How to find non hub important nodes in protein networks?
\newblock \emph{Biophysical Journal}, 102\penalty0 (3):\penalty0 184a, 2012.

\bibitem[B{\'a}nky et~al.(2013)B{\'a}nky, Iv{\'a}n, and Grolmusz]{Banky2013}
D{\'a}niel B{\'a}nky, G{\'a}bor Iv{\'a}n, and Vince Grolmusz.
\newblock Equal opportunity for low-degree network nodes: a pagerank-based
  method for protein target identification in metabolic graphs.
\newblock \emph{PLoS One}, 8\penalty0 (1):\penalty0 e54204, 2013.

\bibitem[Grolmusz(2015)]{Grolmusz2015}
Vince Grolmusz.
\newblock A note on the pagerank of undirected graphs.
\newblock \emph{Information Processing Letters}, 115\penalty0 (6):\penalty0
  633--634, 2015.

\bibitem[White et~al.(1986)White, Southgate, Thomson, and Brenner]{White1986}
JG~White, E~Southgate, JN~Thomson, and S~Brenner.
\newblock The structure of the nervous system of the nematode {Caenorhabditis}
  elegans: the mind of a worm.
\newblock \emph{Phil. Trans. R. Soc. Lond}, 314:\penalty0 1--340, 1986.

\bibitem[Scheffer et~al.(2020)Scheffer, Xu, Januszewski, and et~al.]{Xu2020}
Lousis Scheffer, Shan Xu, Michal Januszewski, and Zhiyuan~Lu et~al.
\newblock A connectome and analysis of the adult drosophila central brain.
\newblock \emph{eLife}, 9, 2020.
\newblock \doi{10.7554/eLife.57443}.
\newblock URL \url{https://doi.org/10.7554/eLife.57443}.

\bibitem[Hagmann et~al.(2008)Hagmann, Cammoun, Gigandet, Meuli, Honey, Wedeen,
  and Sporns]{Hagmann2008}
Patric Hagmann, Leila Cammoun, Xavier Gigandet, Reto Meuli, Christopher~J.
  Honey, Van~J. Wedeen, and Olaf Sporns.
\newblock Mapping the structural core of human cerebral cortex.
\newblock \emph{PLoS Biol}, 6\penalty0 (7):\penalty0 e159, Jul 2008.
\newblock \doi{10.1371/journal.pbio.0060159}.
\newblock URL \url{http://dx.doi.org/10.1371/journal.pbio.0060159}.

\bibitem[Hagmann et~al.(2012)Hagmann, Grant, and Fair]{Hagmann2012}
Patric Hagmann, Patricia~E. Grant, and Damien~A. Fair.
\newblock {MR} connectomics: a conceptual framework for studying the developing
  brain.
\newblock \emph{Front Syst Neurosci}, 6:\penalty0 43, 2012.
\newblock \doi{10.3389/fnsys.2012.00043}.
\newblock URL \url{http://dx.doi.org/10.3389/fnsys.2012.00043}.

\bibitem[McNab et~al.(2013)McNab, Edlow, Witzel, Huang, Bhat, Heberlein,
  Feiweier, Liu, Keil, Cohen-Adad, Tisdall, Folkerth, Kinney, and
  Wald]{McNab2013}
Jennifer~A. McNab, Brian~L. Edlow, Thomas Witzel, Susie~Y. Huang, Himanshu
  Bhat, Keith Heberlein, Thorsten Feiweier, Kecheng Liu, Boris Keil, Julien
  Cohen-Adad, M~Dylan Tisdall, Rebecca~D. Folkerth, Hannah~C. Kinney, and
  Lawrence~L. Wald.
\newblock The {H}uman {C}onnectome {P}roject and beyond: initial applications
  of 300 m{T}/m gradients.
\newblock \emph{Neuroimage}, 80:\penalty0 234--245, Oct 2013.
\newblock \doi{10.1016/j.neuroimage.2013.05.074}.
\newblock URL \url{http://dx.doi.org/10.1016/j.neuroimage.2013.05.074}.

\bibitem[Kerepesi et~al.(2017)Kerepesi, Szalkai, Varga, and
  Grolmusz]{Kerepesi2016b}
Csaba Kerepesi, Balazs Szalkai, Balint Varga, and Vince Grolmusz.
\newblock The braingraph. org database of high resolution structural
  connectomes and the brain graph tools.
\newblock \emph{Cognitive Neurodynamics}, 11\penalty0 (5):\penalty0 483--486,
  2017.

\bibitem[Szalkai et~al.(2019{\natexlab{a}})Szalkai, Kerepesi, Varga, and
  Grolmusz]{Szalkai2016d}
Balazs Szalkai, Csaba Kerepesi, Balint Varga, and Vince Grolmusz.
\newblock High-resolution directed human connectomes and the consensus
  connectome dynamics.
\newblock \emph{PLoS ONE}, 14\penalty0 (4):\penalty0 e0215473, September
  2019{\natexlab{a}}.
\newblock URL \url{https://doi.org/10.1371/journal.pone.0215473}.

\bibitem[Varga and Grolmusz(2021)]{Varga2020}
Balint Varga and Vince Grolmusz.
\newblock The braingraph.org database with more than 1000 robust human
  structural connectomes in five resolutions.
\newblock \emph{Cognitive Neurodynamics}, 2021.
\newblock URL \url{https://doi.org/10.1007/s11571-021-09670-5}.

\bibitem[Keresztes et~al.()Keresztes, Szogi, Varga, and
  Grolmusz]{Keresztes2020}
Laszlo Keresztes, Evelin Szogi, Balint Varga, and Vince Grolmusz.
\newblock Introducing and applying newtonian blurring: An augmented dataset of
  126,000 human connectomes at braingraph.org.
\newblock \emph{Scientific Reports}.
\newblock \doi{10.1038/s41598-022-06697-4}.
\newblock URL \url{https://www.nature.com/articles/s41598-022-06697-4}.

\bibitem[Daducci et~al.(2012)Daducci, Gerhard, Griffa, Lemkaddem, Cammoun,
  Gigandet, Meuli, Hagmann, and Thiran]{Daducci2012}
Alessandro Daducci, Stephan Gerhard, Alessandra Griffa, Alia Lemkaddem, Leila
  Cammoun, Xavier Gigandet, Reto Meuli, Patric Hagmann, and Jean-Philippe
  Thiran.
\newblock The connectome mapper: an open-source processing pipeline to map
  connectomes with {MRI}.
\newblock \emph{PLoS One}, 7\penalty0 (12):\penalty0 e48121, 2012.
\newblock \doi{10.1371/journal.pone.0048121}.
\newblock URL \url{http://dx.doi.org/10.1371/journal.pone.0048121}.

\bibitem[Tournier et~al.(2012)Tournier, Calamante, Connelly,
  et~al.]{Tournier2012}
J~Tournier, Fernando Calamante, Alan Connelly, et~al.
\newblock Mrtrix: diffusion tractography in crossing fiber regions.
\newblock \emph{International Journal of Imaging Systems and Technology},
  22\penalty0 (1):\penalty0 53--66, 2012.

\bibitem[Szalkai et~al.(2015{\natexlab{a}})Szalkai, Varga, and
  Grolmusz]{Szalkai2015b}
Bal{\'{a}}zs Szalkai, B{\'{a}}lint Varga, and Vince Grolmusz.
\newblock Graph theoretical analysis reveals: Women's brains are better
  connected than men's.
\newblock \emph{PLoS One}, 10\penalty0 (7):\penalty0 e0130045,
  2015{\natexlab{a}}.
\newblock \doi{10.1371/journal.pone.0130045}.
\newblock URL \url{http://dx.doi.org/10.1371/journal.pone.0130045}.

\bibitem[Szalkai et~al.(2018{\natexlab{a}})Szalkai, Varga, and
  Grolmusz]{Szalkai2015c}
Bal{\'a}zs Szalkai, B{\'a}lint Varga, and Vince Grolmusz.
\newblock Brain size bias-compensated graph-theoretical parameters are also
  better in women's connectomes.
\newblock \emph{Brain Imaging and Behavior}, 12\penalty0 (3):\penalty0
  663--673, 2018{\natexlab{a}}.
\newblock \doi{10.1007/s11682-017-9720-0}.
\newblock URL \url{http://dx.doi.org/10.1007/s11682-017-9720-0}.

\bibitem[Szalkai et~al.(2021)Szalkai, Varga, and Grolmusz]{Szalkai2016a}
Bal{\'a}zs Szalkai, B{\'a}lint Varga, and Vince Grolmusz.
\newblock The graph of our mind.
\newblock \emph{Brain Sciences}, 11\penalty0 (3), 2021.
\newblock URL \url{https://doi.org/10.3390/brainsci11030342}.

\bibitem[Keresztes et~al.(2021)Keresztes, Szogi, Varga, and
  Grolmusz]{Keresztes2019}
Laszlo Keresztes, Evelin Szogi, Balint Varga, and Vince Grolmusz.
\newblock Identifying super-feminine, super-masculine and sex-defining
  connections in the human braingraph.
\newblock \emph{Cognitive Neurodynamics}, 15\penalty0 (6):\penalty0 949--959,
  2021.
\newblock URL \url{https://doi.org/10.1007/s11571-021-09687-w}.

\bibitem[Keresztes et~al.(2022)Keresztes, Szogi, Varga, and
  Grolmusz]{Keresztes2022a}
Laszlo Keresztes, Evelin Szogi, Balint Varga, and Vince Grolmusz.
\newblock Discovering sex and age implicator edges in the human connectome.
\newblock \emph{Neuroscience letters}, 791:\penalty0 136913, November 2022.
\newblock ISSN 1872-7972.
\newblock \doi{10.1016/j.neulet.2022.136913}.

\bibitem[Kerepesi et~al.(2016)Kerepesi, Szalkai, Varga, and
  Grolmusz]{Kerepesi2015b}
Csaba Kerepesi, Balazs Szalkai, Balint Varga, and Vince Grolmusz.
\newblock How to direct the edges of the connectomes: Dynamics of the consensus
  connectomes and the development of the connections in the human brain.
\newblock \emph{PLOS One}, 11\penalty0 (6):\penalty0 e0158680, June 2016.
\newblock URL \url{http://dx.doi.org/10.1371/journal.pone.0158680}.

\bibitem[Szalkai et~al.(2017{\natexlab{a}})Szalkai, Varga, and
  Grolmusz]{Szalkai2016e}
Bal{\'a}zs Szalkai, B{\'a}lint Varga, and Vince Grolmusz.
\newblock The robustness and the doubly-preferential attachment simulation of
  the consensus connectome dynamics of the human brain.
\newblock \emph{Scientific Reports}, 7\penalty0 (16118), 2017{\natexlab{a}}.
\newblock \doi{10.1038/s41598-017-16326-0}.

\bibitem[Kerepesi et~al.(2018{\natexlab{a}})Kerepesi, Varga, Szalkai, and
  Grolmusz]{Kerepesi2016}
Csaba Kerepesi, Balint Varga, Balazs Szalkai, and Vince Grolmusz.
\newblock The dorsal striatum and the dynamics of the consensus connectomes in
  the frontal lobe of the human brain.
\newblock \emph{Neuroscience Letters}, 673:\penalty0 51--55, March
  2018{\natexlab{a}}.
\newblock \doi{10.1016/j.neulet.2018.02.052}.

\bibitem[Kerepesi et~al.(2018{\natexlab{b}})Kerepesi, Szalkai, Varga, and
  Grolmusz]{Kerepesi2015a}
Csaba Kerepesi, Bal{\'a}zs Szalkai, B{\'a}lint Varga, and Vince Grolmusz.
\newblock Comparative connectomics: Mapping the inter-individual variability of
  connections within the regions of the human brain.
\newblock \emph{Neuroscience Letters}, 662\penalty0 (1):\penalty0 17--21,
  2018{\natexlab{b}}.
\newblock \doi{10.1016/j.neulet.2017.10.003}.

\bibitem[Szalkai et~al.(2018{\natexlab{b}})Szalkai, Varga, and
  Grolmusz]{Szalkai2017c}
Balazs Szalkai, Balint Varga, and Vince Grolmusz.
\newblock Comparing advanced graph-theoretical parameters of the connectomes of
  the lobes of the human brain.
\newblock \emph{Cognitive Neurodynamics}, 12\penalty0 (6):\penalty0 549--559,
  2018{\natexlab{b}}.

\bibitem[Szalkai et~al.(2019{\natexlab{b}})Szalkai, Varga, and
  Grolmusz]{Szalkai2016c}
Balazs Szalkai, Balint Varga, and Vince Grolmusz.
\newblock Mapping correlations of psychological and connectomical properties of
  the dataset of the human connectome project with the maximum spanning tree
  method.
\newblock \emph{Brain Imaging and Behavior}, 13\penalty0 (5):\penalty0
  1185--1192, feb 2019{\natexlab{b}}.
\newblock \doi{https://doi.org/10.1007/s11682-018-9937-6}.

\bibitem[Fellner et~al.(2019)Fellner, Varga, and Grolmusz]{Fellner2017}
Mate Fellner, Balint Varga, and Vince Grolmusz.
\newblock The frequent subgraphs of the connectome of the human brain.
\newblock \emph{Cognitive Neurodynamics}, 13\penalty0 (5):\penalty0 453--460,
  2019.
\newblock URL \url{https://doi.org/10.1007 /s11571-019-09535-y}.

\bibitem[Fellner et~al.(2020{\natexlab{a}})Fellner, Varga, and
  Grolmusz]{Fellner2018}
Mate Fellner, Balint Varga, and Vince Grolmusz.
\newblock The frequent network neighborhood mapping of the human hippocampus
  shows much more frequent neighbor sets in males than in females.
\newblock \emph{PLOS One}, 15\penalty0 (1):\penalty0 e0227910,
  2020{\natexlab{a}}.
\newblock URL \url{https://doi.org/10.1371/journal.pone.0227910}.

\bibitem[Fellner et~al.(2020{\natexlab{b}})Fellner, Varga, and
  Grolmusz]{Fellner2019}
M{\'a}t{\'e} Fellner, B{\'a}lint Varga, and Vince Grolmusz.
\newblock The frequent complete subgraphs in the human connectome.
\newblock \emph{PloS One}, 15\penalty0 (8):\penalty0 e0236883,
  2020{\natexlab{b}}.
\newblock URL \url{https://doi.org/10.1371/journal.pone.0236883}.

\bibitem[Fellner et~al.(2020{\natexlab{c}})Fellner, Varga, and
  Grolmusz]{Fellner2019a}
Mate Fellner, Balint Varga, and Vince Grolmusz.
\newblock Good neighbors, bad neighbors: The frequent network neighborhood
  mapping of the hippocampus enlightens several structural factors of the human
  intelligence on a 414-subject cohort.
\newblock \emph{Scientific Reports}, 10\penalty0 (11967), 2020{\natexlab{c}}.
\newblock URL \url{https://doi.org/10.1038/s41598-020-68914-2}.

\bibitem[Fischl(2012)]{Fischl2012}
Bruce Fischl.
\newblock Freesurfer.
\newblock \emph{Neuroimage}, 62\penalty0 (2):\penalty0 774--781, 2012.

\bibitem[Desikan et~al.(2006)Desikan, Segonne, Fischl, Quinn, Dickerson,
  Blacker, Buckner, Dale, Maguire, Hyman, Albert, and Killiany]{Desikan2006}
Rahul~S. Desikan, Florent Segonne, Bruce Fischl, Brian~T. Quinn, Bradford~C.
  Dickerson, Deborah Blacker, Randy~L. Buckner, Anders~M. Dale, R~Paul Maguire,
  Bradley~T. Hyman, Marilyn~S. Albert, and Ronald~J. Killiany.
\newblock An automated labeling system for subdividing the human cerebral
  cortex on {MRI} scans into gyral based regions of interest.
\newblock \emph{Neuroimage}, 31\penalty0 (3):\penalty0 968--980, Jul 2006.
\newblock \doi{10.1016/j.neuroimage.2006.01.021}.
\newblock URL \url{http://dx.doi.org/10.1016/j.neuroimage.2006.01.021}.

\bibitem[Szalkai et~al.(2015{\natexlab{b}})Szalkai, Kerepesi, Varga, and
  Grolmusz]{Szalkai2015a}
Bal{\'a}zs Szalkai, Csaba Kerepesi, B{\'a}lint Varga, and Vince Grolmusz.
\newblock The {B}udapest {R}eference {C}onnectome {S}erver v2. 0.
\newblock \emph{Neuroscience Letters}, 595:\penalty0 60--62,
  2015{\natexlab{b}}.

\bibitem[Szalkai et~al.(2017{\natexlab{b}})Szalkai, Kerepesi, Varga, and
  Grolmusz]{Szalkai2016}
Balazs Szalkai, Csaba Kerepesi, Balint Varga, and Vince Grolmusz.
\newblock Parameterizable consensus connectomes from the {H}uman {C}onnectome
  {P}roject: The {B}udapest {R}eference {C}onnectome {S}erver v3.0.
\newblock \emph{Cognitive Neurodynamics}, 11\penalty0 (1):\penalty0 113--116,
  feb 2017{\natexlab{b}}.
\newblock \doi{http://dx.doi.org/10.1007/s11571-016-9407-z}.

\bibitem[Spearman()]{Spearman}
Carl Spearman.
\newblock The proof and measurement of association between two things.
\newblock \emph{The American Journal of Psychology}, 15\penalty0 (1):\penalty0
  72--101.
\newblock URL \url{https://doi.org/10.2307/1412159}.

\end{thebibliography}
%\bibliographystyle{unsrtnat}

\section*{Appendix}

In this section we list the orderings of the 1015 vertices of our consensus braingraph by the examined weight parameters. The nodes are denoted by the cerebral areas, which they are corresponded to. The columns contain the orderings by different parameters, denoted in the column header.



\section*{Supplementary material}

| Degree | Sum of fiber count (SUM-weight) | Max of fiber count (MAX-weight) | Average of fiber count (AVG-weight) | Sum of fiber length (SUM-length) | Max of fiber length (MAX-length) | Average of fiber length (AVG-length) |
|---|---|---|---|---|---|---|
| 1 Left-Thalamus-Proper | Left-Caudate | Left-Caudate | Left-Caudate | Left-Thalamus-Proper | lh.superiorfrontal_11 | rh.rostralmiddlefrontal_14 |
| 2 Right-Thalamus-Proper | Left-Thalamus-Proper | Left-Putamen | Left-Thalamus-Proper | Right-Thalamus-Proper | lh.superiorparietal_16 | lh.rostralmiddlefrontal_5 |
| 3 Right-Caudate | Right-Caudate | Right-Caudate | Right-Caudate | Right-Caudate | lh.lateralorbitofrontal_1 | lh.lateralorbitofrontal_5 |
| 4 Left-Caudate | Right-Thalamus-Proper | Right-Putamen | Left-Putamen | Left-Caudate | lh.superiorparietal_4 | rh.superiorfrontal_2 |
| 5 Right-Putamen | Left-Putamen | Right-Thalamus-Proper | Right-Putamen | Right-Putamen | lh.medialorbitofrontal_5 | lh.precuneus_20 |
| 6 Left-Putamen | Right-Putamen | Left-Thalamus-Proper | Right-Thalamus-Proper | Left-Putamen | lh.rostralmiddlefrontal_1 | lh.superiorfrontal_32 |
| 7 Right-Hippocampus | Left-Hippocampus | lh.posteriorcingulate_6 | Left-Hippocampus | Right-Hippocampus | lh.precuneus_17 | rh.rostralmiddlefrontal_23 |
| 8 Left-Hippocampus | Right-Hippocampus | Left-Pallidum | rh.caudalmiddlefrontal_11 | lh.superiorparietal_25 | lh.rostralmiddlefrontal_3 | rh.insula_14 |
| 9 lh.superiorparietal_25 | rh.caudalmiddlefrontal_11 | Left-Hippocampus | rh.caudalmiddlefrontal_11 | Right-Pallidum | lh.superiorparietal_19 | rh.rostralmiddlefrontal_2 |
| 10 rh.insula_16 | lh.superiorparietal_25 | Right-Hippocampus | lh.lateraloccipital_3 | Left-Hippocampus | lh.lateralorbitofrontal_14 | rh.rostralmiddlefrontal_19 |
| 11 rh.posteriorcingulate_1 | rh.superiorparietal_13 | rh.posteriorcingulate_7 | rh.lateraloccipital_3 | Left-Pallidum | lh.superiorparietal_14 | rh.caudalmiddlefrontal_6 |
| 12 rh.supramarginal_17 | lh.caudalmiddlefrontal_13 | lh.caudalmiddlefrontal_13 | lh.pericalcarine_6 | rh.posteriorcingulate_1 | lh.superiorfrontal_8 | rh.superiorfrontal_11 |
| 13 lh.insula_5 | rh.supramarginal_17 | Right-Pallidum | rh.lingual_6 | rh.caudalmiddlefrontal_11 | lh.precuneus_4 | rh.rostralmiddlefrontal_24 |
| 14 Left-Pallidum | lh.superiorparietal_27 | rh.caudalmiddlefrontal_11 | rh.superiorparietal_13 | rh.posteriorcingulate_5 | lh.superiorparietal_25 | rh.rostralmiddlefrontal_10 |
| 15 Right-Pallidum | lh.isthmuscingulate_2 | lh.posteriorcingulate_4 | rh.fusiform_2 | lh.caudalmiddlefrontal_13 | rh.rostralmiddlefrontal_24 | lh.superiorfrontal_29 |
| 16 rh.posteriorcingulate_5 | lh.precentral_6 | rh.posteriorcingulate_1 | lh.caudalmiddlefrontal_13 | rh.posteriorcingulate_5 | rh.fusiform_9 | rh.caudalmiddlefrontal_10 |
| 17 lh.posteriorcingulate_5 | lh.insula_5 | lh.posteriorcingulate_2 | lh.superiorparietal_25 | rh.caudalanteriorcingulate_3 | lh.medialorbitofrontal_4 | rh.rostralmiddlefrontal_21 |
| 18 lh.posteriorcingulate_6 | lh.supramarginal_16 | rh.isthmuscingulate_3 | lh.lateraloccipital_18 | lh.superiorparietal_27 | lh.superiorfrontal_10 | lh.superiorparietal_22 |
| 19 rh.posteriorcingulate_6 | rh.lingual_6 | rh.isthmuscingulate_5 | lh.isthmuscingulate_2 | rh.insula_10 | lh.precuneus_13 | lh.superiorfrontal_23 |
| 20 lh.superiorparietal_27 | lh.posteriorcingulate_6 | lh.supramarginal_16 | lh.superiorparietal_27 | rh.insula_16 | rh.lateraloccipital_17 | rh.caudalmiddlefrontal_9 |
| 21 rh.caudalanteriorcingulate_3 | rh.isthmuscingulate_5 | lh.caudalanteriorcingulate_1 | rh.lateraloccipital_5 | lh.posteriorcingulate_6 | lh.precuneus_20 | lh.parsopercularis_7 |
| 22 lh.isthmuscingulate_4 | rh.superiorparietal_24 | lh.superiorparietal_27 | rh.supramarginal_17 | rh.superiorparietal_24 | lh.rostralmiddlefrontal_7 | lh.superiorfrontal_44 |
| 23 rh.superiorparietal_13 | rh.insula_13 | lh.isthmuscingulate_2 | rh.lateraloccipital_15 | rh.superiorparietal_22 | lh.lateralorbitofrontal_13 | rh.caudalmiddlefrontal_1 |
| 24 rh.superiorparietal_24 | rh.lateraloccipital_3 | lh.isthmuscingulate_4 | lh.precuneus_15 | rh.caudalmiddlefrontal_10 | lh.precuneus_19 | rh.parsopercularis_6 |
| 25 rh.caudalmiddlefrontal_11 | lh.caudalmiddlefrontal_5 | lh.posteriorcingulate_1 | lh.precentral_6 | rh.posteriorcingulate_6 | rh.rostralmiddlefrontal_12 | rh.caudalmiddlefrontal_9 |
| 26 rh.insula_13 | lh.precuneus_15 | rh.posteriorcingulate_2 | rh.pericalcarine_1 | rh.superiorparietal_13 | rh.superiorparietal_10 | rh.insula_10 |
| 27 lh.supramarginal_16 | rh.inferiorparietal_20 | rh.posteriorcingulate_5 | rh.superiorfrontal_28 | lh.posteriorcingulate_8 | rh.pericalcarine_2 | rh.insula_8 |
| 28 lh.inferiorparietal_18 | lh.pericalcarine_6 | lh.precentral_6 | rh.pericalcarine_4 | rh.supramarginal_17 | rh.precuneus_16 | rh.rostralmiddlefrontal_18 |
| 29 lh.caudalmiddlefrontal_13 | Right-Pallidum | lh.posteriorcingulate_5 | lh.inferiortemporal_7 | rh.insula_11 | rh.parsorbitalis_3 | rh.rostralmiddlefrontal_23 |
| 30 lh.posteriorcingulate_8 | rh.supramarginal_11 | lh.superiorparietal_25 | lh.supramarginal_16 | rh.inferiorparietal_13 | rh.superiorparietal_15 | lh.precuneus_4 |
| 31 lh.posteriorcingulate_2 | rh.precentral_7 | lh.caudalanteriorcingulate_3 | lh.fusiform_17 | rh.superiorparietal_14 | rh.lingual_4 | rh.caudalmiddlefrontal_7 |
| 32 lh.inferiorparietal_13 | Left-Pallidum | lh.posteriorcingulate_8 | lh.precentral_7 | lh.superiorparietal_16 | lh.rostralmiddlefrontal_3 | lh.superiorparietal_19 |
| 33 lh.posteriorcingulate_1 | lh.precuneus_12 | rh.superiorparietal_13 | lh.insula_5 | lh.posteriorcingulate_2 | rh.rostralmiddlefrontal_22 | rh.superiorparietal_25 |
| 34 rh.supramarginal_19 | lh.lateraloccipital_3 | rh.posteriorcingulate_6 | lh.supramarginal_16 | rh.fusiform_2 | Right-Pallidum |
| 35 lh.precuneus_12 | lh.inferiorparietal_21 | rh.posteriorcingulate_9 | rh.isthmuscingulate_5 | lh.isthmuscingulate_5 | rh.superiorparietal_9 | rh.parsopercularis_9 |
| 36 rh.isthmuscingulate_5 | rh.supramarginal_19 | lh.pericalcarine_6 | rh.inferiorparietal_20 | lh.caudalanteriorcingulate_2 | lh.lateralorbitofrontal_5 | Right-Thalamus-Proper |
| 37 lh.inferiorparietal_20 | rh.lateraloccipital_5 | lh.lateraloccipital_18 | rh.fusiform_17 | lh.precuneus_12 | lh.medialorbitofrontal_7 | rh.insula_9 |
| 38 lh.isthmuscingulate_2 | rh.bankssts_1 | Brain-Stem | lh.caudalmiddlefrontal_5 | rh.caudalmiddlefrontal_6 | lh.rostralmiddlefrontal_2 | rh.caudalmiddlefrontal_11 |
| 39 lh.isthmuscingulate_1 | lh.isthmuscingulate_3 | lh.caudalanteriorcingulate_2 | rh.supramarginal_11 | lh.caudalmiddlefrontal_5 | rh.superiorparietal_11 | rh.caudalmiddlefrontal_4 |
| 40 rh.insula_10 | rh.insula_11 | lh.isthmuscingulate_3 | rh.lingual_7 | lh.insula_5 | rh.rostralmiddlefrontal_6 | lh.superiorfrontal_21 |
| 41 lh.inferiorparietal_21 | rh.precentral_20 | rh.insula_16 | rh.superiorparietal_24 | lh.posteriorcingulate_1 | rh.superiorparietal_22 | rh.superiorfrontal_14 |
| 42 rh.insula_11 | rh.supramarginal_5 | rh.caudalanteriorcingulate_3 | lh.posteriorcingulate_6 | lh.inferiorparietal_20 | lh.superiorfrontal_21 | rh.rostralmiddlefrontal_22 |
| 43 lh.caudalanteriorcingulate_2 | rh.precentral_19 | lh.insula_5 | rh.pericalcarine_2 | lh.isthmuscingulate_5 | lh.precuneus_6 | rh.rostralmiddlefrontal_4 |
| 44 lh.insula_16 | lh.lateraloccipital_18 | lh.inferiorparietal_21 | rh.precentral_20 | rh.parsopercularis_9 | rh.parsorbitalis_1 | lh.rostralmiddlefrontal_26 |

| | | | | | | | |
|---|---|---|---|---|---|---|---|
| 45 | lh.posteriorcingulate_9 | rh.inferiorparietal_18 | rh.caudalanteriorcingulate_2 | lh.lateraloccipital_23 | rh.caudalmiddlefrontal_1 | rh.caudalmiddlefrontal_6 | lh.rostralmiddlefrontal_3 |
| 46 | lh.precentral_6 | rh.posteriorcingulate_1 | lh.isthmuscingulate_1 | lh.lingual_14 | rh.insula_13 | lh.rostralmiddlefrontal_14 | lh.superiorparietal_10 |
| 47 | rh.precuneus_20 | lh.insula_2 | rh.parsopercularis_2 | rh.lateraloccipital_7 | rh.inferiorparietal_18 | rh.rostralmiddlefrontal_14 | rh.parsopercularis_6 |
| 48 | lh.isthmuscingulate_6 | rh.posteriorcingulate_6 | rh.posteriorcingulate_3 | lh.isthmuscingulate_3 | rh.insula_1 | rh.lingual_5 | rh.inferiorparietal_16 |
| 49 | lh.superiorparietal_16 | lh.inferiorparietal_20 | lh.isthmuscingulate_6 | rh.insula_13 | rh.rostralmiddlefrontal_21 | lh.precuneus_7 | lh.superiorparietal_15 |
| 50 | rh.isthmuscingulate_4 | rh.caudalmiddlefrontal_10 | rh.caudalanteriorcingulate_5 | rh.lateraloccipital_23 | lh.insula_15 | Left-Hippocampus | Left-Thalamus-Proper |
| 51 | rh.superiorparietal_22 | rh.inferiorparietal_13 | rh.supramarginal_11 | rh.precentral_19 | rh.precuneus_20 | rh.rostralmiddlefrontal_5 | lh.rostralmiddlefrontal_9 |
| 52 | rh.bankssts_1 | rh.precuneus_20 | rh.precuneus_6 | rh.supramarginal_5 | rh.inferiorparietal_21 | lh.precuneus_11 | lh.superiorfrontal_8 |
| 53 | rh.isthmuscingulate_3 | lh.posteriorcingulate_5 | rh.supramarginal_19 | rh.superiorfrontal_22 | rh.inferiorparietal_8 | rh.rostralmiddlefrontal_26 | lh.superiorfrontal_11 |
| 54 | lh.caudalmiddlefrontal_5 | rh.insula_16 | rh.lateraloccipital_3 | rh.inferiortemporal_13 | lh.superiorparietal_13 | rh.superiorparietal_13 | lh.caudalmiddlefrontal_3 |
| 55 | lh.superiorparietal_13 | lh.inferiortemporal_7 | rh.lingual_6 | rh.lingual_4 | lh.caudalanteriorcingulate_3 | lh.lateralorbitofrontal_7 | rh.inferiorparietal_1 |
| 56 | rh.caudalmiddlefrontal_10 | lh.isthmuscingulate_4 | rh.superiorfrontal_44 | rh.pericalcarine_7 | rh.parsopercularis_2 | lh.precuneus_5 | rh.rostralmiddlefrontal_3 |
| 57 | rh.superiorparietal_14 | lh.precentral_15 | rh.caudalmiddlefrontal_10 | rh.inferiorparietal_22 | lh.precentral_15 | lh.lateralorbitofrontal_8 | rh.rostralmiddlefrontal_13 |
| 58 | lh.supramarginal_15 | rh.superiorfrontal_28 | rh.inferiorparietal_20 | rh.bankssts_1 | lh.inferiorparietal_12 | lh.precuneus_21 | rh.rostralmiddlefrontal_20 |
| 59 | rh.inferiorparietal_21 | rh.posteriorcingulate_7 | rh.posteriorcingulate_8 | rh.pericalcarine_5 | rh.parsopercularis_6 | rh.rostralmiddlefrontal_10 | lh.caudalmiddlefrontal_8 |
| 60 | rh.inferiorparietal_12 | lh.superiortemporal_24 | lh.insula_2 | lh.precuneus_12 | rh.inferiorparietal_17 | lh.precuneus_22 | lh.caudalmiddlefrontal_13 |
| 61 | lh.caudalanteriorcingulate_3 | lh.lingual_2 | lh.insula_16 | lh.inferiorparietal_21 | rh.superiorparietal_9 | lh.lateralorbitofrontal_2 | rh.inferiorparietal_3 |
| 62 | lh.inferiorparietal_12 | lh.inferiorparietal_18 | rh.isthmuscingulate_6 | rh.precuneus_3 | lh.precuneus_4 | rh.lateraloccipital_3 | lh.caudalmiddlefrontal_6 |
| 63 | lh.precuneus_6 | lh.superiorparietal_13 | rh.insula_10 | lh.insula_2 | lh.isthmuscingulate_2 | rh.inferiorparietal_15 | rh.superiorfrontal_17 |
| 64 | lh.precentral_15 | lh.bankssts_4 | rh.paracentral_8 | rh.insula_11 | rh.parsopercularis_6 | lh.superiorparietal_25 | lh.paracentral_7 |
| 65 | lh.supramarginal_2 | rh.insula_10 | lh.caudalmiddlefrontal_9 | rh.supramarginal_19 | rh.inferiorparietal_12 | lh.lateralorbitofrontal_3 | lh.superiorparietal_24 |
| 66 | rh.precuneus_10 | rh.inferiorparietal_22 | rh.insula_13 | rh.bankssts_2 | lh.posteriorcingulate_9 | lh.superiorfrontal_10 | lh.superiorfrontal_41 |
| 67 | rh.superiorparietal_9 | lh.supramarginal_15 | lh.caudalanteriorcingulate_4 | lh.lateraloccipital_5 | rh.precuneus_6 | rh.precuneus_20 | rh.parsopercularis_1 |
| 68 | rh.inferiorparietal_8 | rh.fusiform_2 | rh.lateraloccipital_5 | lh.superiorfrontal_2 | rh.precuneus_10 | rh.rostralanteriorcingulate_3 | lh.rostralmiddlefrontal_22 |
| 69 | rh.caudalanteriorcingulate_5 | rh.caudalmiddlefrontal_6 | lh.postcentral_31 | rh.superiorfrontal_22 | rh.caudalanteriorcingulate_5 | lh.superiorfrontal_23 | lh.parstriangularis_2 |
| 70 | rh.inferiorparietal_20 | rh.isthmuscingulate_3 | rh.precentral_20 | rh.inferiorparietal_18 | rh.supramarginal_19 | lh.rostralmiddlefrontal_25 | Left-Pallidum |
| 71 | lh.isthmuscingulate_5 | lh.insula_14 | rh.superiorparietal_24 | Right-Pallidum | lh.caudalmiddlefrontal_9 | lh.superiorparietal_12 | lh.lateralorbitofrontal_1 |
| 72 | lh.superiortemporal_24 | rh.precuneus_6 | rh.caudalmiddlefrontal_1 | lh.lingual_8 | rh.inferiorparietal_14 | rh.superiorfrontal_2 | rh.precentral_29 |
| 73 | lh.insula_2 | rh.precuneus_3 | lh.inferiorparietal_18 | lh.lingual_1 | rh.rostralmiddlefrontal_4 | rh.rostralmiddlefrontal_21 | rh.insula_1 |
| 74 | rh.parsopercularis_2 | lh.caudalmiddlefrontal_10 | lh.superiorfrontal_27 | lh.lingual_6 | rh.posteriorcingulate_7 | rh.rostralmiddlefrontal_27 | lh.superiorparietal_14 |
| 75 | lh.insula_15 | lh.precuneus_6 | rh.caudalmiddlefrontal_5 | lh.bankssts_2 | rh.isthmuscingulate_1 | rh.precuneus_10 | lh.superiorfrontal_14 |
| 76 | rh.insula_1 | lh.lingual_7 | rh.caudalmiddlefrontal_6 | lh.bankssts_4 | rh.rostralmiddlefrontal_10 | lh.parstriangularis_2 | lh.insula_15 |
| 77 | rh.posteriorcingulate_7 | lh.caudalanteriorcingulate_3 | rh.superiorfrontal_28 | lh.postcentral_7 | rh.precentral_12 | lh.superiorparietal_1 | rh.paracentral_2 |
| 78 | rh.posteriorcingulate_2 | lh.superiorfrontal_44 | lh.precentral_15 | rh.caudalmiddlefrontal_10 | rh.isthmuscingulate_4 | lh.superiorparietal_24 | rh.parsopercularis_9 |
| 79 | rh.superiorparietal_11 | rh.parsopercularis_2 | rh.supramarginal_17 | lh.superiorfrontal_20 | lh.inferiorparietal_21 | rh.rostralmiddlefrontal_8 | rh.insula_7 |
| 80 | rh.supramarginal_11 | lh.lateraloccipital_23 | rh.inferiorparietal_12 | lh.superiorfrontal_44 | rh.rostralmiddlefrontal_18 | rh.precuneus_2 | lh.rostralmiddlefrontal_10 |
| 81 | rh.inferiorparietal_17 | rh.precentral_12 | rh.pericalcarine_1 | Left-Pallidum | lh.caudalmiddlefrontal_3 | rh.isthmuscingulate_4 | lh.caudalmiddlefrontal_1 |
| 82 | rh.precuneus_6 | lh.postcentral_7 | lh.rostralmiddlefrontal_4 | lh.inferiorparietal_20 | lh.precuneus_13 | rh.lateralorbitofrontal_15 | lh.superiorparietal_16 |
| 83 | lh.inferiorparietal_2 | rh.paracentral_8 | rh.bankssts_1 | rh.posteriorcingulate_7 | lh.insula_4 | rh.rostralmiddlefrontal_9 | lh.paracentral_4 |
| 84 | rh.supramarginal_5 | lh.postcentral_31 | lh.precentral_25 | rh.superiorfrontal_11 | lh.precentral_6 | rh.rostralmiddlefrontal_26 | lh.caudalmiddlefrontal_11 |
| 85 | rh.inferiorparietal_14 | rh.pericalcarine_1 | rh.precuneus_4 | rh.caudalmiddlefrontal_10 | lh.supramarginal_15 | rh.rostralanteriorcingulate_4 | lh.superiorparietal_5 |
| 86 | rh.insula_6 | rh.postcentral_24 | rh.inferiorparietal_13 | lh.precentral_15 | rh.isthmuscingulate_3 | rh.superiorparietal_18 | lh.superiorparietal_25 |
| 87 | lh.insula_14 | lh.posteriorcingulate_1 | lh.lateraloccipital_7 | rh.inferiorparietal_18 | lh.isthmuscingulate_6 | lh.superiorfrontal_32 | rh.posteriorcingulate_1 |
| 88 | rh.parsopercularis_9 | rh.precuneus_10 | lh.supramarginal_15 | lh.superiorfrontal_27 | lh.insula_14 | lh.rostralmiddlefrontal_5 | rh.inferiorparietal_17 |
| 89 | rh.precentral_19 | lh.bankssts_2 | rh.insula_11 | lh.superiorfrontal_27 | lh.superiorfrontal_44 | rh.superiorfrontal_14 | rh.insula_11 |
| 90 | lh.precuneus_13 | rh.superiorfrontal_11 | rh.superiorfrontal_18 | lh.superiortemporal_24 | lh.superiorfrontal_44 | rh.superiorfrontal_14 | rh.superiorfrontal_19 |

| # | Col1 | Col2 | Col3 | Col4 | Col5 | Col6 | Col7 |
|---|------|------|------|------|------|------|------|
| 91 | rh.precentral_12 | lh.inferiorparietal_2 | lh.lateraloccipital_3 | rh.precuneus_20 | lh.posteriorcingulate_4 | rh.superiorparietal_2 | rh.parsopercularis_2 |
| 92 | rh.supramarginal_10 | lh.superiorfrontal_22 | lh.lingual_1 | rh.caudalmiddlefrontal_6 | rh.inferiorparietal_20 | rh.rostralmiddlefrontal_1 | lh.rostralmiddlefrontal_21 |
| 93 | rh.caudalmiddlefrontal_1 | rh.bankssts_2 | lh.precuneus_15 | lh.precentral_21 | lh.caudalmiddlefrontal_4 | rh.paracentral_10 | lh.insula_4 |
| 94 | rh.caudalmiddlefrontal_6 | lh.precentral_23 | lh.bankssts_4 | lh.lateraloccipital_17 | rh.superiorparietal_11 | lh.superiorfrontal_14 | rh.inferiorparietal_2 |
| 95 | lh.posteriorcingulate_4 | lh.posteriorcingulate_8 | rh.caudalanteriorcingulate_4 | rh.lingual_3 | rh.posteriorcingulate_2 | rh.lateralorbitofrontal_1 | lh.posteriorcingulate_5 |
| 96 | lh.isthmuscingulate_7 | rh.inferiorparietal_14 | lh.inferiorparietal_17 | rh.lateraloccipital_22 | lh.caudalmiddlefrontal_10 | lh.precuneus_9 | rh.inferiorparietal_19 |
| 97 | lh.isthmuscingulate_3 | rh.precuneus_1 | lh.pericalcarine_4 | rh.postcentral_24 | rh.bankssts_1 | lh.rostralmiddlefrontal_23 | lh.precentral_23 |
| 98 | lh.insula_8 | lh.superiorfrontal_27 | rh.rostralmiddlefrontal_10 | rh.lingual_10 | rh.inferiorparietal_2 | lh.rostralmiddlefrontal_18 | lh.caudalmiddlefrontal_7 |
| 99 | rh.parsopercularis_6 | lh.posteriorcingulate_4 | rh.inferiorparietal_18 | rh.inferiorparietal_13 | rh.supramarginal_10 | rh.lateralorbitofrontal_3 | rh.superiorfrontal_27 |
| 100 | rh.posteriorcingulate_3 | lh.caudalmiddlefrontal_3 | lh.superiortemporal_24 | lh.insula_14 | rh.caudalmiddlefrontal_5 | lh.superiorparietal_21 | Right-Putamen |
| 101 | lh.postcentral_31 | lh.supramarginal_2 | lh.lingual_14 | rh.superiorparietal_3 | rh.insula_8 | rh.precuneus_9 | rh.rostralmiddlefrontal_27 |
| 102 | lh.rostralmiddlefrontal_21 | lh.isthmuscingulate_6 | lh.posteriorcingulate_9 | rh.posteriorcingulate_6 | lh.superiorparietal_24 | lh.superiorparietal_20 | lh.caudalmiddlefrontal_10 |
| 103 | lh.insula_4 | rh.inferiortemporal_13 | rh.lateraloccipital_17 | rh.superiorparietal_13 | lh.isthmuscingulate_5 | rh.precuneus_3 | lh.parsopercularis_8 |
| 104 | lh.bankssts_4 | lh.precuneus_4 | rh.fusiform_2 | rh.precuneus_6 | rh.precuneus_16 | rh.superiorfrontal_11 | lh.caudalmiddlefrontal_5 |
| 105 | rh.posteriorcingulate_9 | rh.superiorparietal_22 | lh.insula_4 | lh.superiorfrontal_45 | rh.superiorfrontal_11 | lh.precentral_33 | rh.precentral_13 |
| 106 | rh.supramarginal_6 | rh.precentral_22 | lh.lateraloccipital_7 | rh.paracentral_8 | rh.precentral_20 | rh.rostralmiddlefrontal_19 | lh.caudalmiddlefrontal_5 |
| 107 | lh.parsopercularis_6 | rh.inferiorparietal_12 | rh.superiorfrontal_11 | lh.lateraloccipital_17 | rh.precentral_23 | rh.medialorbitofrontal_10 | lh.lateralorbitofrontal_13 |
| 108 | lh.precuneus_11 | lh.precentral_21 | lh.precuneus_9 | lh.inferiorparietal_16 | rh.insula_6 | rh.precuneus_19 | rh.parstriangularis_8 |
| 109 | rh.precentral_20 | rh.superiorparietal_11 | lh.lingual_8 | rh.supramarginal_15 | lh.parsopercularis_8 | rh.superiorfrontal_29 | rh.inferiorfrontal_8 |
| 110 | rh.paracentral_8 | rh.precentral_15 | rh.insula_1 | rh.posteriorcingulate_1 | rh.inferiorparietal_2 | rh.superiorfrontal_35 | lh.precuneus_17 |
| 111 | rh.precuneus_1 | rh.superiorfrontal_22 | lh.superiorfrontal_28 | rh.superiorfrontal_25 | rh.precuneus_6 | lh.postcentral_18 | rh.precentral_12 |
| 112 | rh.precuneus_4 | lh.superiorfrontal_2 | lh.superiorfrontal_2 | rh.superiorfrontal_24 | rh.postcentral_31 | lh.paracentral_3 | rh.inferiorparietal_14 |
| 113 | lh.rostralmiddlefrontal_4 | rh.fusiform_17 | lh.insula_8 | rh.precentral_12 | lh.rostralmiddlefrontal_19 | lh.rostralanteriorcingulate_2 | rh.paracentral_3 |
| 114 | lh.supramarginal_3 | lh.caudalmiddlefrontal_9 | rh.isthmuscingulate_4 | lh.superiorfrontal_43 | rh.inferiorparietal_19 | rh.superiorparietal_24 | rh.rostralmiddlefrontal_15 |
| 115 | lh.superiorparietal_11 | rh.posteriorcingulate_9 | lh.inferiorparietal_12 | rh.precentral_23 | rh.supramarginal_11 | rh.rostralmiddlefrontal_4 | rh.rostralmiddlefrontal_4 |
| 116 | lh.precuneus_4 | lh.precuneus_9 | rh.precentral_19 | rh.pericalcarine_2 | rh.posteriorcingulate_7 | rh.superiorparietal_22 | lh.superiorfrontal_42 |
| 117 | lh.caudalmiddlefrontal_3 | lh.caudalmiddlefrontal_4 | rh.inferiorparietal_22 | rh.parsopercularis_2 | lh.caudalanteriorcingulate_4 | rh.lingual_6 | rh.superiorparietal_10 |
| 118 | lh.superiorparietal_20 | lh.precuneus_13 | rh.precuneus_3 | lh.superiorparietal_1 | lh.superiorparietal_11 | lh.isthmuscingulate_3 | rh.precentral_33 |
| 119 | lh.caudalmiddlefrontal_10 | rh.lateraloccipital_15 | lh.precuneus_12 | lh.superiorfrontal_1 | lh.supramarginal_2 | lh.superiorparietal_13 | rh.precentral_14 |
| 120 | lh.inferiorparietal_18 | rh.insula_6 | lh.lateraloccipital_23 | rh.insula_2 | lh.isthmuscingulate_3 | lh.rostralmiddlefrontal_21 | lh.caudalanteriorcingulate_2 |
| 121 | rh.caudalmiddlefrontal_5 | rh.caudalmiddlefrontal_1 | lh.postcentral_7 | rh.isthmuscingulate_3 | lh.isthmuscingulate_3 | rh.rostralanteriorcingulate_1 | rh.caudalanteriorcingulate_3 |
| 122 | lh.caudalmiddlefrontal_9 | rh.postcentral_26 | lh.insula_14 | lh.precuneus_6 | rh.superiorparietal_10 | rh.rostralmiddlefrontal_23 | lh.rostralanteriorcingulate_3 |
| 123 | lh.caudalanteriorcingulate_4 | rh.posteriorcingulate_5 | lh.caudalmiddlefrontal_3 | rh.pericalcarine_5 | rh.precuneus_4 | lh.parstriangularis_6 | lh.lateralorbitofrontal_8 |
| 124 | rh.precuneus_16 | rh.lateraloccipital_7 | lh.lingual_7 | lh.lateraloccipital_21 | rh.posteriorcingulate_3 | lh.precuneus_6 | rh.precuneus_16 |
| 125 | rh.precentral_7 | rh.superiorfrontal_31 | lh.inferiortemporal_7 | rh.superiorfrontal_31 | lh.isthmuscingulate_7 | lh.inferiorparietal_20 | rh.superiorparietal_24 |
| 126 | lh.precentral_4 | lh.fusiform_17 | lh.lingual_2 | lh.isthmuscingulate_4 | lh.insula_8 | rh.inferiorparietal_13 | lh.precentral_15 |
| 127 | rh.isthmuscingulate_6 | lh.superiorfrontal_20 | lh.superiorfrontal_45 | rh.precentral_15 | lh.rostralmiddlefrontal_13 | rh.superiorparietal_5 | lh.posteriorcingulate_8 |
| 128 | rh.supramarginal_12 | lh.insula_8 | rh.pericalcarine_7 | rh.posteriorcingulate_5 | rh.supramarginal_5 | rh.inferiorparietal_24 | lh.precentral_33 |
| 129 | rh.precuneus_9 | lh.rostralmiddlefrontal_4 | rh.lateraloccipital_23 | rh.superiorfrontal_41 | lh.rostralmiddlefrontal_26 | lh.paracentral_4 | lh.superiorparietal_27 |
| 130 | lh.posteriorcingulate_7 | rh.superiorfrontal_24 | lh.caudalanteriorcingulate_3 | lh.precuneus_20 | lh.insula_4 | rh.rostralmiddlefrontal_19 |
| 131 | rh.rostralmiddlefrontal_18 | lh.precentral_16 | rh.inferiorparietal_8 | rh.inferiorparietal_16 | rh.supramarginal_8 | rh.rostralanteriorcingulate_5 | rh.posteriorcingulate_7 |
| 132 | rh.precentral_9 | lh.isthmuscingulate_5 | lh.rostralmiddlefrontal_9 | rh.inferiorparietal_23 | lh.superiorparietal_20 | lh.inferiorparietal_17 | rh.precentral_19 |
| 133 | lh.precuneus_9 | rh.supramarginal_12 | lh.superiorfrontal_37 | rh.postcentral_26 | lh.precentral_4 | rh.caudalmiddlefrontal_8 | lh.precuneus_13 |
| 134 | rh.posteriorcingulate_8 | lh.precentral_4 | lh.superiorfrontal_43 | rh.isthmuscingulate_6 | lh.parsopercularis_7 | lh.superiorparietal_11 | lh.superiorparietal_15 |
| 135 | rh.precentral_22 | lh.inferiorparietal_12 | lh.inferiorparietal_6 | rh.lingual_12 | rh.insula_7 | lh.lingual_2 | lh.inferiorparietal_21 |
| 136 | rh.supramarginal_8 | lh.inferiorparietal_17 | lh.caudalanteriorcingulate_5 | rh.precentral_22 | rh.precentral_19 | Right-Hippocampus | lh.caudalanteriorcingulate_3 |

| # | Col1 | Col2 | Col3 | Col4 | Col5 | Col6 | Col7 |
|---|------|------|------|------|------|------|------|
| 137 | rh.inferiorparietal_9 | lh.precuneus_4 | rh.caudalanteriorcingulate_1 | lh.caudalmiddlefrontal_3 | rh.inferiorparietal_15 | lh.superiorparietal_15 | lh.superiorfrontal_10 |
| 138 | lh.parsopercularis_8 | Brain-Stem | lh.pericalcarine_5 | rh.precuneus_1 | rh.precentral_9 | lh.inferiorparietal_21 | rh.rostralmiddlefrontal_13 |
| 139 | lh.precentral_23 | lh.lateraloccipital_5 | lh.precuneus_11 | lh.fusiform_11 | rh.parstriangularis_8 | rh.isthmuscingulate_2 | rh.parsopercularis_4 |
| 140 | rh.rostralmiddlefrontal_10 | lh.superiorparietal_20 | rh.superiorfrontal_13 | lh.lateraloccipital_11 | lh.precuneus_11 | lh.inferiorparietal_15 | rh.superiorfrontal_18 |
| 141 | rh.inferiorparietal_2 | rh.caudalmiddlefrontal_5 | rh.precentral_7 | rh.middletemporal_12 | rh.precentral_7 | rh.paracentral_10 | rh.superiorfrontal_34 |
| 142 | lh.caudalmiddlefrontal_4 | rh.inferiorparietal_21 | rh.lateraloccipital_6 | rh.lateraloccipital_6 | rh.superiortemporal_24 | lh.parstriangularis_4 | rh.superiorparietal_13 |
| 143 | rh.postcentral_24 | rh.middletemporal_12 | rh.precuneus_19 | lh.inferiorparietal_2 | rh.posteriorcingulate_8 | lh.precuneus_10 | rh.superiorparietal_9 |
| 144 | rh.postcentral_30 | lh.inferiorparietal_16 | rh.inferiorparietal_23 | rh.precuneus_4 | rh.rostralmiddlefrontal_9 | rh.precuneus_8 | Right-Caudate |
| 145 | rh.superiorparietal_10 | rh.isthmuscingulate_1 | lh.bankssts_2 | lh.superiorfrontal_28 | lh.inferiorparietal_18 | lh.superiorparietal_42 | lh.paracentral_3 |
| 146 | rh.inferiorparietal_19 | lh.supramarginal_1 | lh.rostralmiddlefrontal_19 | rh.insula_16 | lh.precuneus_9 | rh.superiorparietal_2 | lh.superiorfrontal_35 |
| 147 | rh.inferiorparietal_15 | lh.precuneus_11 | rh.insula_5 | lh.posteriorcingulate_4 | lh.parstriangularis_2 | rh.superiorfrontal_8 | Right-Amygdala |
| 148 | lh.superiorparietal_24 | rh.superiorparietal_14 | lh.pericalcarine_2 | rh.precuneus_10 | rh.inferiorparietal_16 | lh.inferiorparietal_12 | lh.posteriorcingulate_4 |
| 149 | lh.precuneus_8 | lh.inferiorparietal_15 | rh.lateraloccipital_7 | lh.inferiorparietal_17 | rh.rostralmiddlefrontal_22 | rh.lingual_15 | rh.parsopercularis_7 |
| 150 | lh.superiorfrontal_44 | rh.parsopercularis_9 | rh.inferiorparietal_21 | lh.caudalmiddlefrontal_4 | rh.superiorparietal_18 | rh.rostralmiddlefrontal_7 | rh.inferiorparietal_26 |
| 151 | rh.rostralmiddlefrontal_24 | rh.inferiorparietal_23 | rh.precentral_22 | rh.inferiorparietal_14 | lh.supramarginal_3 | rh.superiorparietal_6 | lh.caudalmiddlefrontal_2 |
| 152 | rh.insula_5 | lh.isthmuscingulate_6 | rh.superiorparietal_24 | Brain-Stem | lh.caudalmiddlefrontal_1 | lh.bankssts_1 | rh.posteriorcingulate_5 |
| 153 | rh.insula_4 | lh.pericalcarine_4 | lh.inferiorparietal_20 | lh.superiorfrontal_38 | rh.paracentral_8 | lh.inferiorparietal_12 | lh.insula_10 |
| 154 | lh.superiorparietal_6 | lh.superiorparietal_24 | lh.inferiorparietal_16 | rh.superiorfrontal_42 | rh.supramarginal_6 | rh.caudalmiddlefrontal_11 | lh.superiorfrontal_31 |
| 155 | lh.precuneus_15 | rh.precuneus_16 | rh.superiorparietal_14 | lh.precuneus_9 | lh.precuneus_9 | rh.lateraloccipital_19 | lh.parstriangularis_3 |
| 156 | rh.inferiorparietal_22 | lh.caudalanteriorcingulate_1 | rh.superiorparietal_14 | lh.caudalmiddlefrontal_9 | rh.inferiorparietal_3 | rh.rostralanteriorcingulate_1 | rh.superiorfrontal_26 |
| 157 | lh.precentral_25 | rh.superiorparietal_10 | lh.precuneus_20 | lh.superiorfrontal_1 | rh.insula_4 | lh.precuneus_12 | rh.caudalanteriorcingulate_5 |
| 158 | lh.paracentral_6 | lh.insula_15 | lh.superiorparietal_13 | lh.superiorfrontal_31 | rh.precuneus_1 | rh.parsopercularis_6 | rh.caudalmiddlefrontal_4 |
| 159 | lh.supramarginal_1 | rh.superiorfrontal_42 | lh.insula_15 | rh.posteriorcingulate_9 | rh.postcentral_31 | lh.posteriorcingulate_7 | lh.superiorparietal_13 |
| 160 | rh.insula_8 | lh.precentral_25 | lh.caudalmiddlefrontal_1 | lh.inferiorparietal_6 | lh.superiorfrontal_29 | rh.lateraloccipital_5 | lh.inferiorparietal_12 |
| 161 | lh.paracentral_5 | rh.precentral_23 | lh.superiorfrontal_31 | rh.superiorfrontal_27 | rh.superiorparietal_25 | rh.caudalmiddlefrontal_10 | lh.superiorfrontal_25 |
| 162 | rh.precentral_8 | rh.paracentral_7 | lh.fusiform_17 | lh.supramarginal_2 | rh.posteriorcingulate_9 | rh.lateraloccipital_4 | rh.paracentral_1 |
| 163 | rh.precentral_15 | lh.posteriorcingulate_2 | lh.superiorfrontal_31 | lh.precuneus_18 | rh.precuneus_8 | lh.superiorfrontal_44 | lh.posteriorcingulate_7 |
| 164 | lh.postcentral_7 | lh.superiorfrontal_43 | rh.bankssts_2 | lh.posteriorcingulate_1 | lh.caudalanteriorcingulate_1 | lh.superiorparietal_14 | lh.posteriorcingulate_6 |
| 165 | lh.inferiorparietal_15 | rh.rostralmiddlefrontal_10 | lh.rostralmiddlefrontal_13 | rh.lateraloccipital_8 | lh.bankssts_4 | lh.parsopercularis_1 | rh.posteriorcingulate_6 |
| 166 | rh.parstriangularis_8 | rh.supramarginal_6 | lh.lateraloccipital_21 | rh.superiorfrontal_2 | rh.supramarginal_12 | lh.parahippocampal_1 | rh.parsopercularis_1 |
| 167 | rh.paracentral_7 | rh.isthmuscingulate_4 | lh.lingual_15 | rh.supramarginal_1 | rh.rostralmiddlefrontal_20 | lh.precentral_17 | rh.paracentral_5 |
| 168 | rh.postcentral_29 | rh.inferiorparietal_8 | rh.supramarginal_5 | lh.superiorfrontal_37 | rh.isthmuscingulate_6 | lh.superiorparietal_27 | Left-Putamen |
| 169 | rh.superiorparietal_18 | lh.posteriorcingulate_9 | rh.lateraloccipital_15 | rh.superiorparietal_11 | rh.inferiorparietal_9 | rh.precuneus_5 | rh.precuneus_10 |
| 170 | lh.supramarginal_21 | rh.lateraloccipital_23 | rh.lingual_10 | rh.fusiform_14 | rh.paracentral_6 | Right-Accumbens-area | rh.precuneus_20 |
| 171 | rh.caudalanteriorcingulate_1 | lh.fusiform_11 | rh.superiorfrontal_41 | lh.inferiorparietal_15 | lh.rostralmiddlefrontal_22 | rh.superiorfrontal_19 | lh.rostralanteriorcingulate_2 |
| 172 | rh.precuneus_3 | rh.caudalanteriorcingulate_5 | rh.superiorfrontal_17 | lh.caudalanteriorcingulate_1 | lh.supramarginal_21 | rh.caudalmiddlefrontal_9 | rh.superiorfrontal_22 |
| 173 | rh.precuneus_19 | lh.caudalanteriorcingulate_2 | rh.parsopercularis_6 | rh.supramarginal_12 | rh.postcentral_24 | rh.inferiorparietal_6 | lh.caudalanteriorcingulate_1 |
| 174 | rh.insula_7 | lh.insula_16 | rh.superiorparietal_3 | rh.superiorparietal_22 | rh.paracentral_5 | lh.postcentral_14 | rh.precentral_20 |
| 175 | lh.paracentral_11 | rh.superiorfrontal_25 | rh.inferiorparietal_24 | rh.precentral_23 | lh.inferiorparietal_22 | lh.superiorfrontal_9 | lh.superiorparietal_21 |
| 176 | lh.superiorparietal_2 | rh.postcentral_29 | rh.insula_6 | lh.superiorfrontal_42 | lh.precentral_8 | rh.lateralorbitofrontal_12 | lh.inferiorparietal_7 |
| 177 | lh.inferiorparietal_17 | lh.superiorparietal_1 | rh.precuneus_1 | rh.superiorfrontal_38 | lh.superiorfrontal_6 | rh.isthmuscingulate_5 | rh.precuneus_12 |
| 178 | lh.precuneus_22 | rh.superiorfrontal_41 | rh.lingual_5 | lh.precentral_4 | lh.caudalmiddlefrontal_8 | rh.paracentral_8 | lh.parsopercularis_4 |
| 179 | rh.superiorfrontal_11 | lh.precuneus_18 | lh.superiorparietal_1 | rh.caudalmiddlefrontal_1 | rh.caudalanteriorcingulate_4 | rh.rostralmiddlefrontal_17 | lh.superiorparietal_12 |
| 180 | rh.supramarginal_13 | rh.posteriorcingulate_2 | lh.pericalcarine_3 | rh.precuneus_13 | rh.rostralmiddlefrontal_19 | rh.superiorparietal_29 | rh.superiorfrontal_31 |
| 181 | lh.rostralmiddlefrontal_13 | rh.lateraloccipital_22 | rh.inferiorparietal_12 | rh.precentral_33 | rh.caudalanteriorcingulate_1 | rh.superiorfrontal_4 | lh.inferiorparietal_11 |
| 182 | lh.caudalanteriorcingulate_1 | lh.parsopercularis_6 | rh.lingual_4 | lh.rostralmiddlefrontal_4 | rh.caudalanteriorcingulate_1 | rh.supramarginal_11 | rh.caudalanteriorcingulate_4 |

| | | | | | | | |
|---|---|---|---|---|---|---|---|
| 183 rh.inferiorparietal_25 | lh.superiorparietal_2 | rh.middletemporal_12 | rh.superiorparietal_8 | lh.inferiorparietal_5 | rh.lateraloccipital_10 | lh.inferiorparietal_20 | |
| 184 rh.precuneus_5 | lh.lingual_6 | rh.precuneus_10 | rh.superiorfrontal_18 | rh.caudalmiddlefrontal_9 | rh.superiorfrontal_37 | rh.inferiorparietal_13 | |
| 185 lh.inferiorparietal_22 | rh.pericalcarine_2 | rh.supramarginal_10 | rh.lateraloccipital_9 | rh.inferiorparietal_1 | rh.rostralmiddlefrontal_13 | lh.superiorfrontal_16 | |
| 186 rh.precentral_3 | lh.pericalcarine_5 | rh.postcentral_24 | lh.superiorfrontal_29 | lh.precuneus_22 | lh.isthmuscingulate_5 | lh.caudalanteriorcingulate_4 | |
| 187 rh.superiortemporal_14 | lh.lateraloccipital_21 | lh.supramarginal_1 | rh.fusiform_9 | rh.precentral_22 | lh.superiorfrontal_24 | rh.insula_2 | |
| 188 lh.rostralmiddlefrontal_26 | rh.caudalanteriorcingulate_3 | lh.superiorfrontal_22 | rh.insula_6 | lh.parsopercularis_4 | rh.supramarginal_10 | lh.insula_14 | |
| 189 lh.inferiorparietal_5 | lh.superiorfrontal_31 | lh.paracentral_5 | lh.precuneus_4 | rh.precuneus_19 | rh.insula_2 | lh.precentral_12 | |
| 190 rh.precentral_23 | rh.superiorparietal_9 | rh.lingual_12 | rh.insula_8 | rh.inferiorparietal_22 | rh.insula_10 | lh.superiorfrontal_3 | |
| 191 rh.postcentral_26 | rh.superiorfrontal_27 | lh.inferiorparietal_2 | rh.paracentral_7 | lh.superiorfrontal_31 | rh.superiorfrontal_3 | rh.isthmuscingulate_5 | |
| 192 lh.rostralmiddlefrontal_19 | lh.superiorfrontal_38 | rh.rostralmiddlefrontal_21 | lh.fusiform_7 | lh.superiorparietal_10 | rh.inferiorparietal_8 | lh.precuneus_6 | |
| 193 rh.caudalanteriorcingulate_4 | rh.inferiorparietal_15 | rh.precentral_12 | rh.caudalmiddlefrontal_5 | rh.postcentral_29 | lh.insula_5 | rh.paracentral_4 | |
| 194 lh.parstriangularis_2 | lh.supramarginal_3 | rh.precentral_23 | lh.superiorparietal_20 | rh.inferiorparietal_25 | lh.postcentral_23 | lh.supramarginal_18 | |
| 195 lh.caudalmiddlefrontal_1 | lh.lingual_14 | lh.posteriorcingulate_3 | lh.posteriorcingulate_8 | rh.precentral_3 | rh.rostralmiddlefrontal_19 | rh.inferiorparietal_15 | |
| 196 rh.transversetemporal_3 | lh.superiorfrontal_29 | rh.rostralmiddlefrontal_24 | rh.precentral_25 | rh.insula_15 | rh.inferiorparietal_7 | rh.inferiorparietal_12 | |
| 197 lh.superiorfrontal_27 | rh.postcentral_30 | rh.supramarginal_6 | lh.paracentral_3 | rh.precentral_15 | rh.precuneus_4 | rh.inferiorparietal_22 | |
| 198 rh.insula_15 | lh.paracentral_5 | lh.superiorfrontal_29 | lh.superiorparietal_24 | lh.caudalmiddlefrontal_6 | rh.lateraloccipital_15 | lh.rostralmiddlefrontal_14 | |
| 199 rh.isthmuscingulate_2 | rh.precuneus_19 | rh.lateraloccipital_5 | lh.lateraloccipital_4 | rh.precentral_29 | lh.parsopercularis_7 | lh.superiorparietal_4 | |
| 200 lh.rostralmiddlefrontal_9 | rh.precentral_9 | rh.precentral_15 | lh.paracentral_4 | lh.posteriorcingulate_3 | lh.parahippocampal_5 | lh.posteriorcingulate_2 | |
| 201 lh.posteriorcingulate_3 | lh.lateraloccipital_7 | lh.isthmuscingulate_5 | rh.inferiorparietal_24 | rh.superiorparietal_15 | rh.rostralanteriorcingulate_3 | rh.supramarginal_8 | |
| 202 lh.isthmuscingulate_1 | rh.superiorparietal_3 | lh.caudalmiddlefrontal_10 | rh.superiorparietal_10 | lh.inferiorparietal_15 | rh.superiorparietal_12 | rh.precentral_17 | |
| 203 rh.postcentral_21 | rh.lingual_3 | lh.fusiform_11 | lh.lingual_10 | lh.postcentral_7 | lh.lateralorbitofrontal_10 | rh.superiorfrontal_35 | |
| 204 lh.supramarginal_4 | lh.superiorparietal_11 | rh.superiorfrontal_22 | rh.postcentral_29 | rh.insula_5 | lh.lateraloccipital_18 | lh.superiorfrontal_38 | |
| 205 rh.inferiorparietal_1 | lh.inferiorparietal_6 | lh.caudalmiddlefrontal_4 | rh.precuneus_11 | rh.caudalmiddlefrontal_7 | rh.precentral_23 | lh.rostralmiddlefrontal_1 | |
| 206 lh.bankssts_2 | rh.rostralmiddlefrontal_21 | lh.precuneus_4 | lh.superiorfrontal_3 | lh.superiorparietal_19 | lh.parstriangularis_2 | rh.supramarginal_10 | |
| 207 rh.rostralmiddlefrontal_19 | lh.superiorparietal_6 | lh.superiorfrontal_3 | lh.isthmuscingulate_5 | rh.rostralmiddlefrontal_23 | rh.rostralanteriorcingulate_2 | rh.superiorfrontal_13 | |
| 208 lh.parsopercularis_6 | lh.parsopercularis_6 | rh.precentral_21 | rh.precuneus_16 | rh.postcentral_21 | rh.lateralorbitofrontal_8 | Left-Caudate | |
| 209 rh.inferiorparietal_16 | lh.paracentral_4 | rh.supramarginal_12 | rh.superiorparietal_2 | rh.insula_14 | lh.inferiorparietal_21 | rh.rostralanteriorcingulate_3 | |
| 210 lh.parsopercularis_7 | lh.precentral_33 | rh.fusiform_17 | lh.lingual_15 | rh.paracentral_7 | rh.rostralanteriorcingulate_4 | lh.parsopercularis_3 | |
| 211 lh.parsopercularis_4 | rh.superiorparietal_8 | rh.lingual_3 | rh.pericalcarine_4 | lh.rostralmiddlefrontal_5 | rh.inferiorparietal_22 | lh.paracentral_9 | |
| 212 lh.superiorfrontal_31 | lh.lingual_15 | lh.lateraloccipital_17 | rh.superiorfrontal_41 | rh.rostralmiddlefrontal_14 | rh.insula_1 | rh.superiorparietal_16 | |
| 213 rh.rostralmiddlefrontal_22 | rh.supramarginal_10 | lh.supramarginal_2 | rh.rostralmiddlefrontal_10 | lh.precentral_33 | rh.precuneus_8 | lh.inferiorparietal_5 | |
| 214 rh.lingual_6 | rh.pericalcarine_7 | rh.superiorfrontal_42 | lh.inferiortemporal_16 | rh.paracentral_4 | rh.parstriangularis_7 | rh.superiorparietal_5 | |
| 215 lh.precentral_16 | lh.supramarginal_17 | rh.precuneus_6 | rh.superiorfrontal_21 | rh.precuneus_17 | rh.isthmuscingulate_1 | lh.rostralmiddlefrontal_7 | |
| 216 lh.supramarginal_17 | rh.lingual_4 | lh.superiorfrontal_20 | lh.inferiorparietal_12 | rh.isthmuscingulate_2 | lh.supramarginal_3 | rh.parstriangularis_6 | |
| 217 rh.precentral_33 | lh.paracentral_3 | lh.paracentral_4 | rh.lingual_9 | lh.paracentral_11 | lh.rostralmiddlefrontal_13 | Right-Hippocampus | |
| 218 rh.precuneus_23 | rh.inferiorparietal_24 | lh.paracentral_6 | rh.superiorfrontal_30 | lh.parsopercularis_1 | rh.parsopercularis_8 | lh.rostralanteriorcingulate_5 | |
| 219 lh.middletemporal_12 | rh.paracentral_11 | lh.parsopercularis_8 | rh.lingual_5 | lh.caudalmiddlefrontal_11 | rh.precuneus_14 | lh.superiorfrontal_22 | |
| 220 rh.postcentral_27 | rh.superiorparietal_18 | rh.postcentral_21 | rh.precentral_33 | lh.superiorfrontal_41 | lh.inferiorparietal_16 | rh.superiorparietal_11 | |
| 221 rh.paracentral_10 | lh.lateraloccipital_17 | rh.superiorparietal_8 | rh.parsopercularis_9 | lh.superiorfrontal_27 | lh.inferiorparietal_22 | lh.supramarginal_16 | |
| 222 lh.superiortemporal_18 | lh.superiorfrontal_42 | rh.lateraloccipital_22 | rh.precuneus_19 | rh.rostralmiddlefrontal_23 | rh.lateralorbitofrontal_9 | lh.superiorparietal_11 | |
| 223 rh.caudalanteriorcingulate_2 | rh.superiortemporal_14 | rh.inferiorparietal_14 | lh.lateraloccipital_19 | rh.parsopercularis_5 | rh.superiorfrontal_38 | rh.superiorparietal_12 | |
| 224 rh.supramarginal_7 | rh.lateraloccipital_17 | rh.postcentral_26 | lh.cuneus_1 | lh.inferiorparietal_7 | rh.precentral_13 | rh.insula_4 | |
| 225 rh.rostralmiddlefrontal_20 | rh.postcentral_21 | lh.lateraloccipital_16 | lh.lingual_11 | lh.paracentral_9 | rh.rostralmiddlefrontal_20 | lh.medialorbitofrontal_4 | |
| 226 rh.precentral_13 | lh.superiorfrontal_41 | rh.pericalcarine_4 | rh.inferiorparietal_21 | rh.precentral_23 | rh.inferiorparietal_6 | rh.superiorfrontal_24 | |
| 227 lh.precuneus_20 | lh.insula_4 | rh.lateraloccipital_19 | lh.pericalcarine_3 | rh.caudalanteriorcingulate_2 | rh.superiorfrontal_25 | rh.paracentral_11 | |
| 228 lh.rostralmiddlefrontal_22 | lh.precentral_33 | rh.insula_9 | rh.supramarginal_6 | rh.parsopercularis_7 | | lh.posteriorcingulate_1 | |

| # | Col1 | Col2 | Col3 | Col4 | Col5 | Col6 | Col7 | Col8 |
|---|------|------|------|------|------|------|------|------|
| 229 | rh.caudalmiddlefrontal_8 | lh.superiorparietal_16 | lh.precuneus_13 | lh.paracentral_5 | rh.inferiorparietal_11 | lh.inferiorparietal_3 | lh.posteriorcingulate_3 | |
| 230 | lh.precentral_21 | rh.insula_1 | rh.pericalcarine_5 | lh.precentral_33 | lh.parsopercularis_9 | lh.precentral_6 | rh.rostralmiddlefrontal_16 | |
| 231 | rh.superiorfrontal_42 | rh.supramarginal_13 | lh.lateraloccipital_11 | lh.lateraloccipital_16 | lh.inferiorparietal_1 | rh.lateralorbitofrontal_17 | lh.caudalmiddlefrontal_12 | |
| 232 | rh.precuneus_8 | rh.superiorfrontal_2 | rh.superiorparietal_10 | rh.precentral_29 | lh.caudalmiddlefrontal_7 | rh.rostralmiddlefrontal_10 | rh.insula_15 | |
| 233 | lh.bankssts_2 | rh.posteriorcingulate_8 | rh.rostralmiddlefrontal_23 | rh.inferiorparietal_15 | lh.rostralmiddlefrontal_15 | rh.rostralmiddlefrontal_18 | rh.superiorfrontal_24 | |
| 234 | lh.caudalmiddlefrontal_8 | lh.paracentral_6 | lh.pericalcarine_2 | lh.supramarginal_17 | rh.insula_9 | rh.precuneus_22 | rh.parsopercularis_8 | |
| 235 | rh.superiorfrontal_24 | rh.precentral_29 | rh.lingual_15 | rh.inferiortemporal_15 | rh.superiorfrontal_27 | lh.supramarginal_16 | rh.posteriorcingulate_2 | |
| 236 | rh.posteriorcingulate_4 | rh.caudalmiddlefrontal_8 | lh.lingual_6 | rh.superiorparietal_13 | rh.parstriangularis_3 | rh.precentral_12 | lh.postcentral_18 | |
| 237 | rh.superiorparietal_25 | lh.isthmuscingulate_7 | lh.superiorfrontal_38 | lh.superiorparietal_14 | rh.superiorfrontal_24 | rh.precentral_12 | Left-Hippocampus | |
| 238 | lh.precentral_33 | lh.lateraloccipital_6 | lh.superiorfrontal_27 | rh.postcentral_30 | rh.caudalmiddlefrontal_4 | lh.rostralmiddlefrontal_25 | lh.supramarginal_21 | |
| 239 | lh.postcentral_21 | lh.superiorfrontal_45 | rh.postcentral_30 | lh.insula_15 | rh.postcentral_30 | lh.caudalanteriorcingulate_5 | lh.superiorfrontal_2 | |
| 240 | rh.inferiorparietal_10 | rh.fusiform_14 | rh.precentral_16 | rh.paracentral_3 | rh.supramarginal_13 | lh.inferiorparietal_2 | rh.inferiorparietal_2 | |
| 241 | lh.inferiorparietal_7 | rh.lingual_5 | rh.superiorfrontal_2 | rh.isthmuscingulate_1 | rh.superiorfrontal_14 | rh.precentral_15 | rh.caudalanteriorcingulate_2 | |
| 242 | rh.inferiorparietal_11 | lh.precuneus_8 | rh.precuneus_16 | lh.parsopercularis_6 | rh.precuneus_5 | lh.superiorparietal_3 | rh.precuneus_11 | |
| 243 | rh.inferiorparietal_1 | lh.caudalmiddlefrontal_1 | lh.inferiorparietal_15 | lh.superiorparietal_6 | rh.caudalmiddlefrontal_8 | rh.postcentral_27 | rh.precentral_4 | |
| 244 | rh.lateraloccipital_5 | lh.precuneus_22 | rh.rostralmiddlefrontal_19 | rh.posteriorcingulate_5 | rh.supramarginal_18 | Right-Pallidum | rh.postcentral_21 | |
| 245 | rh.precentral_26 | rh.paracentral_3 | lh.precentral_33 | rh.bankssts_3 | rh.paracentral_3 | lh.insula_15 | rh.precentral_9 | |
| 246 | rh.precentral_16 | lh.lingual_8 | lh.paracentral_1 | rh.caudalmiddlefrontal_8 | lh.superiorparietal_2 | rh.parsopercularis_2 | rh.paracentral_12 | |
| 247 | lh.postcentral_9 | rh.bankssts_3 | rh.paracentral_11 | rh.postcentral_21 | rh.superiorfrontal_17 | rh.precentral_23 | rh.precentral_26 | |
| 248 | lh.precuneus_17 | rh.paracentral_10 | rh.superiortemporal_14 | lh.isthmuscingulate_6 | rh.superiorfrontal_42 | rh.caudalanteriorcingulate_4 | rh.precuneus_4 | |
| 249 | lh.precuneus_18 | rh.inferiorparietal_17 | rh.postcentral_29 | lh.superiorfrontal_7 | lh.inferiorparietal_17 | rh.parahippocampal_1 | rh.superiorfrontal_37 | |
| 250 | lh.lingual_15 | rh.posteriorcingulate_3 | rh.superiorfrontal_37 | rh.precentral_9 | lh.supramarginal_1 | lh.caudalmiddlefrontal_3 | lh.posteriorcingulate_9 | |
| 251 | rh.parsopercularis_7 | lh.rostralmiddlefrontal_19 | lh.isthmuscingulate_7 | rh.lateraloccipital_18 | Brain-Stem | rh.isthmuscingulate_3 | rh.fusiform_16 | |
| 252 | lh.superiorfrontal_29 | rh.pericalcarine_5 | rh.superiorfrontal_30 | lh.lateraloccipital_16 | rh.superiorfrontal_2 | rh.medialorbitofrontal_11 | lh.rostralmiddlefrontal_6 | |
| 253 | rh.precentral_29 | lh.inferiortemporal_16 | lh.superiorfrontal_1 | lh.lateraloccipital_8 | rh.precuneus_23 | rh.medialorbitofrontal_4 | rh.caudalmiddlefrontal_3 | |
| 254 | rh.caudalmiddlefrontal_3 | rh.lingual_1 | rh.precentral_4 | lh.middletemporal_6 | lh.precentral_13 | rh.inferiorparietal_16 | lh.caudalanteriorcingulate_5 | |
| 255 | lh.caudalmiddlefrontal_6 | lh.superiorfrontal_3 | lh.superiorfrontal_41 | lh.supramarginal_3 | rh.precuneus_3 | rh.superiorfrontal_10 | rh.superiorfrontal_25 | |
| 256 | rh.parstriangularis_3 | lh.rostralmiddlefrontal_24 | rh.inferiorparietal_17 | rh.superiortemporal_14 | lh.supramarginal_4 | rh.postcentral_29 | rh.posteriorcingulate_8 | |
| 257 | rh.caudalmiddlefrontal_4 | rh.inferiortemporal_15 | rh.lingual_10 | rh.inferiorparietal_8 | lh.superiorfrontal_31 | rh.parstriangularis_3 | rh.inferiorparietal_20 | |
| 258 | rh.precentral_34 | lh.superiorfrontal_28 | lh.superiorparietal_6 | rh.precentral_26 | rh.precuneus_15 | rh.inferiorparietal_26 | rh.caudalanteriorcingulate_1 | |
| 259 | lh.paracentral_4 | lh.precentral_26 | lh.superiortemporal_18 | rh.paracentral_11 | rh.postcentral_26 | lh.caudalmiddlefrontal_4 | rh.rostralanteriorcingulate_2 | |
| 260 | lh.supramarginal_18 | rh.superiorfrontal_13 | lh.precuneus_20 | rh.paracentral_11 | lh.superiorparietal_15 | lh.paracentral_11 | lh.precuneus_5 | |
| 261 | lh.lateraloccipital_23 | rh.parsopercularis_8 | rh.lingual_9 | lh.fusiform_12 | rh.caudalmiddlefrontal_3 | lh.superiorparietal_22 | lh.superiorparietal_26 | |
| 262 | lh.inferiortemporal_7 | lh.fusiform_7 | lh.lingual_11 | rh.caudalanteriorcingulate_5 | rh.precentral_16 | lh.inferiorparietal_5 | lh.superiorfrontal_17 | |
| 263 | rh.inferiorparietal_23 | rh.supramarginal_8 | lh.rostralanteriorcingulate_5 | rh.superiorparietal_18 | rh.paracentral_10 | rh.superiorfrontal_27 | rh.superiorparietal_17 | |
| 264 | lh.superiorparietal_10 | lh.middletemporal_6 | rh.caudalmiddlefrontal_5 | rh.posteriorcingulate_2 | rh.supramarginal_7 | rh.postcentral_21 | rh.rostralanteriorcingulate_4 | |
| 265 | rh.lingual_2 | rh.paracentral_11 | rh.superiorparietal_22 | rh.supramarginal_13 | rh.middletemporal_12 | rh.insula_11 | rh.precentral_3 | |
| 266 | lh.bankssts_3 | rh.rostralmiddlefrontal_18 | rh.parsopercularis_9 | rh.isthmuscingulate_4 | lh.isthmuscingulate_1 | rh.inferiorparietal_11 | lh.isthmuscingulate_2 | |
| 267 | rh.superiorparietal_15 | rh.supramarginal_7 | lh.supramarginal_4 | lh.superiorparietal_11 | lh.supramarginal_17 | rh.superiorfrontal_15 | lh.precuneus_22 | |
| 268 | rh.bankssts_5 | rh.postcentral_27 | rh.supramarginal_13 | lh.posteriorcingulate_9 | lh.superiorfrontal_22 | lh.caudalanteriorcingulate_4 | rh.inferiorparietal_25 | |
| 269 | lh.caudalmiddlefrontal_11 | lh.pericalcarine_2 | rh.superiorfrontal_19 | rh.rostralmiddlefrontal_21 | rh.superiorfrontal_22 | lh.precuneus_1 | lh.paracentral_6 | |
| 270 | rh.caudalmiddlefrontal_9 | lh.cuneus_1 | rh.insula_4 | lh.cuneus_2 | rh.caudalmiddlefrontal_2 | lh.isthmuscingulate_7 | lh.isthmuscingulate_4 | |
| 271 | rh.postcentral_13 | lh.supramarginal_21 | rh.insula_8 | lh.caudalanteriorcingulate_2 | lh.paracentral_3 | rh.parsopercularis_9 | lh.isthmuscingulate_4 | |
| 272 | rh.lingual_15 | lh.superiorfrontal_37 | rh.insula_14 | lh.insula_16 | rh.precuneus_8 | rh.precentral_20 | lh.isthmuscingulate_4 | |
| 273 | lh.lingual_7 | lh.supramarginal_4 | lh.superiorparietal_11 | rh.paracentral_10 | lh.bankssts_2 | rh.inferiorparietal_23 | lh.precentral_8 | |
| 274 | lh.bankssts_5 | lh.rostralmiddlefrontal_9 | lh.parsopercularis_6 | rh.superiorparietal_9 | rh.postcentral_27 | lh.superiorparietal_18 | rh.precuneus_7 | |

| # | | | | | | | |
|---|---|---|---|---|---|---|---|
| 275 | lh.rostralmiddlefrontal_15 | rh.lingual_10 | lh.precentral_23 | rh.parsopercularis_6 | rh.precentral_34 | Left-Accumbens-area | rh.parsopercularis_3 |
| 276 | lh.postcentral_8 | lh.superiorfrontal_21 | lh.rostralmiddlefrontal_22 | lh.paracentral_6 | rh.rostralmiddlefrontal_2 | lh.inferiorparietal_3 | lh.paracentral_5 |
| 277 | lh.transversetemporal_2 | rh.lingual_12 | rh.superiorparietal_9 | lh.caudalmiddlefrontal_1 | rh.precuneus_7 | rh.precentral_14 | lh.rostralmiddlefrontal_17 |
| 278 | lh.insula_12 | rh.superiorfrontal_18 | lh.isthmuscingulate_1 | lh.posteriorcingulate_2 | lh.caudalanteriorcingulate_5 | rh.inferiorparietal_20 | rh.isthmuscingulate_2 |
| 279 | rh.parsopercularis_5 | lh.lateraloccipital_11 | lh.inferiorparietal_3 | lh.superiorfrontal_23 | rh.precentral_21 | rh.supramarginal_6 | rh.precentral_7 |
| 280 | rh.postcentral_14 | rh.rostralmiddlefrontal_23 | lh.precentral_26 | lh.precuneus_22 | rh.paracentral_2 | rh.precentral_33 | lh.superiorparietal_20 |
| 281 | lh.parsopercularis_1 | rh.precentral_26 | lh.precuneus_22 | lh.postcentral_19 | lh.superiorfrontal_38 | lh.rostralmiddlefrontal_9 | rh.insula_16 |
| 282 | lh.superiorfrontal_22 | lh.superiortemporal_18 | rh.insula_7 | lh.paracentral_7 | lh.superiorfrontal_2 | lh.supramarginal_18 | rh.insula_6 |
| 283 | lh.supramarginal_7 | rh.supramarginal_7 | lh.superiorparietal_24 | lh.superiorparietal_23 | lh.inferiorparietal_23 | rh.caudalanteriorcingulate_6 | lh.superiorfrontal_13 |
| 284 | lh.inferiorparietal_16 | rh.fusiform_9 | lh.inferiortemporal_16 | rh.posteriorcingulate_8 | rh.superiortemporal_14 | lh.superiorfrontal_19 | rh.insula_12 |
| 285 | rh.precuneus_7 | lh.superiorfrontal_1 | lh.superiorparietal_2 | lh.fusiform_6 | lh.parsopercularis_3 | rh.inferiorparietal_2 | rh.posteriorcingulate_3 |
| 286 | lh.caudalmiddlefrontal_41 | rh.lingual_15 | lh.postcentral_19 | lh.superiorfrontal_24 | rh.precentral_26 | lh.superiortemporal_8 | rh.caudalmiddlefrontal_8 |
| 287 | lh.caudalmiddlefrontal_7 | lh.precuneus_20 | lh.precuneus_18 | rh.supramarginal_10 | rh.posteriorcingulate_4 | rh.lateraloccipital_6 | rh.inferiorparietal_6 |
| 288 | lh.postcentral_30 | lh.superiorfrontal_17 | lh.cuneus_1 | rh.rostralmiddlefrontal_23 | lh.superiorfrontal_42 | rh.insula_16 | lh.parsopercularis_5 |
| 289 | rh.superiorfrontal_27 | lh.inferiorparietal_5 | lh.lateraloccipital_4 | rh.paracentral_9 | rh.lateraloccipital_5 | lh.lateralorbitofrontal_4 | lh.inferiorparietal_18 |
| 290 | rh.caudalmiddlefrontal_7 | rh.precuneus_8 | lh.caudalmiddlefrontal_8 | lh.superiorfrontal_17 | rh.rostralanteriorcingulate_2 | rh.caudalanteriorcingulate_1 | rh.isthmuscingulate_3 |
| 291 | rh.bankssts_3 | rh.precentral_16 | Right-Amygdala | lh.rostralmiddlefrontal_19 | lh.paracentral_7 | rh.parahippocampal_3 | rh.precuneus_9 |
| 292 | lh.postcentral_13 | rh.precuneus_9 | rh.paracentral_3 | lh.insula_4 | lh.precentral_26 | rh.supramarginal_2 | lh.parstriangularis_6 |
| 293 | lh.postcentral_10 | rh.caudalanteriorcingulate_2 | rh.paracentral_11 | rh.superiorfrontal_37 | rh.inferiorparietal_10 | rh.rostralmiddlefrontal_11 | rh.precuneus_19 |
| 294 | lh.parsopercularis_9 | rh.lateraloccipital_8 | lh.superiorparietal_20 | rh.postcentral_6 | lh.postcentral_21 | lh.fusiform_7 | lh.supramarginal_15 |
| 295 | lh.caudalanteriorcingulate_5 | lh.postcentral_21 | rh.supramarginal_8 | rh.superiorfrontal_19 | rh.inferiorparietal_24 | rh.caudalmiddlefrontal_1 | rh.superiorfrontal_42 |
| 296 | lh.superiorparietal_19 | rh.precentral_8 | lh.middletemporal_6 | lh.precuneus_8 | rh.precentral_36 | lh.precuneus_2 | rh.superiorfrontal_28 |
| 297 | rh.rostralmiddlefrontal_23 | lh.rostralmiddlefrontal_13 | lh.superiorfrontal_21 | lh.precentral_28 | rh.paracentral_11 | rh.middletemporal_2 | rh.inferiorparietal_4 |
| 298 | rh.paracentral_3 | rh.caudalmiddlefrontal_3 | rh.rostralmiddlefrontal_20 | lh.superiorfrontal_1 | lh.precuneus_10 | lh.caudalmiddlefrontal_8 | lh.superiorparietal_22 |
| 299 | lh.superiorfrontal_31 | lh.paracentral_7 | rh.lateraloccipital_2 | lh.supramarginal_7 | rh.inferiorparietal_6 | rh.parstriangularis_8 | lh.superiorparietal_6 |
| 300 | lh.precentral_11 | rh.insula_5 | lh.superiorfrontal_42 | rh.supramarginal_7 | rh.superiorfrontal_28 | rh.inferiorparietal_19 | rh.caudalmiddlefrontal_2 |
| 301 | rh.precentral_36 | rh.superiorfrontal_38 | lh.superiorparietal_23 | lh.superiorfrontal_34 | rh.postcentral_8 | rh.precentral_27 | rh.rostralmiddlefrontal_7 |
| 302 | rh.superiorfrontal_22 | lh.inferiorparietal_1 | lh.superiorfrontal_26 | rh.lingual_15 | lh.precuneus_3 | lh.parsopercularis_6 | rh.postcentral_8 |
| 303 | rh.supramarginal_9 | lh.precentral_28 | rh.superiorfrontal_6 | rh.postcentral_27 | Right-Amygdala | lh.lateraloccipital_21 | lh.precuneus_8 |
| 304 | rh.inferiortemporal_13 | rh.rostralmiddlefrontal_19 | lh.superiorfrontal_7 | rh.precentral_26 | rh.transversetemporal_3 | rh.precentral_24 | lh.inferiorparietal_10 |
| 305 | lh.precentral_20 | lh.lingual_10 | rh.superiorparietal_11 | lh.isthmuscingulate_7 | rh.precuneus_21 | rh.lateraloccipital_5 | rh.isthmuscingulate_1 |
| 306 | lh.precuneus_3 | lh.lateraloccipital_4 | lh.precuneus_3 | rh.caudalanteriorcingulate_3 | rh.inferiorparietal_4 | rh.supramarginal_19 | rh.precentral_34 |
| 307 | rh.inferiorparietal_24 | rh.superiorfrontal_30 | rh.paracentral_10 | rh.lateraloccipital_19 | lh.postcentral_30 | lh.superiorfrontal_1 | rh.insula_13 |
| 308 | lh.superiorfrontal_2 | lh.fusiform_12 | rh.lateraloccipital_11 | rh.insula_1 | lh.superiorfrontal_20 | rh.inferiorparietal_25 | rh.precentral_25 |
| 309 | lh.supramarginal_12 | rh.precuneus_23 | rh.bankssts_3 | lh.superiorfrontal_26 | rh.postcentral_9 | lh.postcentral_30 | rh.supramarginal_17 |
| 310 | lh.superiorparietal_15 | rh.fusiform_15 | rh.caudalanteriorcingulate_6 | rh.fusiform_12 | rh.lingual_6 | rh.superiorfrontal_38 | rh.isthmuscingulate_3 |
| 311 | lh.precuneus_10 | lh.precuneus_3 | rh.precentral_26 | rh.fusiform_15 | lh.postcentral_10 | rh.precentral_35 | lh.inferiorparietal_13 |
| 312 | rh.superiorparietal_8 | lh.superiorfrontal_7 | lh.paracentral_3 | lh.fusiform_15 | lh.precentral_20 | lh.isthmuscingulate_2 | lh.precuneus_9 |
| 313 | lh.caudalmiddlefrontal_2 | lh.pericalcarine_3 | rh.superiorfrontal_34 | lh.supramarginal_4 | lh.inferiorparietal_16 | lh.caudalmiddlefrontal_1 | lh.superiorparietal_28 |
| 314 | lh.fusiform_11 | rh.inferiorparietal_9 | rh.paracentral_7 | lh.rostralmiddlefrontal_9 | lh.superiortemporal_18 | rh.postcentral_2 | lh.precentral_3 |
| 315 | lh.paracentral_3 | rh.inferiorparietal_25 | rh.fusiform_14 | rh.rostralmiddlefrontal_24 | lh.bankssts_5 | lh.superiortemporal_41 | rh.precuneus_6 |
| 316 | lh.precuneus_21 | lh.caudalanteriorcingulate_4 | rh.superiorfrontal_36 | lh.superiortemporal_18 | rh.bankssts_2 | rh.precentral_8 | rh.supramarginal_3 |
| 317 | lh.superiorfrontal_38 | rh.precuneus_7 | lh.paracentral_7 | lh.superiorparietal_16 | rh.superiorfrontal_13 | lh.superiorparietal_7 | lh.inferiorparietal_18 |
| 318 | rh.inferiorparietal_6 | rh.postcentral_6 | rh.supramarginal_7 | lh.precuneus_20 | rh.rostralmiddlefrontal_16 | lh.caudalmiddlefrontal_9 | lh.superiorfrontal_27 |
| 319 | lh.parsopercularis_3 | lh.lateraloccipital_16 | rh.superiorparietal_18 | rh.precentral_16 | lh.precentral_12 | Right-Thalamus-Proper | rh.postcentral_29 |
| 320 | lh.precentral_26 | rh.superiorparietal_25 | rh.precuneus_11 | rh.lateraloccipital_11 | lh.superiorfrontal_32 | rh.superiorfrontal_1 | lh.insula_8 |

| | | | | | | | |
|---|---|---|---|---|---|---|---|
| 321 lh.lateraloccipital_18 | lh.superiorfrontal_26 | lh.superiorparietal_16 | lh.supramarginal_21 | lh.precentral_11 | lh.fusiform_17 | lh.isthmuscingulate_7 | |
| 322 rh.superiorfrontal_28 | rh.lateraloccipital_9 | lh.lateraloccipital_19 | rh.precuneus_8 | rh.bankssts_5 | rh.precentral_29 | rh.inferiorparietal_24 | |
| 323 rh.superiorfrontal_17 | lh.inferiorparietal_19 | lh.postcentral_10 | rh.posteriorcingulate_3 | lh.supramarginal_7 | lh.rostralmiddlefrontal_4 | rh.superiorfrontal_3 | |
| 324 rh.postcentral_8 | lh.paracentral_9 | rh.lateraloccipital_9 | rh.caudalmiddlefrontal_3 | rh.rostralanteriorcingulate_3 | rh.postcentral_14 | lh.inferiorparietal_11 | |
| 325 rh.lingual_5 | rh.rostralmiddlefrontal_20 | rh.precuneus_23 | lh.superiorfrontal_18 | lh.superiorfrontal_26 | lh.precentral_26 | lh.precentral_6 | |
| 326 rh.insula_14 | lh.superiorfrontal_23 | lh.supramarginal_3 | rh.paracentral_5 | lh.postcentral_18 | rh.superiorparietal_26 | lh.inferiorparietal_21 | |
| 327 rh.insula_9 | rh.inferiorparietal_10 | rh.precuneus_5 | lh.postcentral_21 | rh.superiorfrontal_25 | Left-Thalamus-Proper | lh.precuneus_16 | |
| 328 lh.superiorfrontal_20 | lh.postcentral_19 | rh.rostralmiddlefrontal_18 | lh.parsopercularis_8 | rh.inferiorparietal_26 | lh.fusiform_3 | lh.isthmuscingulate_5 | |
| 329 rh.precentral_10 | lh.fusiform_6 | lh.inferiorparietal_5 | rh.supramarginal_1 | lh.precentral_16 | lh.isthmuscingulate_6 | lh.precuneus_10 | |
| 330 rh.superiortemporal_21 | lh.postcentral_10 | rh.precuneus_18 | rh.caudalanteriorcingulate_2 | lh.supramarginal_9 | lh.insula_12 | rh.supramarginal_19 | |
| 331 rh.inferiorparietal_4 | rh.supramarginal_1 | rh.precuneus_7 | rh.inferiorparietal_17 | lh.inferiortemporal_7 | rh.superiorparietal_8 | rh.precentral_35 | |
| 332 lh.postcentral_9 | lh.caudalmiddlefrontal_11 | lh.precentral_19 | rh.supramarginal_8 | rh.lateralorbitofrontal_6 | lh.lateraloccipital_23 | lh.precuneus_11 | |
| 333 lh.inferiorparietal_6 | lh.parstriangularis_2 | rh.lateraloccipital_8 | rh.lateraloccipital_2 | lh.superiorfrontal_21 | rh.inferiorparietal_5 | lh.inferiorparietal_1 | |
| 334 rh.superiorfrontal_14 | lh.precuneus_17 | rh.precentral_29 | lh.superiorfrontal_35 | lh.lateraloccipital_23 | lh.isthmuscingulate_1 | lh.superiorfrontal_20 | |
| 335 rh.supramarginal_14 | rh.supramarginal_18 | rh.precuneus_8 | lh.precuneus_3 | rh.postcentral_14 | lh.parsopercularis_4 | rh.superiorfrontal_6 | |
| 336 rh.lateraloccipital_3 | rh.paracentral_5 | lh.paracentral_9 | rh.rostralmiddlefrontal_18 | rh.lateraloccipital_18 | lh.inferiorparietal_9 | rh.parsopercularis_10 | |
| 337 lh.rostralmiddlefrontal_5 | rh.precentral_27 | rh.superiorfrontal_26 | rh.cuneus_6 | lh.precentral_10 | lh.superiortemporal_24 | rh.precuneus_14 | |
| 338 rh.superiorfrontal_41 | rh.precentral_3 | rh.inferiortemporal_15 | lh.inferiorparietal_5 | lh.superiorparietal_22 | Right-Putamen | rh.supramarginal_11 | |
| 339 rh.paracentral_11 | rh.paracentral_2 | rh.fusiform_9 | rh.precentral_14 | rh.superiorfrontal_41 | rh.insula_9 | rh.precentral_15 | |
| 340 lh.rostralmiddlefrontal_23 | rh.inferiorparietal_6 | lh.superiorparietal_23 | rh.superiorfrontal_16 | rh.lateraloccipital_3 | lh.caudalmiddlefrontal_10 | lh.insula_3 | |
| 341 rh.lateralorbitofrontal_6 | rh.pericalcarine_4 | lh.rostralmiddlefrontal_15 | lh.precentral_3 | lh.inferiorparietal_9 | rh.superiorfrontal_20 | lh.paracentral_10 | |
| 342 rh.rostralmiddlefrontal_14 | lh.lateraloccipital_19 | rh.precentral_27 | rh.precentral_27 | lh.insula_12 | lh.caudalmiddlefrontal_5 | lh.inferiorparietal_15 | |
| 343 lh.superiorfrontal_43 | lh.posteriorcingulate_7 | rh.superiorfrontal_1 | lh.fusiform_10 | rh.precuneus_11 | rh.posteriorcingulate_4 | rh.postcentral_24 | |
| 344 lh.rostralanteriorcingulate_2 | lh.superiorfrontal_24 | lh.postcentral_13 | lh.paracentral_1 | rh.paracentral_5 | lh.rostralmiddlefrontal_15 | lh.inferiorparietal_8 | |
| 345 lh.inferiorparietal_9 | rh.lateraloccipital_18 | rh.cuneus_6 | rh.lateraloccipital_4 | lh.postcentral_13 | lh.postcentral_21 | lh.insula_11 | |
| 346 lh.middletemporal_6 | rh.superiorfrontal_37 | rh.precentral_9 | rh.paracentral_2 | rh.superiorparietal_5 | lh.medialorbitofrontal_6 | lh.precuneus_21 | |
| 347 lh.precentral_5 | lh.superiorfrontal_35 | lh.lateraloccipital_9 | lh.superiorfrontal_36 | lh.bankssts_3 | lh.posteriorcingulate_6 | lh.isthmuscingulate_6 | |
| 348 lh.superiorparietal_1 | lh.lateraloccipital_8 | rh.inferiorparietal_25 | lh.inferiorparietal_1 | rh.supramarginal_3 | rh.parsopercularis_4 | lh.postcentral_7 | |
| 349 lh.superiorfrontal_42 | rh.lingual_9 | lh.parsopercularis_4 | rh.precuneus_7 | rh.supramarginal_14 | lh.parsopercularis_9 | lh.superiorfrontal_19 | |
| 350 rh.paracentral_2 | lh.caudalmiddlefrontal_2 | rh.lateraloccipital_16 | rh.paracentral_10 | lh.superiorfrontal_3 | rh.superiorfrontal_16 | rh.inferiorparietal_22 | |
| 351 rh.inferiortemporal_15 | lh.caudalmiddlefrontal_6 | lh.fusiform_7 | lh.rostralmiddlefrontal_13 | lh.superiorfrontal_23 | lh.lingual_7 | rh.precentral_36 | |
| 352 lh.supramarginal_19 | rh.middletemporal_16 | rh.supramarginal_1 | lh.fusiform_16 | lh.inferiorparietal_4 | rh.supramarginal_8 | lh.superiorfrontal_9 | |
| 353 lh.inferiorfrontal_4 | rh.lateraloccipital_2 | lh.lingual_17 | rh.paracentral_12 | rh.parsopercularis_1 | lh.posteriorcingulate_1 | lh.superiorfrontal_1 | |
| 354 rh.rostralmiddlefrontal_16 | lh.parsopercularis_7 | rh.precentral_14 | rh.rostralmiddlefrontal_19 | lh.precentral_32 | lh.precuneus_16 | lh.precentral_13 | |
| 355 rh.superiorfrontal_25 | rh.precentral_36 | lh.precentral_8 | rh.superiorfrontal_21 | lh.superiorfrontal_35 | rh.lateraloccipital_22 | rh.inferiorparietal_23 | |
| 356 rh.precentral_27 | rh.inferiorparietal_16 | lh.transversetemporal_2 | rh.middletemporal_16 | lh.precentral_28 | lh.caudalanteriorcingulate_1 | rh.lateralorbitofrontal_6 | |
| 357 rh.fusiform_15 | lh.paracentral_10 | rh.caudalmiddlefrontal_8 | rh.precentral_8 | lh.parstriangularis_6 | rh.insula_12 | rh.superiorfrontal_4 | |
| 358 lh.precentral_32 | rh.inferiorparietal_1 | rh.lateraloccipital_18 | rh.lateraloccipital_10 | rh.postcentral_13 | lh.inferiorparietal_7 | rh.bankssts_1 | |
| 359 lh.precentral_28 | lh.caudalmiddlefrontal_8 | lh.postcentral_27 | lh.postcentral_10 | lh.postcentral_8 | lh.precentral_3 | rh.precuneus_18 | |
| 360 lh.superiorparietal_22 | rh.parsopercularis_7 | rh.rostralmiddlefrontal_22 | rh.paracentral_4 | lh.rostralmiddlefrontal_17 | rh.superiorparietal_16 | lh.precentral_32 | |
| 361 lh.pericalcarine_6 | lh.precentral_3 | rh.rostralmiddlefrontal_27 | lh.cuneus_6 | rh.bankssts_3 | lh.inferiorparietal_13 | rh.precentral_23 | |
| 362 rh.superiorfrontal_13 | lh.rostralmiddlefrontal_26 | lh.precentral_3 | rh.precuneus_23 | lh.precuneus_18 | rh.lateralorbitofrontal_16 | lh.paracentral_11 | |
| 363 rh.supramarginal_3 | rh.caudalmiddlefrontal_9 | rh.inferiorparietal_6 | rh.inferiorparietal_6 | lh.superiorfrontal_43 | rh.superiorfrontal_23 | rh.paracentral_7 | |
| 364 rh.middletemporal_16 | rh.superiorfrontal_19 | lh.fusiform_15 | lh.inferiorparietal_3 | lh.superiorparietal_18 | rh.precentral_22 | rh.superiorparietal_2 | |
| 365 lh.postcentral_18 | lh.inferiorparietal_3 | lh.insula_7 | lh.pericalcarine_1 | lh.superiorparietal_8 | lh.paracentral_6 | rh.inferiorparietal_9 | |
| 366 rh.precuneus_18 | rh.rostralmiddlefrontal_14 | rh.rostralmiddlefrontal_14 | rh.superiorfrontal_5 | lh.rostralanteriorcingulate_5 | lh.lingual_17 | lh.supramarginal_3 | |

Left-Amygdala (row 333, column 3)

| | | | | | | |
|---|---|---|---|---|---|---|
| 367 Right-Amygdala | rh.superiorfrontal_14 | rh.precentral_33 | rh.superiorparietal_25 | lh.inferiorparietal_13 | rh.posteriorcingulate_2 | rh.precuneus_8 |
| 368 lh.insula_7 | rh.supramarginal_9 | rh.paracentral_4 | rh.superiorfrontal_26 | lh.lingual_15 | lh.caudalmiddlefrontal_13 | lh.precuneus_3 |
| 369 lh.precentral_12 | rh.transversetemporal_3 | lh.precuneus_8 | rh.insula_5 | rh.superiorparietal_12 | rh.isthmuscingulate_6 | lh.precentral_28 |
| 370 lh.lateraloccipital_21 | rh.superiorfrontal_26 | lh.inferiorparietal_1 | lh.postcentral_5 | rh.inferiortemporal_13 | lh.precentral_8 | lh.insula_6 |
| 371 lh.lateraloccipital_5 | lh.fusiform_10 | rh.inferiorparietal_15 | rh.precuneus_9 | rh.superiorfrontal_34 | rh.paracentral_7 | lh.insula_1 |
| 372 rh.rostralanteriorcingulate_3 | lh.fusiform_16 | lh.supramarginal_21 | rh.superiorfrontal_6 | rh.parsopercularis_8 | rh.precuneus_7 | lh.rostralmiddlefrontal_8 |
| 373 lh.precentral_35 | lh.superiorparietal_23 | lh.supramarginal_7 | rh.inferiorparietal_10 | lh.insula_10 | rh.inferiorparietal_17 | rh.precentral_24 |
| 374 lh.parahippocampal_5 | lh.fusiform_12 | lh.parsopercularis_7 | lh.caudalmiddlefrontal_11 | lh.insula_7 | lh.lingual_15 | lh.inferiorparietal_4 |
| 375 rh.superiorfrontal_2 | lh.rostralmiddlefrontal_5 | rh.pericalcarine_3 | rh.superiorfrontal_33 | lh.lingual_7 | rh.precentral_15 | rh.paracentral_8 |
| 376 rh.precentral_8 | rh.precuneus_11 | rh.superiorparietal_25 | rh.precuneus_11 | rh.precentral_13 | rh.superiorparietal_28 | Left-Amygdala |
| 377 rh.precuneus_2 | lh.parsopercularis_4 | rh.parsopercularis_7 | rh.rostralmiddlefrontal_20 | lh.precuneus_18 | lh.rostralmiddlefrontal_6 | rh.precentral_11 |
| 378 rh.paracentral_7 | rh.insula_8 | rh.rostralanteriorcingulate_3 | lh.supramarginal_18 | rh.rostralanteriorcingulate_3 | rh.supramarginal_12 | rh.rostralanteriorcingulate_1 |
| 379 rh.superiorparietal_2 | lh.inferiorparietal_3 | lh.rostralmiddlefrontal_26 | lh.caudalmiddlefrontal_2 | rh.middletemporal_16 | rh.temporalpole_2 | lh.supramarginal_10 |
| 380 rh.precuneus_21 | lh.isthmuscingulate_1 | rh.superiorparietal_2 | lh.superiorfrontal_25 | lh.inferiorparietal_6 | lh.posteriorcingulate_9 | lh.precentral_25 |
| 381 rh.precuneus_11 | lh.postcentral_5 | lh.cuneus_2 | lh.lateraloccipital_20 | rh.superiorparietal_16 | lh.paracentral_5 | rh.supramarginal_6 |
| 382 lh.postcentral_12 | rh.insula_4 | rh.inferiorparietal_19 | lh.bankssts_6 | rh.superiorparietal_2 | lh.superiorfrontal_4 | lh.insula_2 |
| 383 lh.inferiorparietal_3 | rh.precuneus_2 | lh.fusiform_10 | lh.precuneus_17 | rh.precentral_35 | rh.medialorbitofrontal_8 | rh.superiorfrontal_18 |
| 384 lh.inferiorparietal_13 | lh.bankssts_6 | rh.superiorfrontal_33 | rh.rostralmiddlefrontal_14 | lh.supramarginal_12 | rh.lateralorbitofrontal_13 | lh.insula_7 |
| 385 rh.superiorparietal_1 | lh.paracentral_1 | rh.precuneus_13 | lh.lateraloccipital_9 | lh.superiorparietal_21 | rh.precentral_4 | rh.supramarginal_4 |
| 386 lh.parstriangularis_6 | rh.paracentral_4 | rh.supramarginal_9 | lh.precuneus_2 | lh.supramarginal_19 | rh.inferiorparietal_14 | rh.precuneus_23 |
| 387 lh.inferiortemporal_16 | rh.precentral_14 | lh.precentral_11 | lh.superiorfrontal_19 | lh.postcentral_9 | rh.superiorfrontal_5 | rh.supramarginal_12 |
| 388 lh.bankssts_6 | rh.lateraloccipital_16 | lh.postcentral_21 | rh.lingual_16 | lh.lingual_2 | lh.supramarginal_10 | lh.inferiorparietal_9 |
| 389 rh.superiorfrontal_26 | lh.rostralmiddlefrontal_22 | rh.rostralanteriorcingulate_2 | rh.postcentral_11 | lh.rostralmiddlefrontal_6 | lh.inferiorparietal_18 | rh.supramarginal_7 |
| 390 rh.precentral_32 | rh.postcentral_14 | rh.paracentral_12 | rh.cuneus_5 | rh.precentral_27 | lh.precentral_21 | rh.superiorfrontal_23 |
| 391 rh.inferiorparietal_7 | lh.fusiform_15 | lh.fusiform_12 | rh.inferiorparietal_25 | lh.precentral_5 | lh.precentral_7 | rh.superiorfrontal_41 |
| 392 rh.middletemporal_4 | lh.lateraloccipital_11 | rh.transversetemporal_3 | lh.superiorparietal_17 | rh.paracentral_10 | lh.postcentral_8 | rh.isthmuscingulate_6 |
| 393 lh.lingual_17 | rh.rostralmiddlefrontal_22 | rh.superiorfrontal_5 | lh.precentral_36 | rh.paracentral_4 | rh.caudalanteriorcingulate_2 | lh.superiorfrontal_37 |
| 394 rh.supramarginal_20 | rh.superiorfrontal_16 | lh.cuneus_4 | rh.rostralmiddlefrontal_27 | rh.parsopercularis_4 | lh.insula_14 | lh.insula_5 |
| 395 rh.inferiorparietal_26 | lh.cuneus_2 | rh.superiorparietal_5 | rh.rostralmiddlefrontal_4 | rh.superiorparietal_26 | rh.lingual_17 | rh.supramarginal_5 |
| 396 lh.rostralmiddlefrontal_17 | rh.lateraloccipital_19 | rh.superiorfrontal_12 | rh.superiorfrontal_14 | lh.inferiorparietal_10 | lh.insula_11 | rh.medialorbitofrontal_4 |
| 397 rh.postcentral_7 | lh.inferiorparietal_22 | rh.supramarginal_3 | rh.middletemporal_15 | rh.precentral_8 | lh.inferiortemporal_15 | rh.superiorparietal_29 |
| 398 rh.superiorparietal_5 | lh.lingual_11 | lh.fusiform_16 | lh.superiorfrontal_6 | lh.transversetemporal_2 | lh.superiortemporal_11 | Left-Accumbens-area |
| 399 rh.superiortemporal_9 | lh.inferiorparietal_7 | rh.paracentral_5 | rh.precentral_30 | rh.inferiorparietal_7 | rh.paracentral_12 | rh.paracentral_10 |
| 400 rh.precentral_35 | rh.postcentral_11 | lh.precuneus_17 | rh.rostralmiddlefrontal_5 | rh.lingual_15 | Left-Putamen | lh.precuneus_2 |
| 401 rh.fusiform_14 | lh.postcentral_30 | lh.precentral_27 | rh.pericalcarine_3 | rh.supramarginal_20 | lh.superiorfrontal_2 | rh.precentral_26 |
| 402 lh.supramarginal_6 | rh.parstriangularis_8 | lh.precuneus_16 | lh.caudalmiddlefrontal_6 | lh.rostralmiddlefrontal_21 | lh.inferiorparietal_8 | rh.lateralorbitofrontal_15 |
| 403 lh.superiorparietal_18 | rh.paracentral_12 | lh.fusiform_12 | rh.cuneus_8 | lh.middletemporal_6 | rh.lateralorbitofrontal_5 | rh.middletemporal_16 |
| 404 rh.rostralmiddlefrontal_2 | rh.supramarginal_3 | lh.parsopercularis_1 | rh.superiorparietal_5 | rh.rostralmiddlefrontal_10 | rh.precentral_24 | rh.middletemporal_2 |
| 405 rh.postcentral_2 | rh.postcentral_9 | lh.parstriangularis_2 | rh.supramarginal_9 | lh.supramarginal_6 | rh.superiorparietal_17 | rh.rostralmiddlefrontal_11 |
| 406 lh.paracentral_10 | rh.superiorparietal_5 | rh.inferiorparietal_3 | lh.superiorfrontal_4 | lh.fusiform_11 | rh.precentral_9 | rh.insula_5 |
| 407 lh.superiortemporal_20 | rh.isthmuscingulate_4 | lh.precentral_28 | rh.superiorfrontal_40 | rh.superiortemporal_21 | rh.inferiortemporal_5 | rh.precentral_2 |
| 408 rh.paracentral_5 | lh.superiorfrontal_34 | lh.posteriorcingulate_7 | rh.caudalmiddlefrontal_9 | rh.rostralmiddlefrontal_7 | rh.inferiorparietal_14 | rh.lateraloccipital_3 |
| 409 lh.superiorfrontal_3 | lh.precentral_12 | lh.superiorfrontal_25 | lh.parstriangularis_2 | lh.paracentral_9 | rh.caudalmiddlefrontal_5 | lh.precuneus_1 |
| 410 lh.supramarginal_13 | lh.caudalanteriorcingulate_5 | lh.rostralmiddlefrontal_23 | lh.precentral_19 | lh.precentral_3 | rh.paracentral_2 | rh.supramarginal_20 |
| 411 rh.parsopercularis_8 | rh.postcentral_8 | lh.caudalmiddlefrontal_11 | lh.lingual_5 | lh.paracentral_1 | lh.supramarginal_15 | lh.precentral_20 |
| 412 rh.postcentral_11 | lh.precentral_13 | lh.precentral_30 | lh.lateraloccipital_2 | rh.lingual_5 | lh.precentral_13 | lh.supramarginal_6 |

| | | | | | | |
|---|---|---|---|---|---|---|
| 413 rh.superiorparietal_12 | lh.precentral_35 | lh.postcentral_5 | lh.rostralmiddlefrontal_10 | rh.postcentral_2 | rh.posteriorcingulate_6 | rh.precentral_19 |
| 414 lh.superiorfrontal_26 | lh.inferiorparietal_9 | lh.inferiorparietal_4 | rh.inferiorparietal_9 | lh.superiorparietal_12 | lh.superiorparietal_28 | rh.superiorfrontal_16 |
| 415 lh.rostralanteriorcingulate_5 | rh.precuneus_5 | lh.rostralmiddlefrontal_5 | rh.parsopercularis_7 | rh.precentral_14 | rh.supramarginal_17 | rh.precuneus_1 |
| 416 rh.rostralanteriorcingulate_4 | rh.caudalanteriorcingulate_4 | rh.inferiorparietal_16 | lh.precentral_12 | rh.rostralmiddlefrontal_3 | rh.precuneus_18 | rh.rostralmiddlefrontal_15 |
| 417 rh.lateraloccipital_7 | lh.posteriorcingulate_3 | rh.parstriangularis_8 | rh.supramarginal_3 | rh.caudalmiddlefrontal_12 | rh.caudalmiddlefrontal_12 | rh.rostralmiddlefrontal_6 |
| 418 rh.lateraloccipital_22 | lh.precentral_11 | rh.precentral_16 | rh.inferiorparietal_19 | lh.precentral_19 | lh.paracentral_1 | lh.precentral_16 |
| 419 lh.supramarginal_20 | lh.caudalmiddlefrontal_7 | lh.caudalmiddlefrontal_6 | rh.inferiorparietal_1 | lh.superiorparietal_4 | rh.parstriangularis_4 | lh.precentral_16 |
| 420 lh.cuneus_1 | rh.parsopercularis_5 | lh.insula_12 | rh.precentral_3 | rh.superiorfrontal_35 | rh.precuneus_18 | rh.middletemporal_12 |
| 421 rh.postcentral_5 | lh.postcentral_13 | lh.superiorfrontal_34 | lh.precentral_35 | rh.supramarginal_1 | rh.rostralmiddlefrontal_12 | rh.inferiorparietal_5 |
| 422 rh.fusiform_17 | lh.superiorfrontal_1 | rh.superiortemporal_21 | rh.precentral_24 | rh.rostralmiddlefrontal_4 | rh.superiorparietal_20 | rh.postcentral_2 |
| 423 lh.postcentral_5 | rh.precentral_32 | rh.caudalmiddlefrontal_9 | lh.superiorfrontal_30 | lh.superiorfrontal_26 | rh.superiorfrontal_30 | rh.precentral_30 |
| 424 rh.parsopercularis_1 | lh.middletemporal_15 | lh.precentral_5 | lh.cuneus_4 | lh.superiorparietal_1 | rh.supramarginal_2 | lh.superiorfrontal_7 |
| 425 rh.caudalanteriorcingulate_6 | lh.cuneus_6 | rh.superiorfrontal_21 | lh.precentral_35 | rh.postcentral_11 | rh.precentral_19 | lh.postcentral_30 |
| 426 lh.inferiorparietal_10 | rh.precentral_35 | rh.paracentral_2 | lh.caudalmiddlefrontal_8 | lh.lateraloccipital_5 | rh.precuneus_23 | rh.precuneus_13 |
| 427 rh.superiorparietal_21 | rh.caudalmiddlefrontal_7 | rh.supramarginal_18 | rh.precentral_32 | rh.caudalanteriorcingulate_6 | rh.precuneus_12 | rh.precentral_10 |
| 428 rh.bankssts_4 | lh.superiorparietal_10 | rh.precentral_32 | lh.caudalanteriorcingulate_4 | rh.precentral_24 | rh.supramarginal_21 | Right-Accumbens-area |
| 429 rh.superiortemporal_16 | lh.parsopercularis_1 | rh.precuneus_9 | rh.precentral_5 | lh.precuneus_5 | lh.rostralmiddlefrontal_24 | rh.superiorparietal_28 |
| 430 lh.rostralmiddlefrontal_6 | lh.superiorfrontal_18 | rh.middletemporal_16 | lh.parsopercularis_7 | lh.parahippocampal_5 | lh.rostralmiddlefrontal_16 | lh.rostralmiddlefrontal_16 |
| 431 lh.postcentral_14 | lh.superiorfrontal_25 | lh.bankssts_6 | lh.inferiorparietal_9 | rh.inferiortemporal_15 | rh.superiorfrontal_22 | lh.parsopercularis_2 |
| 432 lh.middletemporal_1 | rh.precentral_24 | rh.superiorfrontal_14 | lh.inferiorparietal_16 | rh.precentral_32 | lh.caudalmiddlefrontal_6 | lh.supramarginal_17 |
| 433 rh.supramarginal_16 | rh.inferiorparietal_2 | lh.superiorparietal_18 | lh.inferiortemporal_3 | lh.precuneus_2 | rh.precuneus_11 | rh.lateraloccipital_5 |
| 434 rh.superiorparietal_16 | lh.precentral_5 | rh.isthmuscingulate_2 | rh.fusiform_8 | rh.parstriangularis_6 | rh.precuneus_15 | lh.supramarginal_4 |
| 435 rh.pericalcarine_1 | lh.rostralmiddlefrontal_4 | rh.fusiform_15 | rh.postcentral_14 | rh.rostralmiddlefrontal_15 | lh.lateraloccipital_7 | lh.superiorfrontal_45 |
| 436 rh.lingual_7 | lh.superiorparietal_18 | lh.postcentral_15 | rh.superiorparietal_1 | lh.inferiorparietal_3 | rh.inferiorparietal_1 | rh.supramarginal_13 |
| 437 lh.superiorfrontal_35 | lh.supramarginal_6 | rh.superiorfrontal_16 | rh.postcentral_8 | lh.precuneus_16 | rh.inferiorparietal_9 | rh.rostralanteriorcingulate_1 |
| 438 lh.fusiform_16 | lh.precuneus_10 | lh.postcentral_30 | lh.superiorfrontal_33 | rh.superiorfrontal_25 | rh.middletemporal_19 | rh.precentral_22 |
| 439 rh.parahippocampal_2 | lh.precentral_32 | lh.supramarginal_5 | rh.lateraloccipital_14 | lh.supramarginal_13 | lh.medialorbitofrontal_2 | rh.precentral_21 |
| 440 rh.rostralmiddlefrontal_7 | lh.postcentral_18 | lh.superiorparietal_24 | lh.superiorparietal_18 | rh.rostralanteriorcingulate_4 | rh.superiorparietal_7 | rh.caudalanteriorcingulate_6 |
| 441 rh.superiorparietal_26 | rh.precuneus_18 | rh.precentral_36 | lh.supramarginal_6 | lh.precentral_35 | lh.fusiform_11 | lh.superiorfrontal_26 |
| 442 lh.lateraloccipital_2 | rh.superiorparietal_2 | rh.postcentral_6 | rh.superiorfrontal_8 | lh.rostralmiddlefrontal_27 | rh.inferiorparietal_4 | rh.postcentral_10 |
| 443 lh.precentral_3 | rh.lateraloccipital_4 | lh.fusiform_6 | lh.postcentral_30 | lh.postcentral_12 | lh.parsopercularis_2 | rh.insula_3 |
| 444 lh.supramarginal_5 | lh.precentral_19 | lh.precentral_35 | rh.fusiform_5 | lh.parsopercularis_5 | lh.supramarginal_4 | rh.posteriorcingulate_9 |
| 445 lh.rostralmiddlefrontal_15 | lh.rostralmiddlefrontal_23 | lh.superiorfrontal_40 | lh.rostralmiddlefrontal_26 | rh.precentral_11 | rh.precentral_30 | rh.precuneus_5 |
| 446 rh.superiorfrontal_34 | rh.superiorfrontal_33 | lh.superiorfrontal_18 | lh.precuneus_2 | lh.inferiorparietal_8 | rh.rostralmiddlefrontal_22 | lh.insula_16 |
| 447 lh.insula_10 | lh.rostralmiddlefrontal_27 | lh.insula_3 | lh.superiorfrontal_36 | rh.paracentral_12 | rh.postcentral_8 | rh.precentral_8 |
| 448 rh.precentral_24 | lh.rostralmiddlefrontal_10 | rh.precentral_8 | lh.lingual_4 | lh.superiorfrontal_17 | rh.inferiorparietal_18 | rh.superiorfrontal_5 |
| 449 lh.precuneus_2 | lh.inferiortemporal_3 | rh.inferiorparietal_9 | rh.transversetemporal_3 | rh.superiorparietal_17 | rh.paracentral_5 | rh.postcentral_27 |
| 450 lh.lateraloccipital_3 | rh.superiorfrontal_6 | lh.cuneus_6 | lh.caudalanteriorcingulate_5 | rh.supramarginal_4 | rh.caudalanteriorcingulate_3 | rh.posteriorcingulate_4 |
| 451 rh.paracentral_4 | lh.superiorfrontal_21 | rh.inferiorparietal_10 | rh.precentral_11 | lh.bankssts_6 | rh.lateraloccipital_23 | rh.inferiortemporal_16 |
| 452 rh.postcentral_6 | lh.transversetemporal_2 | rh.insula_15 | lh.parsopercularis_4 | lh.lateraloccipital_21 | rh.inferiorparietal_4 | lh.lateraloccipital_18 |
| 453 rh.superiorparietal_3 | lh.superiorparietal_15 | lh.middletemporal_15 | lh.precentral_11 | lh.inferiorparietal_11 | lh.middletemporal_6 | lh.inferiorparietal_14 |
| 454 rh.bankssts_6 | rh.superiorfrontal_5 | lh.postcentral_8 | rh.inferiorparietal_3 | lh.postcentral_5 | lh.precuneus_15 | rh.paracentral_6 |
| 455 lh.superiorparietal_21 | lh.precuneus_2 | rh.precentral_24 | lh.cuneus_3 | lh.inferiortemporal_16 | lh.precentral_16 | lh.supramarginal_11 |
| 456 lh.precuneus_7 | lh.superiorfrontal_6 | rh.precuneus_2 | lh.rostralmiddlefrontal_22 | lh.precuneus_1 | rh.precentral_19 | lh.supramarginal_14 |
| 457 lh.fusiform_10 | rh.precentral_30 | lh.parsopercularis_9 | rh.superiorparietal_2 | rh.caudalmiddlefrontal_2 | rh.parstriangularis_6 | lh.precentral_11 |
| 458 rh.fusiform_2 | lh.lingual_17 | rh.lingual_16 | lh.posteriorcingulate_7 | lh.lateralorbitofrontal_5 | lh.lingual_6 | lh.superiorfrontal_43 |

| # | Col1 | Col2 | Col3 | Col4 | Col5 | Col6 | Col7 | Col8 |
|---|---|---|---|---|---|---|---|---|
| 459 | lh.lateralorbitofrontal_10 | lh.postcentral_8 | rh.postcentral_14 | lh.inferiorparietal_7 | lh.precentral_17 | rh.posteriorcingulate_3 | lh.postcentral_31 |
| 460 | lh.superiorfrontal_21 | rh.superiorparietal_15 | lh.superiorfrontal_9 | rh.precuneus_18 | lh.precuneus_7 | lh.insula_17 | rh.postcentral_11 |
| 461 | lh.rostralanteriorcingulate_3 | lh.rostralanteriorcingulate_5 | lh.lingual_5 | lh.caudalmiddlefrontal_7 | rh.superiorfrontal_19 | rh.paracentral_7 | lh.superiorfrontal_28 |
| 462 | lh.inferiortemporal_3 | rh.cuneus_6 | lh.supramarginal_6 | lh.postcentral_18 | rh.precuneus_21 | lh.supramarginal_6 | lh.bankssts_4 |
| 463 | lh.precentral_13 | rh.precentral_34 | lh.caudalmiddlefrontal_2 | lh.precentral_32 | rh.precuneus_2 | rh.lateraloccipital_13 | lh.precentral_7 |
| 464 | rh.middletemporal_15 | lh.superiorfrontal_19 | lh.inferiorparietal_9 | lh.isthmuscingulate_1 | rh.precuneus_14 | rh.precentral_7 | rh.superiorfrontal_8 |
| 465 | rh.precentral_11 | lh.superiorparietal_22 | rh.parsopercularis_5 | rh.superiorfrontal_15 | lh.postcentral_14 | rh.fusiform_16 | lh.inferiorparietal_16 |
| 466 | rh.supramarginal_4 | rh.precentral_11 | rh.caudalmiddlefrontal_3 | rh.postcentral_9 | Left-Amygdala | lh.rostralmiddlefrontal_11 | rh.supramarginal_9 |
| 467 | rh.parsopercularis_4 | rh.caudalanteriorcingulate_1 | rh.cuneus_8 | lh.precentral_5 | lh.supramarginal_5 | rh.precuneus_21 | rh.precuneus_7 |
| 468 | lh.precentral_29 | lh.rostralmiddlefrontal_15 | lh.superiorfrontal_19 | lh.postcentral_13 | rh.fusiform_15 | rh.precentral_16 | lh.isthmuscingulate_1 |
| 469 | lh.paracentral_9 | rh.middletemporal_4 | rh.inferiorparietal_1 | rh.parsopercularis_5 | rh.postcentral_6 | rh.rostralmiddlefrontal_15 | rh.precentral_27 |
| 470 | lh.inferiorparietal_19 | rh.superiortemporal_21 | lh.caudalmiddlefrontal_7 | rh.rostralmiddlefrontal_22 | rh.precuneus_13 | rh.parsopercularis_8 | rh.supramarginal_19 |
| 471 | lh.precuneus_1 | rh.precentral_5 | lh.superiorfrontal_32 | lh.fusiform_1 | rh.lateraloccipital_7 | Left-Caudate | lh.bankssts_2 |
| 472 | lh.inferiortemporal_12 | rh.cuneus_8 | lh.lingual_4 | lh.rostralanteriorcingulate_5 | rh.rostralmiddlefrontal_3 | lh.superiorfrontal_17 | rh.precentral_5 |
| 473 | rh.lateraloccipital_6 | rh.superiorfrontal_34 | lh.superiorfrontal_35 | rh.caudalmiddlefrontal_7 | lh.supramarginal_10 | rh.posteriorcingulate_9 | lh.superiorfrontal_36 |
| 474 | lh.superiorfrontal_32 | lh.supramarginal_13 | rh.rostralmiddlefrontal_2 | rh.lingual_14 | lh.superiorparietal_28 | lh.lateralorbitofrontal_11 | rh.superiorfrontal_33 |
| 475 | lh.caudalmiddlefrontal_12 | rh.precentral_8 | lh.lateraloccipital_22 | lh.inferiortemporal_12 | lh.rostralanteriorcingulate_1 | rh.superiorparietal_21 | rh.inferiorparietal_10 |
| 476 | lh.superiorparietal_26 | lh.inferiortemporal_12 | lh.precentral_2 | lh.precuneus_16 | lh.rostralmiddlefrontal_7 | lh.posteriorcingulate_2 | lh.superiorparietal_2 |
| 477 | lh.precuneus_16 | lh.lateraloccipital_9 | rh.superiorparietal_12 | rh.insula_8 | rh.superiorfrontal_37 | lh.precuneus_3 | lh.insula_17 |
| 478 | rh.precuneus_22 | lh.supramarginal_5 | lh.supramarginal_13 | lh.lingual_17 | lh.rostralanteriorcingulate_4 | lh.insula_2 | rh.parstriangularis_4 |
| 479 | rh.precuneus_13 | rh.caudalmiddlefrontal_4 | rh.cuneus_5 | lh.precentral_13 | rh.supramarginal_16 | lh.superiorfrontal_28 | lh.inferiorparietal_17 |
| 480 | rh.lingual_3 | rh.insula_7 | rh.inferiorparietal_2 | lh.parsopercularis_1 | rh.bankssts_6 | lh.inferiorparietal_10 | rh.postcentral_26 |
| 481 | lh.postcentral_3 | rh.fusiform_8 | rh.lateraloccipital_4 | lh.superiorparietal_10 | lh.lingual_17 | lh.isthmuscingulate_4 | rh.superiorparietal_7 |
| 482 | lh.inferiorparietal_8 | lh.middletemporal_1 | rh.lateralorbitofrontal_6 | lh.superiorfrontal_12 | lh.insula_1 | Right-Caudate | lh.bankssts_5 |
| 483 | lh.lingual_12 | lh.precuneus_15 | rh.postcentral_11 | lh.superiorfrontal_34 | rh.middletemporal_4 | rh.lingual_7 | rh.precentral_31 |
| 484 | lh.precuneus_15 | lh.precentral_13 | rh.precentral_29 | lh.rostralmiddlefrontal_23 | rh.middletemporal_15 | rh.lateraloccipital_7 | lh.superiortemporal_24 |
| 485 | lh.precuneus_5 | rh.precentral_13 | rh.lateraloccipital_10 | lh.inferiorparietal_22 | rh.parsopercularis_3 | lh.inferiorparietal_11 | lh.superiorparietal_14 |
| 486 | lh.lingual_6 | lh.bankssts_3 | rh.insula_2 | rh.fusiform_13 | lh.superiorparietal_26 | rh.parsopercularis_7 | Brain-Stem |
| 487 | lh.fusiform_7 | rh.postcentral_13 | lh.inferiortemporal_3 | rh.insula_4 | rh.postcentral_5 | rh.insula_5 | lh.superiorfrontal_39 |
| 488 | lh.middletemporal_10 | lh.inferiorparietal_4 | rh.lingual_7 | lh.precuneus_10 | lh.superiorfrontal_24 | rh.superiorparietal_3 | lh.postcentral_21 |
| 489 | rh.superiorfrontal_35 | lh.lingual_5 | lh.superiorfrontal_13 | rh.isthmuscingulate_2 | lh.superiorfrontal_3 | rh.precentral_25 | lh.superiorfrontal_15 |
| 490 | rh.superiorparietal_28 | lh.supramarginal_19 | lh.superiorfrontal_10 | rh.inferiortemporal_14 | lh.superiorfrontal_14 | rh.caudalmiddlefrontal_7 | lh.supramarginal_7 |
| 491 | lh.superiorparietal_4 | rh.superiorfrontal_35 | rh.precentral_32 | lh.superiorfrontal_5 | lh.supramarginal_20 | rh.lateralorbitofrontal_6 | rh.middletemporal_13 |
| 492 | lh.lateraloccipital_7 | rh.lateraloccipital_10 | rh.middletemporal_4 | lh.supramarginal_5 | lh.precentral_29 | rh.insula_4 | lh.superiorparietal_8 |
| 493 | lh.middletemporal_14 | rh.fusiform_13 | rh.cuneus_7 | rh.precentral_21 | rh.lateraloccipital_6 | lh.parstriangularis_1 | rh.superiorfrontal_15 |
| 494 | rh.lingual_17 | rh.lingual_7 | rh.rostralanteriorcingulate_4 | rh.superiorfrontal_35 | Right-Accumbens-area | Left-Pallidum | rh.postcentral_6 |
| 495 | lh.paracentral_4 | rh.inferiorparietal_4 | lh.lateraloccipital_2 | rh.precentral_13 | rh.lateraloccipital_2 | lh.insula_8 | rh.superiorparietal_26 |
| 496 | lh.superiorparietal_12 | lh.superiorfrontal_4 | lh.lateraloccipital_8 | lh.posteriorcingulate_3 | lh.middletemporal_1 | lh.superiorfrontal_45 | rh.supramarginal_1 |
| 497 | lh.inferiorparietal_11 | lh.superiorparietal_19 | lh.rostralmiddlefrontal_10 | lh.superiorfrontal_9 | rh.middletemporal_4 | lh.caudalmiddlefrontal_11 | rh.superiorfrontal_10 |
| 498 | lh.lateraloccipital_22 | lh.superiorfrontal_17 | lh.superiorfrontal_33 | lh.superiortemporal_9 | lh.superiorfrontal_17 | rh.parsopercularis_5 | lh.precentral_5 |
| 499 | rh.precuneus_12 | rh.postcentral_2 | lh.precentral_12 | rh.fusiform_14 | lh.superiortemporal_20 | rh.superiorfrontal_12 | rh.rostralanteriorcingulate_4 |
| 500 | lh.parsopercularis_5 | lh.lateraloccipital_20 | lh.superiorparietal_17 | lh.supramarginal_13 | lh.inferiorparietal_19 | lh.insula_7 | lh.superiorfrontal_5 |
| 501 | lh.superiorfrontal_23 | rh.parstriangularis_3 | lh.superiorfrontal_30 | rh.caudalanteriorcingulate_4 | lh.middletemporal_10 | rh.paracentral_11 | lh.paracentral_8 |
| 502 | lh.fusiform_6 | rh.inferiortemporal_14 | rh.superiorfrontal_20 | rh.parstriangularis_8 | lh.superiorfrontal_10 | rh.middletemporal_18 | rh.lateraloccipital_6 |
| 503 | lh.superiortemporal_25 | rh.superiorfrontal_36 | rh.precentral_3 | lh.superiorparietal_22 | rh.superiorfrontal_10 | lh.lingual_12 | rh.lateralorbitofrontal_8 |
| 504 | Left-Amygdala | rh.inferiorparietal_11 | lh.precuneus_2 | rh.precuneus_5 | rh.superiorfrontal_16 | lh.caudalanteriorcingulate_2 | lh.medialorbitofrontal_5 |

| | | | | | | |
|---|---|---|---|---|---|---|
| 505 lh.supramarginal_10 | lh.supramarginal_10 | rh.postcentral_3 | lh.middletemporal_1 | rh.superiorparietal_21 | rh.caudalmiddlefrontal_4 | rh.lateralorbitofrontal_14 |
| 506 lh.postcentral_19 | rh.cuneus_5 | lh.insula_17 | lh.postcentral_15 | rh.superiortemporal_16 | lh.postcentral_28 | rh.bankssts_5 |
| 507 lh.transversetemporal_3 | lh.parsopercularis_9 | rh.superiorfrontal_35 | lh.superiorfrontal_10 | lh.lateralorbitofrontal_10 | rh.parsopercularis_3 | lh.supramarginal_13 |
| 508 rh.parstriangularis_6 | lh.parahippocampal_5 | rh.lateraloccipital_20 | lh.superiorfrontal_15 | lh.pericalcarine_6 | lh.caudalmiddlefrontal_7 | lh.supramarginal_8 |
| 509 rh.caudalmiddlefrontal_2 | lh.superiorfrontal_5 | lh.supramarginal_10 | lh.supramarginal_10 | rh.lateraloccipital_22 | rh.insula_8 | lh.postcentral_19 |
| 510 rh.superiorparietal_17 | rh.superiorparietal_21 | rh.superiorparietal_26 | lh.superiorparietal_15 | rh.precentral_25 | lh.parsorbitalis_1 | rh.middletemporal_15 |
| 511 rh.precuneus_14 | rh.rostralmiddlefrontal_2 | rh.rostralmiddlefrontal_13 | rh.precentral_8 | lh.superiorfrontal_37 | rh.supramarginal_13 | lh.precentral_14 |
| 512 rh.fusiform_13 | rh.precentral_21 | lh.postcentral_18 | lh.rostralmiddlefrontal_21 | rh.rostralmiddlefrontal_13 | rh.superiorfrontal_41 | lh.middletemporal_10 |
| 513 rh.superiortemporal_11 | lh.pericalcarine_1 | lh.superiortemporal_9 | lh.superiorfrontal_32 | rh.superiorfrontal_18 | lh.superiorparietal_5 | rh.precentral_29 |
| 514 lh.rostralmiddlefrontal_21 | lh.rostralmiddlefrontal_21 | rh.inferiorparietal_11 | lh.transversetemporal_2 | lh.lateralorbitofrontal_8 | lh.parsopercularis_10 | rh.postcentral_9 |
| 515 lh.superiorfrontal_17 | lh.middletemporal_10 | rh.caudalmiddlefrontal_4 | rh.lingual_7 | rh.insula_3 | rh.inferiortemporal_9 | rh.superiorfrontal_30 |
| 516 lh.precentral_19 | lh.inferiorparietal_10 | lh.insula_10 | rh.paracentral_6 | rh.lateraloccipital_3 | rh.middletemporal_13 | rh.precuneus_3 |
| 517 lh.rostralanteriorcingulate_1 | lh.parsopercularis_3 | lh.superiorparietal_10 | rh.middletemporal_10 | lh.postcentral_19 | rh.pericalcarine_8 | lh.inferiorparietal_6 |
| 518 lh.supramarginal_9 | rh.parahippocampal_2 | lh.rostralmiddlefrontal_21 | rh.lingual_1 | rh.inferiorparietal_5 | rh.superiortemporal_11 | lh.postcentral_5 |
| 519 lh.precuneus_19 | lh.lateraloccipital_2 | lh.superiorfrontal_17 | lh.postcentral_8 | lh.superiorfrontal_7 | rh.middletemporal_9 | lh.temporalpole_2 |
| 520 lh.lateraloccipital_17 | lh.precentral_29 | lh.fusiform_5 | rh.superiortemporal_21 | rh.rostralanteriorcingulate_2 | rh.paracentral_3 | lh.supramarginal_5 |
| 521 rh.precentral_14 | lh.postcentral_9 | rh.postcentral_9 | rh.cuneus_1 | lh.inferiortemporal_12 | rh.medialorbitofrontal_6 | rh.precentral_21 |
| 522 rh.paracentral_12 | lh.superiorparietal_17 | rh.precentral_11 | lh.superiorfrontal_13 | lh.parsopercularis_10 | rh.insula_13 | lh.inferiortemporal_15 |
| 523 Right-Accumbens-area | lh.cuneus_4 | rh.rostralmiddlefrontal_4 | lh.superiorfrontal_12 | lh.supramarginal_9 | rh.precentral_34 | rh.superiorfrontal_6 |
| 524 lh.fusiform_12 | rh.pericalcarine_3 | lh.inferiortemporal_12 | rh.superiorparietal_15 | lh.lateraloccipital_17 | rh.superiorfrontal_24 | rh.superiorfrontal_4 |
| 525 lh.rostralmiddlefrontal_10 | rh.superiorparietal_26 | rh.precuneus_14 | lh.pericalcarine_7 | rh.middletemporal_13 | rh.insula_7 | rh.precentral_32 |
| 526 lh.insula_1 | lh.precentral_20 | lh.superiorparietal_4 | rh.superiorfrontal_20 | rh.bankssts_4 | lh.temporalpole_2 | lh.supramarginal_9 |
| 527 lh.fusiform_8 | rh.supramarginal_14 | lh.superiorfrontal_15 | rh.rostralmiddlefrontal_13 | lh.precuneus_19 | lh.bankssts_5 | lh.lateraloccipital_17 |
| 528 rh.rostralmiddlefrontal_4 | rh.fusiform_5 | lh.inferiorparietal_22 | lh.fusiform_5 | lh.precuneus_15 | lh.precentral_30 | lh.superiorfrontal_34 |
| 529 rh.supramarginal_2 | lh.precuneus_7 | rh.insula_12 | lh.precentral_29 | lh.inferiorparietal_3 | rh.middletemporal_16 | rh.lateraloccipital_17 |
| 530 rh.parahippocampal_3 | rh.superiorparietal_1 | lh.superiorparietal_16 | rh.precentral_34 | Left-Accumbens-area | rh.precuneus_13 | lh.precuneus_19 |
| 531 rh.middletemporal_13 | lh.fusiform_1 | rh.precentral_13 | lh.inferiorparietal_4 | rh.rostralanteriorcingulate_1 | rh.insula_14 | rh.bankssts_6 |
| 532 lh.precentral_14 | lh.postcentral_14 | lh.supramarginal_10 | lh.lateraloccipital_14 | lh.supramarginal_11 | rh.rostralmiddlefrontal_16 | lh.precentral_24 |
| 533 lh.fusiform_17 | Right-Amygdala | lh.superiorfrontal_4 | rh.superiorparietal_21 | rh.insula_2 | lh.inferiorparietal_6 | lh.postcentral_13 |
| 534 lh.middletemporal_3 | lh.inferiorparietal_13 | rh.supramarginal_14 | rh.postcentral_2 | lh.precentral_14 | rh.caudalmiddlefrontal_2 | rh.superiorparietal_8 |
| 535 rh.superiorfrontal_16 | rh.posteriorcingulate_4 | rh.lateraloccipital_14 | rh.rostralmiddlefrontal_15 | lh.superiorfrontal_45 | lh.postcentral_31 | lh.inferiorparietal_19 |
| 536 lh.middletemporal_13 | lh.superiorfrontal_10 | lh.middletemporal_1 | lh.supramarginal_19 | lh.fusiform_16 | rh.paracentral_4 | lh.supramarginal_1 |
| 537 lh.insula_3 | rh.cuneus_1 | lh.parsopercularis_3 | rh.parahippocampal_2 | rh.parahippocampal_2 | lh.rostralmiddlefrontal_2 | rh.postcentral_14 |
| 538 lh.superiorfrontal_25 | lh.superiorfrontal_13 | lh.inferiorparietal_13 | rh.cuneus_7 | rh.fusiform_2 | lh.supramarginal_12 | rh.superiorfrontal_9 |
| 539 rh.inferiortemporal_14 | lh.superiorfrontal_8 | lh.precentral_22 | rh.postcentral_3 | lh.precentral_24 | lh.superiortemporal_7 | rh.temporalpole_2 |
| 540 lh.pericalcarine_5 | rh.precuneus_15 | lh.inferiorparietal_8 | lh.inferiorparietal_10 | rh.superiorparietal_3 | lh.entorhinal_3 | rh.parstriangularis_7 |
| 541 rh.superiorfrontal_3 | lh.lateraloccipital_22 | lh.superiorfrontal_6 | rh.rostralmiddlefrontal_2 | rh.precentral_30 | rh.superiorfrontal_7 | rh.superiorfrontal_30 |
| 542 rh.inferiorparietal_5 | lh.fusiform_5 | rh.parahippocampal_2 | lh.precentral_2 | lh.lateralorbitofrontal_1 | lh.posteriorcingulate_8 | lh.postcentral_14 |
| 543 rh.lateraloccipital_23 | lh.postcentral_12 | lh.bankssts_5 | lh.precuneus_7 | lh.superiorfrontal_28 | Right-Amygdala | rh.superiorfrontal_32 |
| 544 lh.superiorfrontal_7 | rh.lateraloccipital_14 | lh.superiorfrontal_36 | lh.precuneus_21 | rh.postcentral_7 | rh.insula_6 | rh.precentral_6 |
| 545 lh.precentral_24 | rh.paracentral_6 | rh.inferiorparietal_4 | rh.superiorparietal_26 | rh.superiorfrontal_6 | rh.superiorfrontal_33 | rh.postcentral_9 |
| 546 rh.rostralmiddlefrontal_3 | lh.superiorfrontal_30 | lh.precuneus_10 | rh.inferiorparietal_2 | rh.lateraloccipital_17 | lh.lateraloccipital_22 | rh.superiorfrontal_1 |
| 547 lh.superiorfrontal_37 | rh.lingual_1 | lh.precentral_13 | lh.parahippocampal_5 | rh.superiorparietal_29 | rh.superiorfrontal_21 | lh.middletemporal_12 |
| 548 lh.precentral_17 | rh.superiorparietal_16 | rh.precentral_21 | rh.caudalmiddlefrontal_4 | lh.fusiform_10 | rh.superiorfrontal_18 | lh.inferiortemporal_7 |
| 549 lh.cuneus_6 | rh.precuneus_14 | rh.superiorfrontal_8 | rh.lingual_2 | rh.precentral_5 | rh.superiorfrontal_34 | lh.middletemporal_6 |
| 550 rh.superiorfrontal_18 | rh.insula_14 | rh.precentral_18 | lh.lateraloccipital_22 | lh.middletemporal_3 | rh.caudalanteriorcingulate_5 | lh.insula_12 |

| # | Col1 | Col2 | Col3 | Col4 | Col5 | Col6 | Col7 | Col8 |
|---|---|---|---|---|---|---|---|---|
| 551 | rh.rostralanteriorcingulate_1 | rh.inferiorparietal_7 | rh.precentral_34 | lh.paracentral_8 | lh.superiorfrontal_1 | lh.superiorfrontal_6 | lh.superiorfrontal_21 | |
| 552 | rh.lateraloccipital_17 | rh.lingual_16 | rh.posteriorcingulate_4 | rh.inferiorparietal_4 | lh.middletemporal_13 | lh.superiorfrontal_13 | rh.precentral_18 | |
| 553 | lh.rostralanteriorcingulate_4 | lh.superiorfrontal_15 | lh.inferiorparietal_7 | rh.precuneus_15 | rh.precentral_21 | rh.inferiortemporal_16 | rh.inferiortemporal_13 | |
| 554 | lh.postcentral_25 | lh.middletemporal_14 | lh.caudalmiddlefrontal_7 | lh.postcentral_14 | rh.precuneus_12 | lh.caudalmiddlefrontal_2 | lh.lateraloccipital_23 | |
| 555 | lh.supramarginal_11 | rh.bankssts_5 | rh.rostralmiddlefrontal_17 | rh.precentral_14 | lh.lateraloccipital_22 | lh.postcentral_7 | rh.superiorfrontal_38 | |
| 556 | lh.superiorparietal_5 | lh.superiorparietal_4 | rh.postcentral_2 | lh.lingual_9 | rh.parstriangularis_7 | rh.parsopercularis_1 | rh.bankssts_2 | |
| 557 | lh.fusiform_3 | lh.insula_12 | rh.paracentral_6 | rh.precentral_18 | lh.lateraloccipital_7 | lh.parsopercularis_3 | lh.rostralmiddlefrontal_2 | |
| 558 | rh.parstriangularis_7 | rh.lingual_2 | rh.precentral_5 | lh.superiorfrontal_7 | lh.superiorfrontal_19 | lh.paracentral_2 | lh.precuneus_15 | |
| 559 | rh.parsopercularis_3 | rh.rostralmiddlefrontal_13 | rh.supramarginal_4 | lh.cuneus_5 | rh.inferiortemporal_14 | lh.superiorfrontal_39 | lh.parahippocampal_5 | |
| 560 | rh.superiorfrontal_37 | rh.precuneus_13 | lh.postcentral_9 | rh.postcentral_13 | lh.postcentral_25 | rh.superiorfrontal_28 | lh.lateraloccipital_5 | |
| 561 | rh.lingual_2 | rh.superiorparietal_12 | lh.superiorfrontal_14 | rh.lateraloccipital_13 | lh.medialorbitofrontal_4 | lh.precentral_28 | rh.precuneus_17 | |
| 562 | rh.fusiform_4 | rh.superiorfrontal_15 | lh.superiorparietal_15 | lh.bankssts_3 | lh.cuneus_1 | rh.postcentral_10 | lh.inferiorparietal_3 | |
| 563 | rh.precentral_21 | lh.bankssts_5 | lh.precuneus_7 | lh.precentral_22 | rh.fusiform_17 | rh.precuneus_17 | rh.precentral_16 | |
| 564 | rh.fusiform_9 | lh.inferiorparietal_19 | lh.parahippocampal_5 | lh.lingual_3 | rh.superiorfrontal_33 | rh.posteriorcingulate_1 | lh.paracentral_2 | |
| 565 | lh.superiorfrontal_45 | lh.superiorfrontal_40 | lh.paracentral_10 | lh.superiorparietal_19 | rh.fusiform_13 | lh.fusiform_13 | rh.middletemporal_18 | |
| 566 | rh.rostralmiddlefrontal_27 | rh.supramarginal_4 | rh.superiorparietal_21 | lh.caudalanteriorcingulate_1 | rh.parahippocampal_3 | lh.superiorfrontal_30 | lh.postcentral_25 | |
| 567 | rh.cuneus_1 | lh.superiortemporal_9 | lh.pericalcarine_1 | lh.middletemporal_14 | rh.superiortemporal_11 | lh.supramarginal_19 | lh.fusiform_18 | |
| 568 | rh.superiortemporal_6 | lh.superiorfrontal_33 | lh.lateralorbitofrontal_10 | lh.superiorparietal_4 | rh.paracentral_6 | lh.posteriorcingulate_5 | rh.superiortemporal_14 | |
| 569 | lh.rostralmiddlefrontal_7 | rh.postcentral_3 | lh.supramarginal_19 | rh.lateraloccipital_12 | lh.parsopercularis_2 | rh.insula_15 | rh.precentral_30 | |
| 570 | lh.superiorfrontal_13 | lh.postcentral_15 | rh.lingual_1 | rh.precentral_31 | lh.superiorfrontal_6 | lh.superiorparietal_9 | rh.inferiortemporal_5 | |
| 571 | Left-Accumbens-area | lh.superiorfrontal_9 | rh.postcentral_13 | rh.fusiform_6 | lh.middletemporal_14 | rh.postcentral_28 | rh.bankssts_3 | |
| 572 | rh.fusiform_12 | lh.inferiorparietal_8 | rh.superiorparietal_1 | lh.precentral_10 | lh.fusiform_6 | lh.lateraloccipital_11 | lh.superiorfrontal_30 | |
| 573 | rh.precentral_4 | Left-Amygdala | rh.fusiform_8 | rh.parsopercularis_9 | lh.inferiorparietal_14 | lh.caudalanteriorcingulate_3 | rh.lateraloccipital_7 | |
| 574 | rh.precentral_25 | lh.lingual_9 | lh.inferiorparietal_19 | rh.precuneus_13 | lh.superiorparietal_28 | lh.superiorfrontal_36 | rh.supramarginal_16 | |
| 575 | lh.superiorfrontal_24 | rh.insula_15 | lh.supramarginal_12 | lh.parsopercularis_3 | lh.precentral_7 | rh.superiorparietal_4 | rh.inferiortemporal_14 | |
| 576 | lh.fusiform_5 | lh.cuneus_3 | lh.pericalcarine_7 | rh.pericalcarine_8 | lh.supramarginal_8 | rh.postcentral_30 | lh.postcentral_12 | |
| 577 | rh.precentral_5 | rh.caudalanteriorcingulate_6 | rh.precentral_25 | Right-Amygdala | rh.supramarginal_2 | lh.fusiform_18 | lh.rostralmiddlefrontal_24 | |
| 578 | lh.superiorfrontal_28 | lh.precuneus_5 | rh.precentral_10 | lh.inferiorparietal_13 | lh.fusiform_7 | lh.lateralorbitofrontal_15 | rh.supramarginal_15 | |
| 579 | rh.lateraloccipital_15 | lh.superiorfrontal_5 | lh.inferiorparietal_10 | lh.superiorparietal_9 | rh.precuneus_22 | lh.postcentral_2 | lh.middletemporal_13 | |
| 580 | rh.inferiortemporal_6 | lh.supramarginal_12 | lh.precentral_20 | rh.superiorparietal_16 | lh.transversetemporal_3 | rh.inferiortemporal_13 | lh.precuneus_15 | |
| 581 | lh.parsopercularis_10 | rh.precentral_10 | rh.insula_3 | lh.postcentral_9 | lh.postcentral_3 | lh.paracentral_8 | lh.bankssts_3 | |
| 582 | lh.middletemporal_7 | rh.precuneus_21 | lh.insula_1 | rh.postcentral_20 | lh.superiortemporal_25 | rh.superiorparietal_6 | lh.rostralmiddlefrontal_12 | |
| 583 | rh.precentral_30 | rh.insula_9 | lh.inferiorparietal_11 | rh.parstriangularis_3 | lh.superiorfrontal_30 | lh.superiorfrontal_37 | lh.middletemporal_3 | |
| 584 | rh.lateraloccipital_11 | lh.precentral_17 | rh.lingual_14 | rh.supramarginal_14 | lh.pericalcarine_1 | rh.insula_3 | rh.inferiortemporal_12 | |
| 585 | rh.fusiform_8 | lh.inferiorparietal_11 | lh.cuneus_3 | lh.precentral_17 | lh.superiorfrontal_8 | rh.paracentral_6 | lh.bankssts_6 | |
| 586 | lh.superiorfrontal_6 | rh.inferiorparietal_26 | rh.fusiform_13 | lh.inferiorparietal_19 | rh.lingual_17 | rh.precentral_36 | lh.inferiortemporal_12 | |
| 587 | rh.transversetemporal_2 | rh.superiortemporal_16 | rh.superiorparietal_15 | rh.insula_7 | rh.superiorfrontal_5 | lh.middletemporal_12 | lh.lateralorbitofrontal_10 | |
| 588 | lh.rostralmiddlefrontal_24 | rh.precentral_18 | lh.fusiform_5 | rh.precentral_25 | rh.precentral_4 | lh.postcentral_13 | lh.precentral_10 | |
| 589 | lh.inferiortemporal_4 | lh.precuneus_1 | rh.superiorfrontal_9 | lh.superiorfrontal_14 | lh.middletemporal_7 | lh.superiorfrontal_27 | lh.superiorfrontal_33 | |
| 590 | lh.middletemporal_16 | rh.precentral_25 | lh.superiorfrontal_15 | rh.supramarginal_4 | lh.fusiform_17 | rh.precentral_11 | lh.middletemporal_7 | |
| 591 | rh.superiorfrontal_19 | rh.rostralanteriorcingulate_3 | rh.precentral_35 | rh.inferiorparietal_11 | lh.rostralmiddlefrontal_24 | lh.superiorfrontal_26 | rh.superiorfrontal_40 | |
| 592 | rh.pericalcarine_5 | rh.rostralmiddlefrontal_7 | lh.middletemporal_14 | rh.pericalcarine_6 | rh.lateraloccipital_11 | lh.posteriorcingulate_4 | rh.superiorparietal_20 | |
| 593 | rh.superiorfrontal_33 | rh.superiorparietal_28 | lh.precuneus_21 | lh.postcentral_12 | lh.lingual_12 | lh.parsorbitalis_2 | lh.supramarginal_12 | |
| 594 | rh.rostralanteriorcingulate_2 | rh.supramarginal_20 | lh.insula_11 | Left-Amygdala | rh.precentral_6 | rh.superiorparietal_23 | lh.precuneus_14 | |
| 595 | lh.supramarginal_8 | lh.lingual_4 | lh.superiorparietal_22 | lh.precentral_20 | rh.fusiform_4 | lh.parsorbitalis_3 | lh.precentral_35 | |
| 596 | lh.pericalcarine_3 | lh.parstriangularis_6 | lh.precuneus_5 | rh.superiorfrontal_9 | rh.fusiform_9 | rh.precentral_32 | rh.lateraloccipital_11 | |

| | | | | | | |
|---|---|---|---|---|---|---|
| 597 lh.inferiortemporal_10 | lh.superiorparietal_21 | lh.postcentral_8 | lh.inferiorparietal_8 | lh.fusiform_8 | rh.superiorfrontal_17 | rh.temporalpole_3 |
| 598 rh.superiorparietal_29 | rh.superiorfrontal_32 | lh.bankssts_3 | lh.inferiortemporal_6 | lh.lingual_6 | lh.lateralorbitofrontal_6 | lh.postcentral_8 |
| 599 rh.superiorparietal_6 | rh.inferiortemporal_6 | rh.cuneus_1 | rh.postcentral_25 | lh.superiorparietal_8 | lh.supramarginal_9 | lh.superiortemporal_18 |
| 600 rh.superiorfrontal_30 | lh.precentral_10 | lh.cuneus_5 | lh.lateraloccipital_15 | lh.superiorfrontal_11 | rh.inferiortemporal_14 | lh.postcentral_23 |
| 601 rh.pericalcarine_2 | lh.precentral_2 | lh.rostralmiddlefrontal_6 | lh.precuneus_5 | rh.lingual_7 | rh.lateraloccipital_21 | lh.superiorparietal_1 |
| 602 rh.lingual_1 | lh.precentral_22 | lh.paracentral_8 | lh.fusiform_4 | lh.superiorfrontal_36 | rh.supramarginal_3 | rh.lateraloccipital_2 |
| 603 lh.superiorparietal_7 | lh.superiorfrontal_12 | lh.superiorparietal_19 | rh.lateraloccipital_21 | lh.rostralmiddlefrontal_14 | rh.medialorbitofrontal_3 | rh.postcentral_3 |
| 604 lh.lingual_9 | rh.rostralmiddlefrontal_16 | rh.inferiortemporal_14 | rh.inferiorparietal_7 | rh.insula_12 | lh.parsopercularis_5 | rh.inferiortemporal_15 |
| 605 lh.inferiortemporal_1 | rh.precentral_31 | lh.postcentral_14 | rh.superiorparietal_12 | rh.superiorparietal_7 | lh.precentral_2 | rh.postcentral_13 |
| 606 rh.paracentral_6 | lh.insula_7 | rh.cuneus_3 | rh.cuneus_4 | rh.precentral_31 | lh.postcentral_25 | lh.inferiortemporal_16 |
| 607 lh.postcentral_26 | lh.superiorfrontal_14 | lh.fusiform_1 | rh.precentral_2 | rh.lateraloccipital_15 | rh.superiorfrontal_34 | lh.superiorfrontal_18 |
| 608 lh.fusiform_15 | lh.middletemporal_13 | rh.parahippocampal_3 | lh.inferiorparietal_11 | lh.superiorfrontal_4 | rh.superiorfrontal_31 | rh.medialorbitofrontal_11 |
| 609 rh.precentral_6 | lh.rostralmiddlefrontal_17 | lh.lingual_9 | lh.superiortemporal_9 | rh.lateraloccipital_23 | lh.superiorfrontal_40 | lh.precentral_22 |
| 610 lh.lateraloccipital_16 | rh.parsopercularis_1 | lh.rostralanteriorcingulate_3 | rh.caudalanteriorcingulate_6 | lh.inferiortemporal_6 | rh.insula_3 | lh.postcentral_2 |
| 611 lh.lingual_14 | lh.rostralanteriorcingulate_2 | lh.superiorparietal_9 | lh.inferiortemporal_4 | rh.superiortemporal_6 | rh.superiorfrontal_13 | rh.precentral_4 |
| 612 rh.superiorfrontal_6 | lh.cuneus_5 | lh.parstriangularis_6 | lh.precuneus_1 | rh.inferiortemporal_10 | rh.precentral_28 | rh.lingual_6 |
| 613 rh.supramarginal_18 | rh.bankssts_6 | rh.inferiorparietal_26 | rh.superiorparietal_28 | rh.fusiform_12 | rh.supramarginal_20 | rh.superiortemporal_21 |
| 614 lh.parsopercularis_2 | lh.superiorparietal_5 | rh.insula_14 | rh.postcentral_3 | lh.fusiform_3 | rh.temporalpole_3 | rh.medialorbitofrontal_10 |
| 615 lh.superiorfrontal_10 | lh.rostralmiddlefrontal_6 | lh.postcentral_25 | lh.middletemporal_13 | rh.superiorfrontal_38 | rh.middletemporal_12 | rh.postcentral_5 |
| 616 lh.superiorfrontal_1 | lh.superiorparietal_12 | rh.supramarginal_20 | lh.superiortemporal_28 | lh.superiorfrontal_15 | lh.supramarginal_5 | rh.inferiortemporal_9 |
| 617 lh.precentral_7 | lh.precuneus_19 | lh.postcentral_2 | rh.cuneus_3 | lh.precentral_10 | lh.superiorfrontal_43 | lh.parstriangularis_1 |
| 618 lh.inferiorparietal_14 | lh.inferiortemporal_4 | rh.precentral_10 | lh.superiorparietal_5 | lh.middletemporal_16 | lh.parstriangularis_3 | lh.lateraloccipital_21 |
| 619 lh.fusiform_9 | rh.lateralorbitofrontal_6 | lh.supramarginal_20 | rh.superiortemporal_16 | lh.lateraloccipital_16 | rh.precentral_18 | rh.precentral_2 |
| 620 lh.parahippocampal_2 | lh.precentral_24 | rh.inferiorparietal_7 | rh.posteriorcingulate_4 | rh.superiorparietal_14 | rh.superiorfrontal_35 | lh.lingual_7 |
| 621 rh.postcentral_3 | lh.superiorparietal_9 | rh.paracentral_9 | lh.precentral_7 | rh.lingual_3 | rh.precentral_35 | lh.lateraloccipital_3 |
| 622 rh.fusiform_3 | lh.superiorfrontal_32 | lh.superiorparietal_21 | lh.rostralmiddlefrontal_24 | rh.precentral_18 | rh.posteriorcingulate_5 | lh.middletemporal_1 |
| 623 lh.rostralmiddlefrontal_3 | lh.rostralmiddlefrontal_24 | lh.rostralmiddlefrontal_24 | rh.rostralmiddlefrontal_7 | lh.superiorfrontal_20 | lh.inferiortemporal_7 | rh.precentral_28 |
| 624 lh.fusiform_2 | lh.superiorfrontal_12 | lh.superiorparietal_12 | rh.superiorfrontal_29 | lh.fusiform_5 | lh.superiorfrontal_20 | lh.inferiortemporal_3 |
| 625 rh.rostralmiddlefrontal_13 | rh.cuneus_7 | lh.middletemporal_10 | rh.precuneus_21 | lh.postcentral_2 | rh.posteriorcingulate_7 | lh.supramarginal_20 |
| 626 rh.postcentral_16 | rh.lingual_14 | lh.fusiform_2 | rh.precentral_24 | rh.transversetemporal_2 | rh.inferiorparietal_10 | rh.precuneus_21 |
| 627 rh.pericalcarine_7 | rh.superiorfrontal_20 | rh.precentral_6 | rh.superiorparietal_21 | rh.superiorfrontal_21 | rh.superiorfrontal_42 | rh.lingual_5 |
| 628 lh.superiorfrontal_19 | rh.pericalcarine_7 | lh.precuneus_1 | rh.inferiorparietal_26 | lh.superiorparietal_7 | lh.medialorbitofrontal_8 | rh.transversetemporal_3 |
| 629 lh.lingual_4 | lh.caudalmiddlefrontal_12 | rh.rostralmiddlefrontal_16 | rh.paracentral_1 | rh.fusiform_12 | rh.postcentral_9 | rh.lateraloccipital_19 |
| 630 rh.superiorfrontal_14 | lh.precentral_7 | rh.pericalcarine_8 | lh.superiorparietal_14 | lh.parahippocampal_2 | rh.superiortemporal_1 | rh.precentral_27 |
| 631 lh.lingual_10 | lh.middletemporal_3 | lh.precentral_17 | lh.precuneus_19 | lh.fusiform_15 | rh.supramarginal_9 | rh.inferiortemporal_3 |
| 632 rh.cuneus_8 | lh.lateraloccipital_14 | lh.postcentral_12 | rh.bankssts_5 | lh.pericalcarine_5 | lh.inferiortemporal_16 | rh.middletemporal_19 |
| 633 rh.lateraloccipital_8 | rh.superiorfrontal_7 | rh.parstriangularis_3 | lh.insula_12 | rh.lateraloccipital_19 | lh.superiorfrontal_25 | rh.parahippocampal_3 |
| 634 rh.lateraloccipital_19 | lh.superiorparietal_28 | rh.lingual_2 | lh.superiorparietal_12 | lh.superiorfrontal_34 | rh.lateraloccipital_11 | rh.rostralmiddlefrontal_1 |
| 635 rh.insula_2 | rh.lateraloccipital_13 | lh.rostralmiddlefrontal_7 | rh.supramarginal_20 | rh.inferiortemporal_12 | lh.insula_16 | lh.precentral_34 |
| 636 lh.lateralorbitofrontal_8 | rh.parsopercularis_4 | lh.superiorfrontal_39 | lh.bankssts_5 | rh.superiorparietal_5 | rh.lingual_10 | lh.fusiform_11 |
| 637 lh.precentral_10 | rh.bankssts_4 | rh.paracentral_1 | lh.superiorparietal_29 | rh.supramarginal_18 | lh.superiorfrontal_31 | rh.precuneus_2 |
| 638 lh.postcentral_2 | lh.supramarginal_11 | rh.precuneus_15 | rh.rostralanteriorcingulate_3 | lh.precentral_2 | lh.supramarginal_20 | rh.paracentral_1 |
| 639 lh.pericalcarine_4 | lh.fusiform_4 | lh.superiorfrontal_12 | rh.bankssts_6 | lh.insula_3 | lh.rostralmiddlefrontal_17 | rh.precentral_17 |
| 640 rh.superiorfrontal_5 | lh.supramarginal_20 | rh.precuneus_21 | lh.fusiform_2 | lh.postcentral_26 | lh.superiorfrontal_39 | rh.superiorparietal_21 |
| 641 rh.inferiortemporal_7 | rh.paracentral_8 | lh.pericalcarine_6 | rh.precentral_1 | lh.inferiortemporal_1 | rh.postcentral_19 | lh.rostralmiddlefrontal_18 |
| 642 lh.superiorparietal_8 | rh.superiorparietal_17 | lh.middletemporal_13 | lh.parstriangularis_6 | lh.inferiortemporal_4 | lh.bankssts_2 | rh.superiorfrontal_20 |

| # | Col1 | Col2 | Col3 | Col4 | Col5 | Col6 | Col7 | Col8 |
|---|---|---|---|---|---|---|---|---|
| 643 | lh.superiorparietal_23 | rh.precentral_6 | lh.lingual_12 | rh.precentral_10 | lh.rostralmiddlefrontal_18 | lh.posteriorcingulate_3 | rh.superiortemporal_16 | |
| 644 | rh.superiorfrontal_38 | lh.fusiform_3 | rh.superiorparietal_14 | lh.cuneus_7 | rh.inferiortemporal_10 | rh.fusiform_10 | lh.lateraloccipital_16 | |
| 645 | rh.postcentral_1 | rh.precentral_2 | rh.bankssts_6 | rh.precentral_6 | rh.pericalcarine_2 | rh.pericalcarine_4 | rh.parsorbitalis_3 | |
| 646 | rh.inferiortemporal_10 | lh.fusiform_2 | lh.parahippocampal_2 | rh.insula_9 | lh.superiorfrontal_9 | lh.supramarginal_5 | lh.superiorfrontal_12 | |
| 647 | rh.superiortemporal_25 | rh.rostralmiddlefrontal_15 | rh.precentral_2 | rh.parsopercularis_1 | rh.parstriangularis_4 | rh.fusiform_1 | rh.inferiortemporal_10 | |
| 648 | lh.rostralmiddlefrontal_18 | lh.fusiform_8 | rh.precentral_7 | lh.supramarginal_12 | lh.superiorfrontal_18 | lh.insula_6 | rh.superiorfrontal_36 | |
| 649 | lh.lateraloccipital_4 | rh.lateraloccipital_12 | lh.lingual_3 | lh.middletemporal_3 | rh.fusiform_8 | lh.superiortemporal_2 | lh.lingual_15 | |
| 650 | lh.lateraloccipital_11 | lh.supramarginal_8 | lh.inferiorparietal_14 | rh.rostralmiddlefrontal_17 | rh.rostralmiddlefrontal_8 | rh.lateraloccipital_18 | lh.postcentral_28 | |
| 651 | lh.transversetemporal_1 | rh.fusiform_6 | rh.parsopercularis_4 | lh.rostralmiddlefrontal_6 | lh.precentral_22 | lh.inferiorparietal_19 | rh.rostralmiddlefrontal_5 | |
| 652 | rh.lingual_12 | lh.superiorparietal_14 | lh.insula_6 | lh.superiorfrontal_32 | lh.cuneus_6 | rh.lingual_12 | rh.precuneus_12 | |
| 653 | lh.superiorfrontal_36 | rh.pericalcarine_8 | rh.supramarginal_2 | lh.supramarginal_11 | rh.precuneus_17 | rh.superiorfrontal_7 | lh.inferiortemporal_10 | |
| 654 | lh.pericalcarine_2 | rh.postcentral_5 | rh.superiorparietal_17 | lh.supramarginal_8 | lh.superiorfrontal_16 | rh.precentral_3 | lh.parahippocampal_2 | |
| 655 | rh.precentral_18 | lh.superiortemporal_20 | rh.precuneus_17 | lh.insula_7 | lh.superiorparietal_6 | lh.supramarginal_17 | lh.postcentral_15 | |
| 656 | rh.superiorparietal_23 | rh.superiorfrontal_9 | lh.superiortemporal_20 | rh.superiorparietal_7 | lh.pericalcarine_4 | rh.paracentral_1 | lh.lateraloccipital_22 | |
| 657 | rh.lateraloccipital_4 | lh.lingual_12 | lh.caudalmiddlefrontal_12 | rh.rostralmiddlefrontal_16 | rh.rostralmiddlefrontal_12 | rh.cuneus_1 | lh.lingual_17 | |
| 658 | rh.inferiortemporal_12 | lh.lateraloccipital_15 | rh.parsopercularis_10 | lh.caudalmiddlefrontal_12 | lh.insula_6 | rh.precentral_26 | rh.middletemporal_4 | |
| 659 | lh.lateralorbitofrontal_5 | rh.inferiortemporal_10 | lh.superiorparietal_7 | lh.fusiform_3 | lh.lateraloccipital_8 | lh.bankssts_3 | lh.precuneus_18 | |
| 660 | rh.superiorfrontal_21 | rh.parsopercularis_8 | rh.bankssts_5 | rh.precentral_17 | rh.lateraloccipital_4 | lh.precentral_32 | rh.supramarginal_2 | |
| 661 | rh.precentral_31 | rh.caudalmiddlefrontal_2 | rh.parsopercularis_1 | lh.superiorfrontal_1 | rh.superiorfrontal_1 | lh.supramarginal_11 | rh.fusiform_13 | |
| 662 | rh.superiorparietal_7 | lh.postcentral_3 | Left-Accumbens-area | lh.inferiortemporal_10 | rh.postcentral_1 | lh.superiorparietal_19 | rh.lateralorbitofrontal_3 | |
| 663 | rh.lateraloccipital_18 | lh.middletemporal_7 | lh.superiorfrontal_32 | lh.lateraloccipital_12 | rh.lateraloccipital_8 | lh.precentral_5 | rh.postcentral_19 | |
| 664 | lh.superiorparietal_28 | rh.superiorfrontal_3 | lh.superiorfrontal_5 | rh.superiorparietal_27 | rh.superiorfrontal_8 | lh.superiorparietal_26 | lh.superiortemporal_20 | |
| 665 | lh.lateraloccipital_8 | rh.postcentral_25 | rh.postcentral_5 | rh.lateraloccipital_11 | lh.entorhinal_1 | rh.fusiform_9 | | |
| 666 | lh.fusiform_1 | lh.rostralanteriorcingulate_3 | rh.superiortemporal_16 | lh.rostralanteriorcingulate_2 | rh.fusiform_3 | lh.superiorparietal_8 | rh.superiortemporal_11 | |
| 667 | lh.superiortemporal_8 | rh.supramarginal_2 | lh.superiorparietal_28 | rh.superiorparietal_17 | rh.medialorbitofrontal_10 | lh.lateraloccipital_2 | rh.fusiform_2 | |
| 668 | lh.lingual_5 | lh.postcentral_2 | lh.precuneus_19 | lh.lateraloccipital_10 | lh.superiorparietal_23 | lh.superiorfrontal_3 | lh.lingual_2 | |
| 669 | rh.postcentral_18 | rh.superiorparietal_7 | rh.precentral_31 | lh.fusiform_8 | rh.supramarginal_15 | lh.superiortemporal_20 | lh.lateraloccipital_9 | |
| 670 | lh.medialorbitofrontal_4 | rh.postcentral_20 | lh.inferiortemporal_4 | rh.parsopercularis_4 | lh.lateraloccipital_9 | rh.superiortemporal_20 | rh.superiortemporal_10 | |
| 671 | lh.superiorfrontal_18 | rh.supramarginal_16 | rh.superiorparietal_20 | rh.insula_15 | lh.fusiform_2 | lh.inferiortemporal_6 | lh.inferiortemporal_6 | |
| 672 | rh.lateraloccipital_15 | rh.lateraloccipital_21 | rh.transversetemporal_3 | lh.middletemporal_7 | rh.postcentral_18 | rh.entorhinal_2 | rh.lateraloccipital_7 | |
| 673 | lh.precentral_22 | rh.cuneus_3 | lh.fusiform_8 | rh.lateralorbitofrontal_6 | rh.lateraloccipital_4 | lh.postcentral_9 | lh.fusiform_15 | |
| 674 | lh.parahippocampal_4 | rh.postcentral_7 | lh.superiorparietal_8 | lh.precentral_34 | rh.postcentral_16 | rh.temporalpole_1 | rh.lateraloccipital_22 | |
| 675 | lh.superiorfrontal_34 | lh.lateralorbitofrontal_10 | rh.parsopercularis_8 | lh.parahippocampal_2 | lh.postcentral_23 | rh.postcentral_3 | rh.temporalpole_1 | |
| 676 | lh.lingual_8 | lh.parahippocampal_2 | lh.superiorparietal_7 | lh.superiortemporal_12 | rh.pericalcarine_3 | lh.insula_9 | rh.fusiform_4 | |
| 677 | lh.rostralmiddlefrontal_16 | lh.rostralmiddlefrontal_7 | rh.cuneus_4 | lh.postcentral_11 | rh.lingual_2 | rh.lateraloccipital_16 | rh.lateraloccipital_15 | |
| 678 | lh.superiorparietal_14 | lh.superiortemporal_25 | rh.rostralmiddlefrontal_7 | rh.superiorfrontal_3 | lh.rostralmiddlefrontal_16 | rh.fusiform_15 | rh.supramarginal_18 | |
| 679 | lh.lateraloccipital_9 | rh.precuneus_22 | lh.middletemporal_3 | rh.bankssts_4 | rh.inferiortemporal_7 | rh.superiorparietal_1 | rh.superiorparietal_7 | |
| 680 | rh.middletemporal_8 | rh.parahippocampal_3 | rh.rostralmiddlefrontal_11 | lh.precentral_1 | lh.medialorbitofrontal_5 | rh.caudalmiddlefrontal_3 | lh.fusiform_13 | |
| 681 | rh.superiortemporal_24 | rh.pericalcarine_6 | rh.rostralanteriorcingulate_2 | lh.lateraloccipital_1 | rh.lingual_4 | rh.posteriorcingulate_8 | rh.fusiform_14 | |
| 682 | lh.insula_9 | rh.cuneus_4 | lh.supramarginal_11 | rh.lingual_8 | rh.superiortemporal_25 | rh.postcentral_12 | lh.lateraloccipital_8 | |
| 683 | rh.postcentral_12 | lh.postcentral_2 | rh.fusiform_6 | lh.lingual_13 | rh.insula_9 | rh.superiorfrontal_9 | lh.transversetemporal_2 | |
| 684 | rh.precuneus_17 | rh.superiorparietal_6 | rh.superiorparietal_27 | lh.postcentral_25 | rh.rostralmiddlefrontal_11 | lh.pericalcarine_5 | rh.inferiortemporal_6 | |
| 685 | rh.fusiform_5 | lh.lingual_3 | Right-Accumbens-area | lh.postcentral_2 | lh.lingual_14 | Left-Amygdala | rh.fusiform_15 | |
| 686 | lh.lateralorbitofrontal_1 | lh.inferiorparietal_14 | rh.supramarginal_16 | rh.caudalmiddlefrontal_2 | rh.rostralmiddlefrontal_6 | lh.postcentral_10 | lh.superiortemporal_9 | |
| 687 | lh.rostralmiddlefrontal_12 | rh.middletemporal_13 | lh.postcentral_26 | rh.rostralmiddlefrontal_15 | lh.rostralmiddlefrontal_1 | lh.insula_10 | rh.precentral_1 | |
| 688 | rh.medialorbitofrontal_10 | lh.superiorparietal_29 | rh.inferiortemporal_6 | rh.supramarginal_2 | lh.fusiform_9 | lh.precentral_10 | rh.bankssts_4 | |

| | | | | | | |
|---|---|---|---|---|---|---|
| 689 lh.superiorfrontal_4 | rh.precuneus_12 | rh.lateraloccipital_13 | lh.lingual_12 | lh.superiorparietal_23 | lh.precentral_22 | lh.insula_9 |
| 690 lh.postcentral_17 | lh.superiorparietal_8 | rh.superiorfrontal_7 | lh.supramarginal_20 | lh.transversetemporal_1 | lh.precuneus_14 | rh.lingual_15 |
| 691 lh.lateraloccipital_19 | lh.superiorparietal_26 | rh.rostralanteriorcingulate_1 | rh.inferiorparietal_14 | lh.paracentral_2 | rh.precentral_17 | lh.superiortemporal_14 |
| 692 rh.postcentral_4 | lh.parsopercularis_10 | rh.superiorparietal_28 | lh.rostralmiddlefrontal_7 | lh.superiorfrontal_5 | rh.superiorfrontal_26 | rh.lateraloccipital_4 |
| 693 lh.superiortemporal_26 | rh.parsopercularis_5 | rh.superiorparietal_6 | rh.superiorparietal_6 | rh.lateraloccipital_18 | rh.supramarginal_4 | rh.medialorbitofrontal_7 |
| 694 rh.lingual_10 | rh.precentral_1 | rh.caudalmiddlefrontal_2 | lh.postcentral_3 | rh.superiorfrontal_15 | lh.lateraloccipital_16 | rh.superiorfrontal_7 |
| 695 lh.lingual_1 | lh.superiorparietal_29 | lh.parsopercularis_5 | lh.supramarginal_14 | rh.pericalcarine_7 | rh.precentral_34 | rh.parsorbitalis_1 |
| 696 lh.inferiortemporal_8 | rh.precentral_4 | rh.rostralmiddlefrontal_15 | rh.paracentral_9 | rh.middletemporal_8 | lh.supramarginal_7 | lh.middletemporal_16 |
| 697 rh.parstriangularis_4 | rh.superiorparietal_20 | lh.precentral_9 | lh.superiortemporal_8 | rh.pericalcarine_2 | lh.precentral_31 | lh.superiorfrontal_12 |
| 698 rh.superiorfrontal_1 | rh.paracentral_1 | lh.lingual_13 | rh.superiortemporal_19 | rh.superiortemporal_24 | lh.lateraloccipital_17 | lh.rostralmiddlefrontal_16 |
| 699 rh.supramarginal_15 | rh.insula_3 | lh.lateraloccipital_15 | lh.rostralmiddlefrontal_16 | lh.lingual_10 | lh.bankssts_4 | rh.parahippocampal_1 |
| 700 rh.superiorparietal_1 | lh.rostralmiddlefrontal_16 | rh.rostralanteriorcingulate_4 | rh.superiorparietal_29 | lh.postcentral_15 | lh.cuneus_4 | lh.postcentral_24 |
| 701 rh.lateraloccipital_9 | rh.inferiorparietal_5 | rh.precentral_17 | lh.postcentral_5 | rh.lateralorbitofrontal_14 | lh.postcentral_24 | rh.rostralmiddlefrontal_25 |
| 702 rh.insula_12 | lh.superiorfrontal_29 | rh.precentral_1 | lh.rostralanteriorcingulate_3 | lh.lingual_9 | rh.lateraloccipital_2 | lh.fusiform_17 |
| 703 lh.postcentral_29 | rh.rostralanteriorcingulate_4 | lh.precentral_34 | rh.paracentral_2 | rh.lateraloccipital_19 | lh.lateraloccipital_3 | rh.superiorparietal_4 |
| 704 lh.superiorparietal_9 | rh.parsopercularis_3 | rh.bankssts_4 | rh.superiorparietal_19 | rh.precentral_2 | lh.superiorfrontal_40 | rh.superiorparietal_3 |
| 705 lh.inferiortemporal_9 | lh.insula_10 | lh.inferiortemporal_10 | lh.rostralmiddlefrontal_12 | rh.postcentral_19 | rh.pericalcarine_7 | lh.rostralmiddlefrontal_25 |
| 706 lh.postcentral_23 | lh.rostralmiddlefrontal_12 | rh.postcentral_3 | lh.superiortemporal_20 | rh.pericalcarine_5 | rh.postcentral_18 | lh.postcentral_26 |
| 707 lh.cuneus_2 | rh.insula_2 | rh.postcentral_20 | lh.lateraloccipital_13 | lh.superiortemporal_8 | lh.postcentral_16 | lh.fusiform_6 |
| 708 rh.postcentral_10 | lh.lingual_13 | lh.supramarginal_8 | rh.parahippocampal_3 | lh.inferiortemporal_9 | rh.inferiortemporal_3 | lh.fusiform_16 |
| 709 rh.pericalcarine_4 | lh.superiorparietal_7 | lh.rostralmiddlefrontal_12 | rh.parsopercularis_8 | lh.postcentral_24 | rh.superiortemporal_21 | rh.parahippocampal_2 |
| 710 lh.precentral_2 | lh.postcentral_26 | lh.precentral_1 | lh.precentral_9 | lh.insula_11 | rh.superiortemporal_18 | rh.superiortemporal_12 |
| 711 lh.superiorparietal_3 | lh.cuneus_7 | lh.paracentral_2 | lh.superiortemporal_25 | lh.lingual_8 | lh.superiorfrontal_18 | lh.transversetemporal_3 |
| 712 lh.cuneus_5 | rh.precentral_17 | rh.superiorparietal_19 | lh.lateralorbitofrontal_10 | lh.postcentral_29 | rh.superiortemporal_14 | rh.lateraloccipital_16 |
| 713 rh.lingual_9 | lh.rostralmiddlefrontal_18 | lh.superiorparietal_29 | rh.precuneus_22 | lh.precentral_27 | lh.precentral_29 | rh.paracentral_9 |
| 714 lh.postcentral_24 | rh.superiortemporal_25 | lh.postcentral_25 | rh.middletemporal_13 | rh.superiorfrontal_9 | rh.inferiortemporal_1 | rh.transversetemporal_2 |
| 715 rh.postcentral_19 | rh.parstriangularis_6 | lh.cuneus_7 | rh.supramarginal_16 | rh.lingual_12 | rh.precentral_11 | lh.precentral_9 |
| 716 lh.rostralmiddlefrontal_14 | lh.superiorfrontal_39 | rh.lateraloccipital_12 | lh.parsopercularis_10 | rh.lingual_10 | rh.supramarginal_16 | rh.postcentral_18 |
| 717 lh.postcentral_15 | lh.paracentral_2 | rh.precuneus_22 | rh.precentral_28 | lh.parahippocampal_4 | lh.precentral_14 | rh.lateralorbitofrontal_1 |
| 718 lh.cuneus_6 | lh.supramarginal_9 | rh.rostralmiddlefrontal_16 | rh.precentral_27 | rh.rostralmiddlefrontal_25 | rh.precentral_25 | lh.lateralorbitofrontal_3 |
| 719 rh.rostralmiddlefrontal_25 | lh.precentral_34 | lh.middletemporal_7 | lh.lingual_16 | rh.superiorfrontal_32 | rh.paracentral_9 | rh.superiortemporal_1 |
| 720 lh.insula_3 | lh.paracentral_9 | lh.lateraloccipital_14 | rh.insula_2 | lh.lingual_1 | rh.precentral_31 | lh.supramarginal_14 |
| 721 rh.lateraloccipital_16 | lh.insula_1 | lh.fusiform_4 | rh.paracentral_4 | rh.postcentral_12 | lh.insula_13 | rh.fusiform_12 |
| 722 lh.cuneus_4 | lh.inferiortemporal_1 | rh.superiorfrontal_29 | rh.precuneus_12 | lh.superiorparietal_20 | rh.postcentral_24 | rh.superiortemporal_6 |
| 723 lh.fusiform_4 | lh.superiorparietal_19 | lh.superiortemporal_11 | rh.postcentral_7 | lh.superiortemporal_26 | lh.postcentral_19 | lh.middletemporal_14 |
| 724 rh.superiorfrontal_8 | rh.precuneus_17 | lh.fusiform_3 | rh.precuneus_17 | lh.precuneus_14 | Brain-Stem | lh.inferiortemporal_1 |
| 725 rh.precentral_2 | rh.rostralmiddlefrontal_3 | rh.superiortemporal_25 | lh.rostralmiddlefrontal_18 | lh.fusiform_1 | rh.superiorfrontal_16 | lh.superiortemporal_7 |
| 726 lh.rostralmiddlefrontal_8 | rh.superiorparietal_27 | lh.postcentral_29 | rh.superiortemporal_25 | rh.lateraloccipital_16 | rh.middletemporal_10 | lh.fusiform_10 |
| 727 lh.superiorfrontal_9 | lh.superiorparietal_23 | rh.rostralanteriorcingulate_1 | lh.superiorparietal_7 | rh.lateraloccipital_9 | rh.postcentral_5 | rh.superiorparietal_23 |
| 728 lh.paracentral_2 | lh.lateraloccipital_10 | rh.lateraloccipital_21 | lh.postcentral_26 | lh.lateralorbitofrontal_13 | lh.temporalpole_3 | lh.precentral_18 |
| 729 lh.lateraloccipital_2 | lh.transversetemporal_3 | lh.supramarginal_9 | lh.insula_17 | rh.cuneus_1 | rh.precentral_10 | rh.parahippocampal_1 |
| 730 lh.superiortemporal_11 | lh.precentral_1 | rh.superiorparietal_6 | rh.rostralmiddlefrontal_11 | lh.postcentral_10 | lh.lateraloccipital_8 | rh.postcentral_1 |
| 731 lh.medialorbitofrontal_5 | rh.supramarginal_18 | lh.superiorfrontal_16 | rh.postcentral_19 | lh.lateraloccipital_2 | lh.lingual_14 | lh.middletemporal_4 |
| 732 lh.superiorfrontal_8 | rh.transversetemporal_2 | rh.precuneus_12 | rh.insula_3 | rh.lingual_1 | rh.pericalcarine_1 | lh.superiorfrontal_40 |
| 733 lh.superiortemporal_4 | lh.supramarginal_14 | lh.precentral_31 | rh.inferiorparietal_5 | lh.fusiform_5 | rh.precentral_4 | lh.inferiortemporal_9 |
| 734 rh.cuneus_5 | lh.lateraloccipital_12 | rh.parstriangularis_6 | lh.parsopercularis_5 | rh.superiorparietal_1 | rh.cuneus_6 | rh.superiorfrontal_39 |

| | | | | | | | |
|---|---|---|---|---|---|---|---|
| 735 | lh.precentral_27 | lh.precentral_27 | lh.lateraloccipital_10 | rh.insula_12 | rh.cuneus_8 | lh.lateralorbitofrontal_12 | lh.lateraloccipital_2 |
| 736 | lh.lateraloccipital_20 | lh.lateraloccipital_1 | lh.postcentral_11 | lh.postcentral_29 | rh.pericalcarine_4 | lh.supramarginal_8 | rh.postcentral_25 |
| 737 | rh.lateralorbitofrontal_12 | lh.middletemporal_16 | lh.superiortemporal_25 | lh.superiorparietal_26 | lh.superiorfrontal_30 | lh.pericalcarine_6 | lh.lateraloccipital_20 |
| 738 | lh.lingual_13 | rh.insula_12 | rh.middletemporal_13 | rh.parsopercularis_3 | lh.lateraloccipital_20 | rh.superiortemporal_17 | lh.pericalcarine_4 |
| 739 | lh.bankssts_1 | lh.postcentral_29 | lh.superiortemporal_25 | rh.superiorfrontal_23 | lh.superiorparietal_9 | rh.precentral_9 | rh.lateralorbitofrontal_2 |
| 740 | lh.insula_6 | Left-Accumbens-area | lh.precuneus_14 | rh.lingual_13 | lh.superiortemporal_4 | lh.superiorfrontal_15 | rh.fusiform_10 |
| 741 | rh.lateraloccipital_14 | rh.postcentral_19 | lh.precentral_27 | lh.insula_3 | lh.cuneus_2 | rh.bankssts_5 | lh.lateraloccipital_19 |
| 742 | rh.middletemporal_10 | rh.fusiform_4 | lh.lingual_16 | lh.superiorfrontal_16 | rh.inferiortemporal_8 | lh.superiorfrontal_6 | rh.pericalcarine_2 |
| 743 | lh.superiorfrontal_11 | rh.superiortemporal_12 | lh.rostralmiddlefrontal_14 | lh.inferiortemporal_1 | lh.paracentral_8 | rh.precentral_20 | lh.superiortemporal_25 |
| 744 | lh.superiorfrontal_5 | rh.inferiortemporal_10 | lh.superiorparietal_29 | rh.superiorfrontal_23 | rh.precentral_28 | lh.postcentral_6 | rh.lateralorbitofrontal_5 |
| 745 | rh.superiorfrontal_15 | rh.superiortemporal_6 | rh.rostralmiddlefrontal_3 | rh.middletemporal_10 | lh.bankssts_1 | rh.lateraloccipital_9 | lh.fusiform_8 |
| 746 | rh.lateraloccipital_10 | lh.precentral_9 | lh.precentral_4 | rh.supramarginal_18 | lh.superiorfrontal_33 | rh.pericalcarine_4 | lh.temporalpole_3 |
| 747 | rh.superiorparietal_19 | Right-Accumbens-area | lh.insula_9 | rh.parstriangularis_6 | rh.paracentral_9 | rh.superiorfrontal_36 | lh.lateraloccipital_11 |
| 748 | rh.pericalcarine_3 | rh.precentral_28 | lh.supramarginal_14 | lh.inferiortemporal_9 | rh.lingual_9 | lh.superiorfrontal_33 | lh.fusiform_5 |
| 749 | lh.lateralorbitofrontal_14 | rh.lingual_17 | rh.lingual_17 | rh.middletemporal_8 | lh.superiorparietal_3 | rh.superiortemporal_10 | lh.fusiform_7 |
| 750 | rh.paracentral_9 | lh.parsopercularis_2 | lh.lingual_8 | rh.inferiortemporal_10 | rh.superiortemporal_11 | lh.postcentral_11 | rh.middletemporal_8 |
| 751 | rh.superiorparietal_20 | lh.middletemporal_8 | rh.transversetemporal_2 | lh.insula_10 | rh.lateraloccipital_14 | lh.parstriangularis_7 | lh.postcentral_29 |
| 752 | rh.parahippocampal_6 | rh.rostralmiddlefrontal_11 | rh.superiorparietal_23 | lh.supramarginal_9 | lh.postcentral_17 | rh.parstriangularis_1 | lh.superiortemporal_4 |
| 753 | lh.precuneus_14 | lh.insula_3 | rh.inferiorparietal_5 | rh.transversetemporal_2 | rh.superiorfrontal_39 | rh.supramarginal_7 | rh.lateraloccipital_8 |
| 754 | rh.rostralmiddlefrontal_11 | lh.lateraloccipital_13 | lh.parsopercularis_2 | rh.rostralmiddlefrontal_3 | rh.cuneus_6 | lh.lingual_11 | rh.superiortemporal_24 |
| 755 | rh.rostralmiddlefrontal_6 | lh.postcentral_11 | lh.lateralorbitofrontal_2 | lh.precuneus_14 | rh.postcentral_4 | lh.postcentral_3 | lh.bankssts_1 |
| 756 | rh.cuneus_3 | lh.inferiortemporal_9 | rh.lateralorbitofrontal_14 | lh.superiorparietal_3 | rh.superiorfrontal_7 | rh.supramarginal_1 | rh.postcentral_20 |
| 757 | lh.lateralorbitofrontal_5 | rh.superiortemporal_11 | rh.parsopercularis_3 | rh.rostralanteriorcingulate_4 | rh.precentral_30 | lh.lingual_10 | rh.precuneus_22 |
| 758 | lh.superiorparietal_17 | lh.precentral_14 | rh.lateralorbitofrontal_12 | rh.insula_1 | rh.lateralorbitofrontal_5 | lh.lingual_1 | lh.fusiform_3 |
| 759 | lh.superiortemporal_3 | lh.superiorparietal_3 | rh.superiorfrontal_3 | lh.middletemporal_16 | rh.superiorfrontal_36 | lh.precentral_1 | rh.lingual_17 |
| 760 | rh.rostralmiddlefrontal_9 | lh.middletemporal_10 | rh.lateralorbitofrontal_5 | rh.supramarginal_15 | lh.cuneus_5 | rh.superiorparietal_27 | lh.postcentral_3 |
| 761 | rh.superiorfrontal_32 | rh.supramarginal_15 | lh.superiortemporal_8 | rh.superiortemporal_2 | rh.lateraloccipital_10 | lh.postcentral_5 | rh.lateraloccipital_23 |
| 762 | rh.superiorfrontal_9 | rh.postcentral_18 | lh.superiorparietal_3 | rh.superiorfrontal_39 | lh.precentral_34 | rh.parstriangularis_5 | lh.inferiortemporal_4 |
| 763 | rh.fusiform_7 | rh.inferiortemporal_12 | rh.postcentral_18 | rh.superiortemporal_7 | rh.middletemporal_10 | rh.bankssts_6 | rh.medialorbitofrontal_6 |
| 764 | lh.lateraloccipital_15 | rh.rostralanteriorcingulate_1 | lh.lateraloccipital_12 | Left-Accumbens-area | rh.lateralorbitofrontal_3 | rh.superiortemporal_5 | rh.lateralorbitofrontal_17 |
| 765 | rh.superiorfrontal_7 | lh.superiorfrontal_16 | lh.insula_13 | lh.transversetemporal_3 | rh.superiorfrontal_4 | rh.supramarginal_14 | rh.lateraloccipital_14 |
| 766 | lh.superiorfrontal_16 | lh.lingual_16 | lh.fusiform_4 | lh.parsopercularis_2 | lh.insula_17 | rh.superiortemporal_23 | lh.superiortemporal_2 |
| 767 | lh.postcentral_22 | rh.rostralanteriorcingulate_4 | rh.postcentral_19 | rh.fusiform_4 | rh.lateralorbitofrontal_12 | lh.superiortemporal_14 | rh.superiortemporal_18 |
| 768 | lh.lingual_11 | lh.insula_17 | rh.postcentral_7 | rh.superiortemporal_6 | lh.cuneus_4 | rh.lateraloccipital_4 | rh.lateraloccipital_13 |
| 769 | rh.precentral_28 | lh.precuneus_14 | lh.rostralmiddlefrontal_18 | lh.rostralmiddlefrontal_14 | lh.parstriangularis_1 | rh.superiortemporal_24 | rh.inferiortemporal_2 |
| 770 | rh.lateraloccipital_13 | rh.lingual_8 | lh.inferiortemporal_1 | lh.insula_11 | rh.superiorfrontal_18 | rh.cuneus_8 | rh.superiortemporal_8 |
| 771 | rh.pericalcarine_1 | lh.rostralmiddlefrontal_14 | rh.postcentral_18 | rh.lateraloccipital_20 | rh.lateralorbitofrontal_1 | rh.supramarginal_1 | rh.lingual_10 |
| 772 | rh.transversetemporal_1 | rh.postcentral_1 | lh.lateralorbitofrontal_1 | rh.inferiortemporal_12 | lh.lingual_5 | rh.parahippocampal_2 | lh.lateraloccipital_4 |
| 773 | lh.superiorfrontal_30 | lh.fusiform_9 | rh.middletemporal_10 | rh.superiorfrontal_10 | rh.precentral_17 | rh.precentral_21 | lh.entorhinal_3 |
| 774 | rh.lateralorbitofrontal_1 | rh.parstriangularis_4 | rh.parahippocampal_1 | rh.superiorparietal_19 | rh.postcentral_25 | lh.inferiortemporal_5 | lh.pericalcarine_6 |
| 775 | lh.postcentral_25 | lh.superiortemporal_8 | rh.rostralmiddlefrontal_6 | Right-Accumbens-area | rh.superiorfrontal_20 | rh.precentral_5 | lh.lingual_1 |
| 776 | rh.lateralorbitofrontal_3 | rh.lingual_13 | rh.lingual_13 | lh.precentral_30 | rh.superiorfrontal_20 | rh.middletemporal_15 | rh.rostralmiddlefrontal_12 |
| 777 | rh.cuneus_7 | lh.parstriangularis_7 | rh.superiorfrontal_39 | rh.parstriangularis_4 | rh.lateraloccipital_13 | rh.fusiform_3 | rh.postcentral_10 |
| 778 | rh.superiorfrontal_33 | rh.superiorfrontal_23 | lh.transversetemporal_1 | rh.rostralmiddlefrontal_6 | lh.precentral_9 | lh.postcentral_22 | rh.fusiform_3 |
| 779 | lh.postcentral_16 | lh.inferiortemporal_8 | lh.precentral_30 | lh.lateraloccipital_6 | lh.superiorfrontal_12 | lh.superiortemporal_8 | lh.precentral_31 |
| 780 | rh.superiorfrontal_36 | lh.insula_11 | lh.medialorbitofrontal_5 | rh.lingual_17 | lh.superiorfrontal_12 | lh.postcentral_26 | rh.lateraloccipital_9 |

| # | Col1 | Col2 | Col3 | Col4 | Col5 | Col6 | Col7 | Col8 |
|---|---|---|---|---|---|---|---|---|
| 781 | rh.pericalcarine_8 | rh.rostralmiddlefrontal_6 | rh.superiortemporal_24 | lh.superiortemporal_19 | rh.precentral_1 | rh.postcentral_15 | rh.fusiform_8 | |
| 782 | rh.lateraloccipital_12 | rh.inferiortemporal_7 | lh.lateraloccipital_13 | rh.postcentral_1 | rh.rostralmiddlefrontal_9 | rh.lingual_13 | rh.lateraloccipital_18 | |
| 783 | rh.fusiform_6 | lh.parahippocampal_4 | lh.lateralorbitofrontal_8 | lh.supramarginal_14 | lh.supramarginal_14 | rh.inferiortemporal_11 | lh.precentral_1 | |
| 784 | lh.precentral_9 | lh.rostralanteriorcingulate_1 | rh.superiortemporal_6 | lh.precentral_14 | lh.fusiform_4 | rh.middletemporal_14 | rh.inferiortemporal_7 | |
| 785 | lh.cuneus_3 | lh.precentral_30 | rh.lateraloccipital_1 | rh.fusiform_11 | lh.superiorparietal_17 | rh.cuneus_5 | rh.superiortemporal_9 | |
| 786 | rh.postcentral_17 | rh.rostralanteriorcingulate_2 | lh.fusiform_9 | rh.superiortemporal_11 | lh.lateraloccipital_15 | lh.lateraloccipital_6 | lh.superiortemporal_26 | |
| 787 | lh.precentral_34 | lh.postcentral_23 | rh.middletemporal_8 | lh.parahippocampal_1 | lh.precentral_18 | rh.fusiform_14 | rh.lateraloccipital_12 | |
| 788 | lh.supramarginal_14 | rh.postcentral_12 | lh.lateralorbitofrontal_4 | rh.precentral_18 | lh.lingual_11 | lh.lateraloccipital_20 | rh.lateraloccipital_21 | |
| 789 | lh.rostralmiddlefrontal_1 | rh.superiortemporal_24 | lh.parahippocampal_1 | rh.rostralanteriorcingulate_1 | rh.superiorfrontal_12 | rh.inferiortemporal_15 | rh.lateraloccipital_10 | |
| 790 | lh.precentral_18 | rh.lateraloccipital_20 | rh.supramarginal_15 | lh.superiortemporal_8 | rh.lateraloccipital_12 | lh.inferiortemporal_11 | rh.superiortemporal_25 | |
| 791 | lh.fusiform_14 | lh.transversetemporal_1 | lh.superiorfrontal_8 | lh.fusiform_9 | lh.postcentral_28 | rh.supramarginal_15 | lh.lingual_12 | |
| 792 | rh.precentral_1 | lh.superiorfrontal_8 | rh.supramarginal_18 | lh.inferiortemporal_8 | rh.pericalcarine_1 | rh.inferiortemporal_7 | lh.postcentral_28 | |
| 793 | rh.postcentral_22 | lh.rostralmiddlefrontal_3 | lh.superiorparietal_26 | lh.superiorfrontal_4 | rh.paracentral_1 | rh.supramarginal_18 | lh.middletemporal_2 | |
| 794 | lh.postcentral_1 | rh.postcentral_16 | rh.postcentral_12 | lh.rostralanteriorcingulate_4 | rh.parahippocampal_6 | lh.middletemporal_4 | rh.superiorparietal_6 | |
| 795 | lh.insula_11 | lh.precentral_18 | lh.lateralorbitofrontal_7 | lh.postcentral_20 | rh.superiorfrontal_10 | lh.supramarginal_13 | lh.fusiform_2 | |
| 796 | rh.precentral_17 | rh.lateralorbitofrontal_12 | lh.parahippocampal_4 | lh.parahippocampal_4 | rh.fusiform_7 | lh.transversetemporal_2 | rh.pericalcarine_4 | |
| 797 | rh.lateraloccipital_21 | rh.lateralorbitofrontal_8 | lh.superiortemporal_19 | lh.postcentral_23 | rh.transversetemporal_1 | lh.cuneus_1 | lh.superiortemporal_11 | |
| 798 | lh.precentral_30 | rh.superiorfrontal_10 | rh.superiorfrontal_10 | lh.superiorfrontal_8 | rh.parahippocampal_1 | rh.postcentral_13 | rh.superiortemporal_20 | |
| 799 | lh.paracentral_8 | rh.superiorfrontal_4 | lh.inferiortemporal_10 | rh.lateralorbitofrontal_12 | rh.lateraloccipital_21 | rh.postcentral_12 | lh.lingual_6 | |
| 800 | lh.superiorparietal_29 | lh.medialorbitofrontal_4 | lh.superiorparietal_4 | lh.inferiortemporal_7 | lh.postcentral_16 | lh.fusiform_15 | rh.middletemporal_10 | |
| 801 | rh.superiorfrontal_20 | rh.rostralmiddlefrontal_25 | lh.inferiortemporal_8 | rh.postcentral_12 | lh.superiorfrontal_40 | rh.inferiortemporal_2 | lh.superiorparietal_5 | |
| 802 | rh.superiorfrontal_12 | rh.rostralmiddlefrontal_9 | lh.inferiortemporal_9 | rh.rostralmiddlefrontal_9 | rh.cuneus_3 | rh.fusiform_17 | lh.lingual_11 | |
| 803 | lh.parstriangularis_1 | rh.superiorfrontal_39 | rh.precentral_28 | rh.transversetemporal_1 | rh.postcentral_20 | rh.superiorfrontal_5 | lh.parstriangularis_5 | |
| 804 | lh.superiorfrontal_12 | lh.insula_6 | lh.middletemporal_16 | rh.cuneus_2 | rh.pericalcarine_3 | lh.bankssts_6 | rh.rostralmiddlefrontal_9 | |
| 805 | rh.postcentral_15 | rh.transversetemporal_1 | lh.parstriangularis_4 | rh.superiortemporal_15 | rh.superiorparietal_4 | lh.parstriangularis_5 | rh.postcentral_16 | |
| 806 | lh.lateraloccipital_14 | lh.lateraloccipital_6 | lh.superiortemporal_9 | rh.superiortemporal_24 | lh.precentral_1 | rh.postcentral_6 | rh.lingual_4 | |
| 807 | lh.precentral_1 | rh.medialorbitofrontal_10 | rh.medialorbitofrontal_10 | lh.parstriangularis_7 | rh.pericalcarine_8 | rh.precentral_6 | rh.postcentral_12 | |
| 808 | lh.pericalcarine_7 | rh.fusiform_7 | rh.fusiform_4 | lh.insula_13 | rh.rostralmiddlefrontal_1 | lh.fusiform_5 | lh.pericalcarine_5 | |
| 809 | rh.postcentral_23 | lh.superiortemporal_19 | lh.precentral_18 | lh.superiortemporal_17 | rh.middletemporal_2 | lh.precentral_27 | rh.frontalpole_2 | |
| 810 | rh.superiorfrontal_39 | rh.lateralorbitofrontal_14 | rh.superiorfrontal_4 | lh.fusiform_7 | lh.lingual_13 | rh.postcentral_1 | rh.superiorfrontal_29 | |
| 811 | rh.rostralmiddlefrontal_17 | rh.fusiform_11 | rh.superiortemporal_7 | lh.insula_6 | rh.superiortemporal_12 | lh.entorhinal_2 | lh.lateraloccipital_15 | |
| 812 | rh.superiorfrontal_4 | rh.postcentral_10 | lh.precentral_14 | rh.parahippocampal_1 | lh.superiortemporal_3 | rh.bankssts_2 | lh.transversetemporal_1 | |
| 813 | rh.postcentral_20 | lh.postcentral_20 | rh.lateralorbitofrontal_3 | rh.superiortemporal_9 | rh.superiortemporal_18 | lh.middletemporal_2 | rh.postcentral_7 | |
| 814 | lh.parahippocampal_1 | lh.superiortemporal_4 | rh.rostralmiddlefrontal_9 | rh.rostralmiddlefrontal_25 | lh.cuneus_3 | rh.superiorfrontal_32 | rh.middletemporal_11 | |
| 815 | lh.superiorfrontal_40 | rh.fusiform_3 | rh.cuneus_2 | lh.transversetemporal_1 | rh.postcentral_15 | lh.superiorfrontal_12 | lh.pericalcarine_2 | |
| 816 | rh.parstriangularis_5 | lh.precentral_31 | rh.postcentral_17 | rh.rostralanteriorcingulate_2 | rh.postcentral_17 | lh.fusiform_6 | rh.fusiform_17 | |
| 817 | rh.cuneus_4 | rh.postcentral_17 | rh.postcentral_16 | rh.lateralorbitofrontal_14 | rh.pericalcarine_4 | lh.lingual_8 | rh.rostralmiddlefrontal_26 | |
| 818 | rh.lingual_16 | lh.insula_9 | lh.superiortemporal_4 | rh.superiorparietal_4 | rh.lateraloccipital_13 | rh.superiortemporal_12 | rh.lateralorbitofrontal_12 | |
| 819 | rh.lingual_11 | rh.parahippocampal_1 | rh.transversetemporal_1 | rh.postcentral_15 | rh.cuneus_7 | rh.postcentral_26 | lh.superiorparietal_23 | |
| 820 | rh.inferiortemporal_2 | rh.cuneus_2 | lh.postcentral_24 | rh.postcentral_17 | rh.lateralorbitofrontal_15 | rh.fusiform_7 | rh.superiorparietal_19 | |
| 821 | lh.postcentral_16 | lh.postcentral_16 | rh.parstriangularis_7 | lh.rostralanteriorcingulate_1 | lh.lateraloccipital_10 | lh.middletemporal_1 | lh.cuneus_1 | |
| 822 | lh.lateraloccipital_1 | rh.postcentral_15 | rh.lateralorbitofrontal_13 | lh.postcentral_16 | lh.lateraloccipital_14 | lh.pericalcarine_3 | rh.parstriangularis_2 | |
| 823 | lh.lateraloccipital_13 | lh.bankssts_1 | rh.inferiortemporal_7 | lh.rostralmiddlefrontal_3 | lh.fusiform_14 | lh.fusiform_16 | lh.medialorbitofrontal_8 | |
| 824 | rh.paracentral_1 | lh.medialorbitofrontal_5 | rh.lateraloccipital_20 | lh.lateralorbitofrontal_2 | rh.rostralmiddlefrontal_17 | rh.lingual_9 | lh.lingual_8 | |
| 825 | lh.superiortemporal_18 | lh.insula_13 | rh.inferiortemporal_12 | lh.lateralorbitofrontal_8 | rh.inferiortemporal_16 | rh.lateraloccipital_14 | lh.lingual_8 | |
| 826 | lh.lateralorbitofrontal_12 | lh.postcentral_24 | lh.rostralmiddlefrontal_3 | lh.medialorbitofrontal_4 | lh.postcentral_22 | rh.superiortemporal_6 | lh.lateraloccipital_10 | |

| # | Col1 | Col2 | Col3 | Col4 | Col5 | Col6 | Col7 | Col8 |
|---|---|---|---|---|---|---|---|---|
| 827 | lh.lateraloccipital_10 | lh.lateralorbitofrontal_1 | lh.medialorbitofrontal_4 | lh.superiortemporal_5 | lh.medialorbitofrontal_7 | lh.entorhinal_1 | lh.lateraloccipital_13 |
| 828 | lh.superiortemporal_13 | lh.superiortemporal_26 | lh.lateralorbitofrontal_13 | lh.postcentral_16 | lh.superiorparietal_29 | rh.middletemporal_8 | lh.superiorparietal_17 |
| 829 | lh.superiortemporal_5 | lh.superiorfrontal_11 | lh.medialorbitofrontal_7 | lh.superiortemporal_4 | rh.middletemporal_19 | lh.superiorparietal_23 | rh.lateralorbitofrontal_16 |
| 830 | lh.insula_17 | rh.superiorparietal_4 | rh.superiorfrontal_2 | rh.medialorbitofrontal_10 | lh.lateraloccipital_1 | rh.fusiform_4 | lh.cuneus_2 |
| 831 | lh.superiorparietal_4 | rh.lateralorbitofrontal_5 | lh.postcentral_20 | rh.superiortemporal_13 | rh.superiorfrontal_29 | lh.inferiortemporal_1 | lh.superiorparietal_9 |
| 832 | rh.superiorparietal_27 | rh.superiortemporal_2 | lh.postcentral_4 | rh.postcentral_10 | rh.lateraloccipital_20 | rh.precentral_1 | rh.fusiform_1 |
| 833 | rh.middletemporal_6 | rh.postcentral_4 | lh.postcentral_15 | lh.parstriangularis_1 | rh.superiortemporal_23 | rh.precentral_2 | lh.pericalcarine_1 |
| 834 | lh.cuneus_7 | lh.superiortemporal_15 | lh.lateraloccipital_6 | rh.lateraloccipital_1 | rh.middletemporal_6 | rh.transversetemporal_2 | lh.superiortemporal_8 |
| 835 | rh.pericalcarine_6 | lh.parstriangularis_1 | lh.superiortemporal_19 | lh.bankssts_1 | rh.inferiortemporal_2 | rh.middletemporal_10 | rh.lingual_9 |
| 836 | rh.lateraloccipital_20 | lh.superiortemporal_5 | lh.superiortemporal_5 | rh.postcentral_28 | rh.fusiform_6 | lh.cuneus_6 | lh.inferiortemporal_11 |
| 837 | lh.parstriangularis_7 | lh.superiortemporal_7 | lh.superiortemporal_26 | rh.lateralorbitofrontal_5 | rh.parstriangularis_5 | rh.rostralmiddlefrontal_8 | rh.superiortemporal_23 |
| 838 | lh.lingual_4 | lh.superiortemporal_19 | lh.parstriangularis_1 | rh.inferiortemporal_2 | lh.postcentral_1 | rh.parahippocampal_6 | lh.fusiform_12 |
| 839 | rh.inferiortemporal_11 | lh.parahippocampal_1 | rh.fusiform_11 | lh.medialorbitofrontal_5 | rh.parstriangularis_2 | lh.superiortemporal_10 | rh.cuneus_5 |
| 840 | rh.superiortemporal_12 | lh.postcentral_17 | rh.parahippocampal_6 | lh.superiorfrontal_11 | lh.lingual_4 | rh.superiortemporal_9 | rh.pericalcarine_1 |
| 841 | rh.superiortemporal_23 | lh.superiortemporal_3 | lh.postcentral_16 | lh.insula_9 | lh.lingual_16 | lh.superiortemporal_15 | lh.lingual_14 |
| 842 | rh.middletemporal_14 | rh.superiortemporal_9 | rh.superiortemporal_15 | lh.postcentral_4 | lh.rostralmiddlefrontal_11 | rh.lateralorbitofrontal_2 | rh.parahippocampal_4 |
| 843 | lh.lingual_16 | rh.inferiortemporal_2 | lh.superiorfrontal_11 | lh.postcentral_24 | lh.inferiortemporal_6 | lh.supramarginal_14 | rh.cuneus_6 |
| 844 | rh.superiorfrontal_29 | lh.lateralorbitofrontal_2 | rh.rostralmiddlefrontal_25 | rh.superiortemporal_16 | lh.parstriangularis_5 | rh.postcentral_7 | rh.fusiform_5 |
| 845 | lh.transversetemporal_4 | lh.lateralorbitofrontal_3 | lh.fusiform_7 | rh.fusiform_3 | rh.inferiortemporal_5 | lh.parahippocampal_6 | rh.rostralmiddlefrontal_11 |
| 846 | rh.parstriangularis_2 | lh.lateralorbitofrontal_5 | lh.bankssts_1 | lh.superiortemporal_3 | lh.cuneus_7 | rh.postcentral_16 | rh.pericalcarine_7 |
| 847 | rh.rostralmiddlefrontal_1 | lh.superiortemporal_11 | lh.postcentral_22 | lh.superiortemporal_26 | rh.cuneus_4 | rh.pericalcarine_5 | rh.transversetemporal_1 |
| 848 | rh.superiorfrontal_23 | rh.lingual_11 | lh.lateralorbitofrontal_5 | rh.lingual_11 | rh.superiorparietal_27 | rh.superiortemporal_16 | lh.superiorparietal_3 |
| 849 | rh.cuneus_2 | lh.parstriangularis_5 | rh.superiortemporal_13 | lh.lateralorbitofrontal_1 | rh.inferiortemporal_11 | lh.lateraloccipital_9 | lh.postcentral_16 |
| 850 | lh.lateralorbitofrontal_13 | rh.lateralorbitofrontal_4 | lh.fusiform_3 | rh.lateralorbitofrontal_3 | rh.parsorbitalis_1 | rh.bankssts_3 | rh.lateraloccipital_20 |
| 851 | lh.rostralmiddlefrontal_11 | rh.rostralmiddlefrontal_8 | rh.postcentral_28 | rh.postcentral_4 | lh.lateralorbitofrontal_7 | lh.superiortemporal_18 | lh.fusiform_9 |
| 852 | lh.lateraloccipital_12 | rh.lateraloccipital_1 | lh.superiortemporal_11 | lh.parstriangularis_5 | rh.parstriangularis_7 | rh.pericalcarine_7 | rh.postcentral_11 |
| 853 | lh.parahippocampal_6 | rh.lateralorbitofrontal_1 | lh.superiortemporal_17 | rh.rostralmiddlefrontal_1 | lh.postcentral_11 | rh.fusiform_9 | lh.lateralorbitofrontal_7 |
| 854 | rh.lingual_13 | rh.inferiortemporal_11 | lh.lateralorbitofrontal_6 | rh.inferiortemporal_11 | lh.lateraloccipital_12 | rh.superiorfrontal_29 | rh.entorhinal_1 |
| 855 | rh.superiorfrontal_10 | rh.rostralmiddlefrontal_1 | lh.postcentral_17 | rh.inferiortemporal_1 | lh.fusiform_13 | lh.superiortemporal_4 | lh.inferiortemporal_8 |
| 856 | lh.parstriangularis_5 | lh.postcentral_4 | rh.medialorbitofrontal_8 | rh.transversetemporal_4 | rh.postcentral_22 | rh.fusiform_2 | lh.pericalcarine_3 |
| 857 | lh.inferiortemporal_5 | lh.fusiform_14 | lh.superiortemporal_23 | rh.parahippocampal_5 | rh.middletemporal_14 | rh.middletemporal_5 | lh.parahippocampal_3 |
| 858 | lh.insula_13 | rh.middletemporal_6 | lh.parstriangularis_5 | rh.parstriangularis_2 | rh.insula_13 | lh.pericalcarine_1 | lh.pericalcarine_7 |
| 859 | lh.medialorbitofrontal_7 | rh.parahippocampal_6 | rh.parahippocampal_5 | lh.postcentral_17 | rh.lateralorbitofrontal_4 | lh.fusiform_14 | lh.lateralorbitofrontal_11 |
| 860 | rh.lateralorbitofrontal_4 | lh.transversetemporal_4 | rh.superiortemporal_13 | rh.middletemporal_6 | lh.parahippocampal_6 | lh.superiortemporal_3 | rh.pericalcarine_8 |
| 861 | rh.lingual_14 | rh.parstriangularis_2 | rh.lateralorbitofrontal_1 | lh.inferiortemporal_5 | lh.lateralorbitofrontal_4 | rh.inferiortemporal_10 | rh.postcentral_15 |
| 862 | lh.lateralorbitofrontal_7 | lh.superiortemporal_13 | rh.superiorfrontal_40 | rh.lateralorbitofrontal_7 | rh.superiorfrontal_39 | lh.insula_1 | lh.lingual_10 |
| 863 | lh.postcentral_20 | rh.postcentral_28 | lh.superiortemporal_3 | rh.superiorfrontal_40 | rh.fusiform_10 | rh.fusiform_12 | lh.superiortemporal_8 |
| 864 | lh.middletemporal_5 | lh.inferiortemporal_5 | lh.fusiform_14 | rh.parstriangularis_1 | lh.superiortemporal_13 | rh.transversetemporal_3 | rh.rostralmiddlefrontal_17 |
| 865 | lh.middletemporal_9 | lh.lateralorbitofrontal_7 | rh.rostralmiddlefrontal_1 | rh.lateralorbitofrontal_13 | rh.superiortemporal_9 | lh.transversetemporal_3 | rh.middletemporal_6 |
| 866 | lh.postcentral_11 | rh.postcentral_22 | rh.rostralmiddlefrontal_1 | lh.superiortemporal_21 | rh.middletemporal_11 | lh.postcentral_15 | lh.cuneus_4 |
| 867 | lh.lateraloccipital_6 | lh.postcentral_22 | rh.postcentral_4 | lh.parstriangularis_5 | lh.precentral_31 | lh.inferiortemporal_3 | lh.lateraloccipital_1 |
| 868 | lh.lingual_3 | lh.parstriangularis_7 | lh.inferiortemporal_5 | lh.superiortemporal_13 | lh.lateraloccipital_6 | rh.medialorbitofrontal_1 | rh.lingual_12 |
| 869 | rh.fusiform_11 | lh.parstriangularis_5 | rh.postcentral_10 | lh.postcentral_6 | lh.lateralorbitofrontal_12 | rh.fusiform_13 | lh.lingual_4 |
| 870 | rh.middletemporal_19 | lh.lateralorbitofrontal_13 | rh.lingual_11 | lh.inferiortemporal_15 | rh.fusiform_1 | rh.lateraloccipital_12 | lh.lateralorbitofrontal_15 |
| 871 | rh.superiortemporal_15 | rh.rostralmiddlefrontal_17 | lh.lateraloccipital_1 | rh.inferiortemporal_12 | lh.lingual_3 |
| 872 | lh.lateralorbitofrontal_2 | lh.rostralmiddlefrontal_1 | rh.inferiortemporal_1 | lh.fusiform_14 | rh.medialorbitofrontal_4 | lh.middletemporal_3 | lh.fusiform_1 |

| # | Col1 | Col2 | Col3 | Col4 | Col5 | Col6 | Col7 |
|---|---|---|---|---|---|---|---|
| 873 | lh.superiortemporal_19 | lh.medialorbitofrontal_7 | rh.inferiortemporal_2 | lh.lateralorbitofrontal_6 | rh.pericalcarine_6 | rh.superiortemporal_25 | rh.fusiform_7 |
| 874 | rh.middletemporal_2 | rh.inferiortemporal_1 | lh.lateralorbitofrontal_9 | lh.superiortemporal_11 | rh.lingual_14 | rh.lateraloccipital_8 | rh.middletemporal_9 |
| 875 | lh.postcentral_27 | rh.superiortemporal_23 | rh.middletemporal_6 | rh.lateralorbitofrontal_1 | lh.superiortemporal_5 | lh.pericalcarine_2 | rh.lingual_7 |
| 876 | lh.inferiortemporal_6 | rh.lateralorbitofrontal_13 | lh.rostralmiddlefrontal_8 | lh.medialorbitofrontal_7 | rh.superiorfrontal_40 | lh.superiortemporal_9 | lh.lateraloccipital_14 |
| 877 | rh.middletemporal_9 | lh.postcentral_1 | rh.parstriangularis_1 | rh.inferiortemporal_16 | lh.inferiortemporal_5 | rh.postcentral_1 | rh.inferiortemporal_2 |
| 878 | lh.postcentral_4 | rh.parstriangularis_1 | lh.medialorbitofrontal_10 | lh.lateralorbitofrontal_13 | lh.postcentral_20 | lh.lateraloccipital_1 | lh.entorhinal_2 |
| 879 | rh.fusiform_1 | rh.postcentral_23 | lh.postcentral_6 | lh.parstriangularis_7 | rh.middletemporal_9 | lh.middletemporal_7 | lh.cuneus_5 |
| 880 | rh.parsorbitalis_1 | lh.rostralmiddlefrontal_11 | lh.inferiortemporal_15 | lh.rostralmiddlefrontal_8 | rh.superiortemporal_15 | lh.inferiortemporal_12 | rh.parstriangularis_5 |
| 881 | rh.middletemporal_11 | lh.superiortemporal_17 | lh.parahippocampal_3 | rh.parahippocampal_6 | lh.lingual_3 | lh.postcentral_29 | lh.parstriangularis_1 |
| 882 | rh.lateralorbitofrontal_13 | lh.lateralorbitofrontal_12 | lh.lateralorbitofrontal_11 | lh.lateralorbitofrontal_5 | rh.middletemporal_18 | rh.middletemporal_4 | lh.postcentral_17 |
| 883 | rh.superiortemporal_9 | rh.inferiortemporal_16 | lh.postcentral_22 | rh.postcentral_22 | lh.lateralorbitofrontal_2 | lh.cuneus_3 | rh.rostralmiddlefrontal_8 |
| 884 | lh.inferiortemporal_16 | lh.postcentral_28 | lh.transversetemporal_4 | lh.lateralorbitofrontal_9 | lh.lateralorbitofrontal_3 | lh.cuneus_2 | lh.cuneus_3 |
| 885 | lh.precentral_31 | rh.superiorfrontal_40 | rh.inferiortemporal_16 | rh.rostralmiddlefrontal_17 | lh.lingual_16 | lh.fusiform_8 | rh.lateralorbitofrontal_9 |
| 886 | rh.lateralorbitofrontal_15 | lh.inferiortemporal_15 | lh.parstriangularis_7 | rh.superiortemporal_23 | lh.lingual_13 | lh.transversetemporal_4 | rh.postcentral_4 |
| 887 | rh.inferiortemporal_5 | lh.middletemporal_9 | lh.lateralorbitofrontal_3 | rh.medialorbitofrontal_8 | lh.transversetemporal_4 | lh.precentral_18 | lh.medialorbitofrontal_2 |
| 888 | lh.inferiortemporal_13 | rh.middletemporal_19 | lh.postcentral_1 | lh.middletemporal_19 | lh.postcentral_27 | lh.inferiortemporal_2 | lh.lingual_9 |
| 889 | rh.superiorfrontal_39 | rh.superiortemporal_13 | lh.lateralorbitofrontal_4 | rh.middletemporal_9 | rh.medialorbitofrontal_8 | lh.superiortemporal_5 | rh.parahippocampal_6 |
| 890 | rh.lateraloccipital_1 | rh.middletemporal_9 | lh.postcentral_27 | rh.lateralorbitofrontal_11 | rh.lateralorbitofrontal_13 | rh.postcentral_20 | lh.lateraloccipital_6 |
| 891 | rh.lingual_8 | rh.middletemporal_2 | rh.parstriangularis_5 | lh.postcentral_22 | lh.lingual_11 | lh.lateraloccipital_12 | lh.lateralorbitofrontal_13 |
| 892 | rh.fusiform_10 | lh.postcentral_27 | lh.inferiortemporal_6 | lh.superiortemporal_2 | lh.postcentral_4 | rh.bankssts_4 | lh.fusiform_4 |
| 893 | rh.inferiortemporal_1 | lh.middletemporal_14 | rh.middletemporal_14 | rh.middletemporal_2 | rh.rostralmiddlefrontal_12 | lh.parahippocampal_3 | rh.postcentral_17 |
| 894 | rh.parstriangularis_1 | lh.inferiortemporal_6 | rh.superiortemporal_18 | lh.parahippocampal_3 | rh.superiortemporal_20 | lh.cuneus_5 | rh.lingual_16 |
| 895 | lh.fusiform_13 | rh.medialorbitofrontal_8 | lh.inferiortemporal_6 | lh.middletemporal_9 | lh.middletemporal_12 | lh.lateraloccipital_19 | rh.superiortemporal_15 |
| 896 | rh.medialorbitofrontal_8 | lh.superiortemporal_16 | lh.temporalpole_2 | rh.rostralmiddlefrontal_1 | rh.postcentral_28 | lh.middletemporal_13 | rh.inferiortemporal_11 |
| 897 | lh.inferiortemporal_15 | lh.superiortemporal_18 | rh.middletemporal_19 | lh.rostralmiddlefrontal_11 | rh.parstriangularis_1 | lh.lateraloccipital_13 | rh.lingual_2 |
| 898 | rh.superiortemporal_20 | rh.fusiform_1 | rh.rostralmiddlefrontal_17 | rh.superiortemporal_22 | rh.cuneus_2 | rh.postcentral_4 | lh.cuneus_7 |
| 899 | lh.lateralorbitofrontal_3 | lh.superiortemporal_23 | lh.rostralmiddlefrontal_11 | lh.temporalpole_2 | lh.inferiortemporal_1 | lh.lateraloccipital_10 | lh.lateraloccipital_12 |
| 900 | rh.rostralmiddlefrontal_12 | lh.superiortemporal_21 | lh.parahippocampal_6 | lh.superiortemporal_23 | lh.parahippocampal_1 | lh.postcentral_11 | lh.entorhinal_1 |
| 901 | rh.postcentral_28 | rh.fusiform_10 | lh.superiortemporal_16 | lh.postcentral_27 | rh.inferiortemporal_3 | rh.inferiortemporal_8 | lh.lateralorbitofrontal_2 |
| 902 | lh.lateralorbitofrontal_11 | lh.parahippocampal_6 | rh.parstriangularis_2 | lh.postcentral_1 | lh.middletemporal_9 | rh.superiortemporal_22 | lh.superiorparietal_29 |
| 903 | lh.parstriangularis_3 | lh.superiortemporal_2 | lh.superiortemporal_21 | lh.lateralorbitofrontal_12 | lh.middletemporal_5 | lh.superiortemporal_21 | rh.inferiortemporal_1 |
| 904 | lh.superiortemporal_23 | lh.fusiform_13 | rh.inferiortemporal_11 | rh.postcentral_23 | rh.fusiform_11 | lh.superiorparietal_17 | rh.cuneus_3 |
| 905 | rh.superiortemporal_5 | lh.lateralorbitofrontal_6 | lh.middletemporal_9 | lh.inferiortemporal_6 | lh.rostralmiddlefrontal_25 | rh.lateraloccipital_20 | lh.fusiform_14 |
| 906 | rh.superiorfrontal_40 | lh.postcentral_6 | rh.middletemporal_9 | lh.postcentral_28 | rh.inferiortemporal_9 | lh.cuneus_7 | lh.insula_13 |
| 907 | lh.inferiortemporal_14 | rh.superiortemporal_20 | rh.medialorbitofrontal_4 | lh.inferiortemporal_14 | rh.parsorbitalis_3 | lh.middletemporal_16 | rh.superiorparietal_27 |
| 908 | lh.middletemporal_8 | lh.inferiortemporal_14 | rh.superiortemporal_5 | rh.fusiform_10 | lh.lateralorbitofrontal_11 | lh.inferiortemporal_10 | rh.pericalcarine_5 |
| 909 | rh.lateralorbitofrontal_16 | rh.parahippocampal_5 | lh.entorhinal_2 | rh.superiortemporal_20 | lh.middletemporal_2 | lh.cuneus_7 | rh.lingual_14 |
| 910 | lh.middletemporal_2 | lh.parstriangularis_3 | lh.superiortemporal_23 | rh.fusiform_13 | rh.lateralorbitofrontal_17 | rh.lingual_14 | lh.temporalpole_1 |
| 911 | rh.medialorbitofrontal_3 | lh.temporalpole_2 | lh.postcentral_28 | rh.fusiform_1 | lh.superiortemporal_2 | lh.lateraloccipital_15 | rh.cuneus_4 |
| 912 | rh.middletemporal_18 | lh.parstriangularis_4 | rh.middletemporal_2 | lh.parstriangularis_4 | rh.superiortemporal_8 | rh.middletemporal_3 | rh.cuneus_8 |
| 913 | rh.parahippocampal_1 | lh.middletemporal_2 | rh.entorhinal_2 | lh.middletemporal_12 | lh.middletemporal_4 | rh.superiortemporal_2 | lh.postcentral_27 |
| 914 | rh.lateralorbitofrontal_17 | lh.middletemporal_4 | rh.superiortemporal_22 | lh.middletemporal_14 | rh.temporalpole_2 | rh.cuneus_4 | rh.cuneus_7 |
| 915 | rh.rostralmiddlefrontal_26 | lh.lateralorbitofrontal_3 | rh.fusiform_10 | lh.medialorbitofrontal_1 | rh.lateralorbitofrontal_2 | lh.frontalpole_1 | lh.lingual_3 |
| 916 | rh.superiortemporal_8 | lh.middletemporal_12 | rh.postcentral_23 | lh.middletemporal_4 | rh.medialorbitofrontal_11 | rh.transversetemporal_1 | rh.pericalcarine_3 |
| 917 | lh.superiortemporal_2 | lh.inferiortemporal_13 | lh.entorhinal_1 | lh.middletemporal_2 | lh.entorhinal_3 | lh.lingual_4 | lh.middletemporal_15 |
| 918 | lh.parstriangularis_4 | rh.medialorbitofrontal_3 | rh.medialorbitofrontal_3 | lh.parstriangularis_3 | rh.lateralorbitofrontal_16 | lh.fusiform_12 | |

| | | | | | | | |
|---|---|---|---|---|---|---|---|
| 919 lh.rostralmiddlefrontal_25 | lh.middletemporal_5 | rh.lateralorbitofrontal_15 | lh.medialorbitofrontal_10 | lh.temporalpole_2 | lh.bankssts_1 | lh.superiortemporal_3 | |
| 920 lh.entorhinal_3 | lh.middletemporal_8 | rh.medialorbitofrontal_3 | rh.medialorbitofrontal_3 | rh.rostralmiddlefrontal_26 | rh.lingual_8 | rh.inferiortemporal_8 | |
| 921 rh.lateralorbitofrontal_2 | rh.superiortemporal_5 | rh.lateralorbitofrontal_16 | lh.temporalpole_3 | lh.fusiform_18 | rh.postcentral_25 | lh.frontalpole_1 | |
| 922 rh.medialorbitofrontal_4 | lh.medialorbitofrontal_10 | rh.superiortemporal_20 | lh.parahippocampal_6 | rh.lateraloccipital_1 | lh.fusiform_1 | lh.postcentral_20 | |
| 923 lh.inferiortemporal_11 | rh.inferiortemporal_5 | rh.lateralorbitofrontal_12 | lh.lateralorbitofrontal_3 | rh.rostralmiddlefrontal_5 | rh.postcentral_17 | rh.postcentral_4 | |
| 924 lh.middletemporal_4 | rh.lateralorbitofrontal_9 | rh.lateralorbitofrontal_8 | lh.lateralorbitofrontal_4 | lh.superiortemporal_19 | rh.fusiform_5 | rh.medialorbitofrontal_1 | |
| 925 lh.superiortemporal_22 | lh.superiortemporal_10 | lh.middletemporal_4 | lh.middletemporal_8 | lh.inferiortemporal_11 | rh.cuneus_3 | lh.parahippocampal_6 | |
| 926 rh.medialorbitofrontal_1 | rh.lateralorbitofrontal_16 | rh.fusiform_1 | lh.superiortemporal_10 | rh.rostralmiddlefrontal_2 | lh.fusiform_10 | lh.parstriangularis_7 | |
| 927 rh.inferiortemporal_3 | rh.parsorbitalis_1 | lh.medialorbitofrontal_1 | rh.superiortemporal_18 | lh.parstriangularis_3 | rh.inferiortemporal_6 | lh.parstriangularis_4 | |
| 928 lh.middletemporal_15 | rh.lateralorbitofrontal_15 | lh.middletemporal_2 | rh.superiortemporal_5 | rh.lingual_8 | lh.transversetemporal_1 | lh.postcentral_1 | |
| 929 rh.parsorbitalis_3 | rh.inferiortemporal_3 | rh.parahippocampal_4 | rh.entorhinal_2 | rh.superiortemporal_5 | lh.parahippocampal_4 | lh.medialorbitofrontal_6 | |
| 930 lh.middletemporal_12 | rh.rostralmiddlefrontal_26 | rh.parsorbitalis_1 | rh.inferiortemporal_8 | rh.lateralorbitofrontal_10 | rh.middletemporal_11 | rh.inferiortemporal_5 | |
| 931 rh.inferiortemporal_9 | rh.middletemporal_11 | rh.temporalpole_2 | rh.lateralorbitofrontal_10 | lh.parstriangularis_4 | lh.postcentral_20 | rh.lingual_1 | |
| 932 lh.superiortemporal_10 | rh.rostralmiddlefrontal_12 | rh.superiortemporal_2 | rh.inferiortemporal_3 | lh.temporalpole_3 | lh.middletemporal_15 | rh.middletemporal_14 | |
| 933 lh.medialorbitofrontal_10 | rh.middletemporal_8 | rh.middletemporal_8 | rh.inferiortemporal_2 | lh.middletemporal_15 | lh.postcentral_27 | rh.fusiform_6 | |
| 934 rh.medialorbitofrontal_11 | rh.medialorbitofrontal_4 | lh.fusiform_18 | rh.inferiortemporal_13 | rh.medialorbitofrontal_1 | lh.middletemporal_14 | lh.postcentral_22 | |
| 935 rh.inferiortemporal_4 | rh.medialorbitofrontal_1 | lh.parstriangularis_3 | rh.parahippocampal_4 | rh.fusiform_16 | lh.inferiortemporal_8 | lh.lingual_13 | |
| 936 rh.middletemporal_3 | lh.temporalpole_3 | rh.rostralmiddlefrontal_26 | rh.lateralorbitofrontal_16 | lh.lateralorbitofrontal_15 | rh.frontalpole_2 | lh.parstriangularis_3 | |
| 937 lh.parsorbitalis_3 | lh.superiortemporal_22 | rh.temporalpole_3 | lh.medialorbitofrontal_2 | rh.lateralorbitofrontal_8 | rh.pericalcarine_3 | lh.lateralorbitofrontal_6 | |
| 938 lh.lateralorbitofrontal_15 | rh.inferiortemporal_9 | rh.lateralorbitofrontal_17 | rh.rostralmiddlefrontal_26 | lh.inferiortemporal_14 | rh.superiortemporal_22 | rh.middletemporal_5 | |
| 939 lh.medialorbitofrontal_8 | lh.middletemporal_15 | rh.rostralmiddlefrontal_8 | rh.medialorbitofrontal_6 | lh.middletemporal_8 | lh.inferiortemporal_9 | rh.superiortemporal_5 | |
| 940 rh.rostralmiddlefrontal_5 | lh.parsorbitalis_3 | rh.inferiortemporal_3 | rh.inferiortemporal_9 | lh.entorhinal_2 | lh.parahippocampal_2 | rh.superiortemporal_2 | |
| 941 rh.temporalpole_2 | lh.entorhinal_3 | rh.parstriangularis_4 | rh.medialorbitofrontal_1 | lh.medialorbitofrontal_3 | lh.lateraloccipital_14 | rh.lingual_5 | |
| 942 rh.superiortemporal_2 | lh.rostralmiddlefrontal_25 | lh.middletemporal_12 | rh.medialorbitofrontal_1 | rh.rostralmiddlefrontal_8 | lh.superiortemporal_25 | lh.parsorbitalis_3 | |
| 943 lh.fusiform_18 | rh.lateralorbitofrontal_11 | lh.middletemporal_15 | lh.fusiform_18 | lh.medialorbitofrontal_2 | rh.fusiform_8 | rh.cuneus_1 | |
| 944 lh.entorhinal_1 | rh.superiortemporal_22 | rh.lateralorbitofrontal_10 | rh.medialorbitofrontal_4 | rh.superiortemporal_22 | rh.fusiform_11 | rh.postcentral_23 | |
| 945 rh.rostralmiddlefrontal_8 | rh.inferiortemporal_4 | lh.medialorbitofrontal_6 | lh.superiortemporal_12 | lh.parsorbitalis_3 | lh.inferiortemporal_4 | rh.superiortemporal_17 | |
| 946 lh.medialorbitofrontal_6 | lh.fusiform_18 | rh.medialorbitofrontal_11 | rh.inferiortemporal_4 | lh.superiortemporal_23 | lh.superiorparietal_29 | rh.pericalcarine_6 | |
| 947 lh.temporalpole_2 | lh.medialorbitofrontal_2 | rh.rostralmiddlefrontal_12 | lh.superiortemporal_22 | rh.superiortemporal_10 | rh.lingual_16 | rh.lateraloccipital_1 | |
| 948 rh.superiortemporal_7 | rh.lateralorbitofrontal_17 | rh.inferiortemporal_5 | lh.middletemporal_15 | lh.medialorbitofrontal_10 | rh.superiortemporal_13 | lh.superiortemporal_13 | |
| 949 lh.medialorbitofrontal_2 | rh.lateralorbitofrontal_2 | lh.inferiortemporal_13 | rh.middletemporal_18 | lh.medialorbitofrontal_6 | lh.superiortemporal_13 | rh.lingual_13 | |
| 950 lh.rostralmiddlefrontal_2 | lh.parahippocampal_3 | rh.middletemporal_11 | rh.inferiortemporal_5 | lh.superiortemporal_10 | lh.superiortemporal_26 | lh.lateralorbitofrontal_12 | |
| 951 rh.temporalpole_3 | lh.lateralorbitofrontal_11 | rh.medialorbitofrontal_6 | rh.temporalpole_2 | rh.temporalpole_3 | rh.lateraloccipital_1 | lh.parsorbitalis_1 | |
| 952 lh.superiortemporal_21 | rh.temporalpole_2 | lh.middletemporal_5 | lh.entorhinal_3 | lh.superiortemporal_7 | lh.lingual_3 | rh.fusiform_11 | |
| 953 lh.lateralorbitofrontal_6 | rh.entorhinal_2 | rh.medialorbitofrontal_1 | lh.rostralmiddlefrontal_25 | rh.medialorbitofrontal_6 | lh.temporalpole_1 | lh.medialorbitofrontal_10 | |
| 954 lh.postcentral_6 | lh.inferiortemporal_11 | rh.inferiortemporal_9 | rh.rostralmiddlefrontal_12 | rh.middletemporal_3 | lh.middletemporal_9 | lh.superiortemporal_5 | |
| 955 lh.entorhinal_1 | rh.inferiortemporal_8 | lh.entorhinal_3 | lh.medialorbitofrontal_6 | rh.superiortemporal_2 | rh.superiortemporal_7 | rh.middletemporal_3 | |
| 956 rh.entorhinal_2 | lh.medialorbitofrontal_6 | rh.middletemporal_18 | lh.middletemporal_5 | rh.lateralorbitofrontal_6 | lh.fusiform_4 | rh.entorhinal_2 | |
| 957 rh.inferiortemporal_8 | lh.lateralorbitofrontal_4 | lh.superiortemporal_10 | rh.lateralorbitofrontal_15 | lh.entorhinal_1 | rh.frontalpole_1 | lh.postcentral_6 | |
| 958 rh.lateralorbitofrontal_8 | lh.entorhinal_2 | lh.rostralmiddlefrontal_2 | lh.entorhinal_2 | lh.medialorbitofrontal_8 | lh.medialorbitofrontal_10 | lh.superiortemporal_10 | |
| 959 rh.medialorbitofrontal_6 | rh.superiortemporal_8 | rh.rostralmiddlefrontal_25 | rh.lateralorbitofrontal_2 | rh.superiortemporal_1 | lh.superiortemporal_23 | lh.transversetemporal_4 | |
| 960 rh.middletemporal_17 | rh.medialorbitofrontal_6 | rh.lateralorbitofrontal_2 | rh.parsorbitalis_1 | lh.inferiortemporal_2 | lh.postcentral_4 | lh.superiortemporal_22 | |
| 961 rh.superiortemporal_13 | rh.parsorbitalis_3 | rh.inferiortemporal_4 | rh.lateralorbitofrontal_17 | rh.entorhinal_1 | rh.lingual_3 | lh.middletemporal_9 | |
| 962 lh.middletemporal_16 | lh.medialorbitofrontal_1 | rh.inferiortemporal_2 | rh.rostralmiddlefrontal_2 | rh.superiortemporal_7 | rh.cuneus_2 | lh.inferiortemporal_14 | |
| 963 rh.parahippocampal_5 | lh.rostralmiddlefrontal_2 | lh.lateralorbitofrontal_11 | lh.entorhinal_1 | rh.inferiortemporal_8 | lh.middletemporal_5 | lh.middletemporal_5 | |
| 964 rh.superiortemporal_10 | rh.rostralmiddlefrontal_8 | rh.medialorbitofrontal_5 | rh.lateralorbitofrontal_9 | lh.postcentral_6 | lh.lingual_13 | lh.lingual_16 | |

| # | Col1 | Col2 | Col3 | Col4 | Col5 | Col6 | Col7 | Col8 |
|---|---|---|---|---|---|---|---|---|
| 965 | lh.superiortemporal_7 | lh.inferiortemporal_2 | rh.inferiortemporal_8 | rh.middletemporal_11 | rh.inferiortemporal_4 | lh.lingual_9 | rh.medialorbitofrontal_3 | |
| 966 | rh.entorhinal_1 | rh.medialorbitofrontal_11 | lh.parsorbitalis_3 | rh.lateralorbitofrontal_7 | rh.lateralorbitofrontal_9 | rh.postcentral_23 | rh.parahippocampal_5 | |
| 967 | rh.lateralorbitofrontal_9 | rh.parahippocampal_4 | rh.lateralorbitofrontal_9 | rh.rostralmiddlefrontal_8 | lh.superiortemporal_21 | lh.lingual_16 | lh.middletemporal_8 | |
| 968 | rh.parsorbitalis_2 | lh.entorhinal_1 | lh.superiortemporal_22 | lh.inferiortemporal_11 | rh.entorhinal_2 | rh.parahippocampal_4 | rh.frontalpole_2 | |
| 969 | lh.lateralorbitofrontal_4 | lh.medialorbitofrontal_8 | rh.entorhinal_1 | lh.lateralorbitofrontal_8 | lh.superiortemporal_6 | rh.superiortemporal_13 | lh.lateralorbitofrontal_4 | rh.lingual_8 |
| 970 | lh.superiortemporal_17 | rh.lateralorbitofrontal_8 | lh.inferiortemporal_11 | lh.parsorbitalis_3 | lh.parahippocampal_3 | rh.postcentral_17 | rh.middletemporal_6 | rh.postcentral_22 |
| 971 | rh.fusiform_16 | rh.lateralorbitofrontal_9 | rh.parsorbitalis_3 | lh.lateralorbitofrontal_11 | rh.middletemporal_17 | rh.lingual_1 | rh.postcentral_17 | rh.parahippocampal_4 |
| 972 | rh.temporalpole_3 | lh.lateralorbitofrontal_15 | lh.medialorbitofrontal_2 | rh.parsorbitalis_3 | rh.medialorbitofrontal_11 | rh.pericalcarine_6 | rh.lingual_1 | rh.superiortemporal_7 |
| 973 | lh.inferiortemporal_2 | rh.middletemporal_3 | rh.medialorbitofrontal_2 | rh.medialorbitofrontal_11 | lh.lateralorbitofrontal_4 | rh.parahippocampal_4 | lh.frontalpole_2 | rh.cuneus_2 |
| 974 | rh.parahippocampal_4 | rh.superiortemporal_10 | rh.superiortemporal_10 | rh.medialorbitofrontal_2 | rh.parahippocampal_4 | lh.superiortemporal_19 | rh.superiortemporal_13 | |
| 975 | rh.superiortemporal_1 | rh.lateralorbitofrontal_10 | rh.parsorbitalis_2 | rh.superiortemporal_10 | lh.lateralorbitofrontal_9 | rh.lingual_2 | lh.superiortemporal_19 | lh.superiortemporal_19 |
| 976 | lh.lateralorbitofrontal_9 | rh.parsorbitalis_2 | rh.rostralmiddlefrontal_5 | rh.superiortemporal_8 | rh.superiortemporal_14 | rh.middletemporal_1 | lh.medialorbitofrontal_8 | |
| 977 | lh.rostralmiddlefrontal_20 | rh.rostralmiddlefrontal_5 | rh.medialorbitofrontal_9 | rh.parsorbitalis_4 | rh.rostralmiddlefrontal_20 | rh.superiortemporal_4 | rh.lingual_11 | |
| 978 | lh.middletemporal_11 | rh.entorhinal_1 | rh.lateralorbitofrontal_7 | rh.superiortemporal_17 | lh.superiortemporal_17 | rh.inferiortemporal_4 | lh.lateralorbitofrontal_4 | |
| 979 | lh.parsorbitalis_2 | lh.superiortemporal_12 | lh.superiortemporal_6 | rh.parsorbitalis_2 | rh.parsorbitalis_2 | rh.medialorbitofrontal_9 | lh.inferiortemporal_13 | |
| 980 | rh.superiortemporal_19 | rh.middletemporal_17 | rh.fusiform_16 | lh.parsorbitalis_2 | rh.middletemporal_5 | lh.middletemporal_8 | lh.lateralorbitofrontal_9 | |
| 981 | rh.superiortemporal_22 | lh.parsorbitalis_2 | lh.superiortemporal_12 | lh.medialorbitofrontal_9 | lh.parsorbitalis_2 | lh.inferiortemporal_14 | lh.superiortemporal_21 | |
| 982 | rh.lateralorbitofrontal_11 | lh.rostralmiddlefrontal_20 | lh.superiortemporal_17 | lh.medialorbitofrontal_8 | lh.frontalpole_1 | rh.fusiform_6 | lh.parsorbitalis_2 | |
| 983 | Brain-Stem | rh.fusiform_16 | rh.superiortemporal_8 | rh.entorhinal_1 | rh.temporalpole_1 | rh.parahippocampal_5 | lh.rostralmiddlefrontal_20 | |
| 984 | lh.medialorbitofrontal_1 | rh.parsorbitalis_4 | rh.parsorbitalis_4 | rh.lateralorbitofrontal_14 | rh.superiortemporal_22 | rh.postcentral_22 | rh.middletemporal_1 | |
| 985 | lh.parahippocampal_3 | lh.medialorbitofrontal_9 | lh.medialorbitofrontal_8 | rh.medialorbitofrontal_5 | lh.middletemporal_11 | lh.inferiortemporal_13 | rh.frontalpole_1 | |
| 986 | rh.middletemporal_5 | lh.superiortemporal_6 | rh.middletemporal_17 | rh.rostralmiddlefrontal_20 | lh.parsorbitalis_1 | lh.lingual_5 | rh.middletemporal_17 | |
| 987 | rh.parsorbitalis_4 | rh.superiortemporal_1 | lh.rostralmiddlefrontal_20 | rh.middletemporal_17 | lh.lateralorbitofrontal_11 | lh.lateralorbitofrontal_9 | rh.superiortemporal_23 | |
| 988 | lh.medialorbitofrontal_9 | lh.middletemporal_11 | rh.temporalpole_1 | lh.lateralorbitofrontal_15 | rh.superiortemporal_19 | rh.parsorbitalis_2 | rh.medialorbitofrontal_9 | |
| 989 | lh.parsorbitalis_1 | rh.superiortemporal_17 | rh.middletemporal_3 | lh.lateralorbitofrontal_16 | lh.medialorbitofrontal_1 | lh.rostralmiddlefrontal_20 | rh.superiortemporal_22 | |
| 990 | lh.frontalpole_1 | rh.medialorbitofrontal_5 | lh.parsorbitalis_4 | rh.fusiform_16 | lh.frontalpole_2 | rh.lingual_11 | rh.inferiortemporal_4 | |
| 991 | rh.lateralorbitofrontal_10 | rh.medialorbitofrontal_2 | lh.lateralorbitofrontal_16 | rh.middletemporal_3 | rh.superiortemporal_17 | lh.superiortemporal_17 | lh.medialorbitofrontal_1 | |
| 992 | rh.medialorbitofrontal_5 | rh.lateralorbitofrontal_7 | lh.lateralorbitofrontal_15 | rh.rostralmiddlefrontal_5 | lh.temporalpole_1 | lh.superiortemporal_16 | lh.superiortemporal_17 | |
| 993 | lh.medialorbitofrontal_3 | rh.temporalpole_3 | lh.middletemporal_11 | lh.parsorbitalis_4 | lh.temporalpole_1 | rh.parsorbitalis_2 | | |
| 994 | rh.superiortemporal_6 | lh.lateralorbitofrontal_14 | rh.parsorbitalis_2 | rh.superiortemporal_1 | lh.medialorbitofrontal_9 | rh.lateralorbitofrontal_16 | rh.lateralorbitofrontal_11 | |
| 995 | lh.superiortemporal_12 | lh.superiortemporal_7 | rh.frontalpole_2 | lh.middletemporal_11 | rh.parsorbitalis_4 | rh.lateralorbitofrontal_11 | rh.superiortemporal_4 | |
| 996 | rh.middletemporal_7 | lh.parsorbitalis_4 | lh.lateralorbitofrontal_14 | rh.frontalpole_1 | lh.medialorbitofrontal_3 | rh.superiortemporal_19 | lh.medialorbitofrontal_9 | |
| 997 | rh.frontalpole_2 | lh.lateralorbitofrontal_16 | rh.superiortemporal_1 | lh.superiortemporal_15 | rh.medialorbitofrontal_5 | rh.medialorbitofrontal_7 | rh.parsorbitalis_4 | |
| 998 | lh.superiortemporal_14 | rh.frontalpole_1 | lh.frontalpole_2 | rh.temporalpole_1 | rh.frontalpole_1 | rh.medialorbitofrontal_2 | rh.superiortemporal_19 | |
| 999 | rh.superiortemporal_17 | lh.medialorbitofrontal_3 | rh.temporalpole_1 | rh.frontalpole_2 | rh.medialorbitofrontal_9 | rh.parsorbitalis_4 | rh.medialorbitofrontal_2 | |
| 1000 | rh.frontalpole_1 | lh.parsorbitalis_1 | rh.temporalpole_3 | lh.medialorbitofrontal_3 | rh.middletemporal_7 | lh.medialorbitofrontal_1 | lh.medialorbitofrontal_3 | |
| 1001 | rh.medialorbitofrontal_9 | rh.middletemporal_5 | rh.middletemporal_5 | rh.temporalpole_3 | rh.lateralorbitofrontal_10 | rh.medialorbitofrontal_5 | lh.middletemporal_11 | |
| 1002 | lh.parsorbitalis_4 | rh.temporalpole_1 | rh.frontalpole_1 | lh.parsorbitalis_1 | rh.frontalpole_2 | rh.lateralorbitofrontal_10 | rh.medialorbitofrontal_5 | |
| 1003 | rh.medialorbitofrontal_2 | lh.frontalpole_2 | lh.superiortemporal_7 | lh.frontalpole_2 | rh.medialorbitofrontal_2 | lh.medialorbitofrontal_3 | rh.middletemporal_7 | |
| 1004 | rh.temporalpole_1 | lh.frontalpole_1 | lh.medialorbitofrontal_3 | lh.temporalpole_1 | rh.middletemporal_1 | rh.middletemporal_17 | lh.lateralorbitofrontal_16 | |
| 1005 | lh.temporalpole_1 | lh.temporalpole_1 | lh.frontalpole_1 | rh.middletemporal_5 | lh.superiortemporal_6 | rh.middletemporal_7 | lh.superiortemporal_16 | |
| 1006 | rh.medialorbitofrontal_7 | rh.middletemporal_7 | rh.medialorbitofrontal_9 | lh.superiortemporal_7 | lh.superiortemporal_12 | lh.lateralorbitofrontal_14 | rh.medialorbitofrontal_7 | |
| 1007 | lh.lateralorbitofrontal_16 | lh.superiortemporal_14 | lh.parsorbitalis_1 | lh.frontalpole_1 | lh.parsorbitalis_4 | rh.medialorbitofrontal_9 | lh.lateralorbitofrontal_14 | |
| 1008 | rh.middletemporal_1 | rh.medialorbitofrontal_9 | lh.superiortemporal_15 | rh.medialorbitofrontal_9 | lh.lateralorbitofrontal_16 | lh.middletemporal_11 | rh.lateralorbitofrontal_10 | |
| 1009 | lh.lateralorbitofrontal_14 | lh.superiortemporal_15 | lh.superiortemporal_14 | rh.medialorbitofrontal_7 | rh.medialorbitofrontal_7 | lh.superiortemporal_6 | lh.parsorbitalis_4 | |
| 1010 | rh.lateralorbitofrontal_7 | rh.frontalpole_2 | rh.middletemporal_1 | lh.superiortemporal_14 | lh.lateralorbitofrontal_14 | lh.parsorbitalis_4 | lh.superiortemporal_6 | |

| | | | | | | | |
|---|---|---|---|---|---|---|---|
| 1011 | rh.frontalpole_2 | rh.medialorbitofrontal_7 | rh.middletemporal_7 | rh.superiortemporal_4 | rh.lateralorbitofrontal_7 | rh.lateralorbitofrontal_7 | lh.superiortemporal_15 |
| 1012 | lh.superiortemporal_15 | rh.middletemporal_1 | rh.medialorbitofrontal_7 | rh.middletemporal_7 | rh.superiortemporal_4 | lh.superiortemporal_15 | rh.superiortemporal_3 |
| 1013 | rh.superiortemporal_4 | rh.superiortemporal_4 | rh.superiortemporal_4 | rh.middletemporal_1 | lh.superiortemporal_15 | lh.superiortemporal_12 | rh.lateralorbitofrontal_7 |
| 1014 | rh.superiortemporal_3 | rh.superiortemporal_3 | rh.superiortemporal_3 | rh.superiortemporal_3 | rh.superiortemporal_3 | rh.superiortemporal_3 | lh.superiortemporal_12 |
| 1015 | lh.superiortemporal_1 | lh.superiortemporal_1 | lh.superiortemporal_1 | lh.superiortemporal_1 | lh.superiortemporal_1 | lh.superiortemporal_1 | lh.superiortemporal_1 |



\end{document}